\documentclass[ASNA,onecolumn]{USGpersonal} 
\usepackage{anyfontsize} %

\usepackage{colortbl}   % For coloring table lines
\usepackage{xcolor}     % For color definitions

\usepackage{steinmetz}
\startpage{1}
\begin{document}
\title{Inferential applications of the moments of the logit-normal distribution}
%\subtitle{Subtitle}
\transtitle{Inferential applications of the moments of the logit-normal distribution}
\subtranstitle{trans-subtitle}
\author[1]{John B. Holmes} 
\author[1]{Ness Arps}
\author[1]{Marco Reale}

\authormark{Holmes, Arps and Reale}
\titlemark{Inferential applications of the moments of the logit-normal distribution5}

\address[1]{\orgdiv{Department of Mathematics and Statistics, }\orgname{University of Canterbury, }%
	\orgaddress{\state{Canterbury, }\country{New Zealand}}}

\corres{John Holmes  (email: \email{john.holmes@canterbury.ac.nz}, orcid: \url{https://orcid.org/0009-0001-6687-7505})}

%\editor{\textbf{Academic Editor:}  ~|~ \textbf{Guest Editor:} }

\presentaddress{}

%\fundingInfo{None}


\transkeywords{Logit-normal distributions | Moment estimation |  Expectation propagation}

\abstract[ABSTRACT]{Despite the implicit appearance of logit-normal random variables in many inferential problems, the logit-normal distribution is poorly studied. Most frustratingly, no default method exists for finding logit-normal moments, which are often assumed analytically unknown. In this paper, we introduce a method for estimating logit-normal moments of any positive integer order, based on approximating the logistic function. We will show our method is highly accurate up to the $8^\text{th}$ moment, avoids the numerical instability observed with Mordell integral based approximations of the first moment \cite{Johnson1949, Holmes2022a}, and is faster than numerical integration in R. Focusing on two inferential applications, we will show our approximation methods are sufficiently accurate to enable faster implementation of Expectation Propagation \cite{Minka2001, Minka2001b} for logistic regression, but is not general enough to directly evaluate the logistic normal integral that appears in some logistic mixed models \cite{Crouch1990}.  }

 \transabstract[transABSTRACT]{Despite the implicit appearance of logit-normal random variables in many inferential problems, the logit-normal distribution is poorly studied. Most frustratingly, no default method exists for finding logit-normal moments, which are often assumed analytically unknown. In this paper, we introduce a method for estimating logit-normal moments of any positive integer order, based on approximating the logistic function. We will show our method is highly accurate up to the $8^\text{th}$ moment, avoids the numerical instability observed with Mordell integral based approximations of the first moment \cite{Johnson1949, Holmes2022a}, and is faster than numerical integration in R. Focusing on two inferential applications, we will show our approximation methods are sufficiently accurate to enable faster implementation of Expectation Propagation \cite{Minka2001, Minka2001b} for logistic regression, but is not general enough to directly evaluate the logistic normal integral that appears in some logistic mixed models \cite{Crouch1990}.}

%\abbr{5-FU, 5-fluorouracil; CFD, computational fluid dynamics; CH, channel; EFS, event-free survival; GBM, glioblastoma multiforme; OS, overall survival; PFS, progression-freesurvival; SD, standard deviation.}

\contributed{John Holmes proposed the study and potential applications. Ness Arps was involved with parameter testing. Marco Reale was involved in manuscript refinement.}

%\dedicated{Dedicated to Srivivasa Ramanujano on the occasion of his 125th birth anniversary.}

\copyright{This is an open access article under the terms of the \href{Creative Commons Attribution-NonCommercial}{Creative Commons Attribution-NonCommercial} License, which permits use, distribution and reproduction in any medium, provided the
original~work~is~properly cited and is not used for commercial purposes.
\\[5pt]
  ©  2024 The Author(s) \textit{AIChE Journal} published by Wiley Periodicals LLC on behalf of American Institute of Chemical Engineers.}

%\openaccessstatement

\maketitle
%\afterpage{\aftergroup\restoregeometry}

%\onecolumn
\section{Introduction}\label{Intro}

In Statistics, the logistic transformation is one of the most common used, foremost in the context of logistic regression for binary, binomial and multinomial data, but also implicitly in Bayesian versions of negative binomial regression. In this paper, we focus on problems where the inverse logit function has been applied to a normal random variable $X$, with realisation $x \sim \mathcal{N}(\mu,\sigma^2)$ and hence the random variable $P := (1+e^{-X})^{-1}$ is distributed logit-normal with parameters $\mu, \sigma^2$. In frequentist inference, such random variables, including powers of, appear in logistic random intercept models when the link function $\boldsymbol{\theta}_i$ in cluster $i$ has the form $\boldsymbol{\theta}_i ={\bf X}_i\boldsymbol{\beta} + {\bf u}_i, {\bf u} \sim \mathcal{N}(0,\sigma^2_u{\bf I})$. In this setting, logit-normal random variables appear in the marginal likelihood of $\boldsymbol{\beta}$ \cite{Stiratelli1984, Crouch1990}. Additionally, the (marginal) response probability in a frequentist treatment of any logistic mixed model is a mean of a logit-normal random variable. In Bayesian inference, logit-normal random variables, including powers of, appear implicitly in Variational Bayes \cite{Jaakkola1997a, Jaakkola1997, Jordan1999} and Expectation Propagation \cite{Minka2001, Minka2001b} algorithms used to determine parameters of approximate posteriors in logistic and negative binomial regression.  

Despite the need to find a collection of logit-normal moments in many inferential problems, such moments are poorly understood, and assumed analytically unknown. As a result, literature on the logit-normal distribution that discusses moments splits into disconnected strands. One strand gives up on attempting to find analytic formulae for any logit-normal moments, instead promoting the use of standard numerical integration techniques \cite{Lad2008, Wutzler2017}. The second, usually buried in textbooks, are approximations of the first moment based on similarity of the logit and probit link or between the normal and logistic distributions \cite{Demidenko2013, Nelder1989book}. The last, pioneered in \citet{Johnson1949, Johnson1949b} and updated in \citet{Holmes2022a}, starts from the expected value of a logit-normal random variable, $E(P)$, being a special case of the Mordell integral \cite{Mordell1933} from complex analysis, but restricted to the real numbers. While the Mordell integral based representation of $E(P)$ given in \citet{Holmes2022a} can be combined with other properties of the logit-normal distribution to determine higher order moments, in practice Mordell integral based estimates of $E(P)$, and thus higher moments, become numerically unstable when $|\mu/\sigma^2| \rightarrow \infty$. Moreover, the repeated application of the quotient rule from Calculus needed to find higher order moments from the Mordell integral representation of $E(P)$ rapidly becomes algebraically very complicated.   

In this paper, we change tack from \citet{Holmes2022a}, and no longer attempt to estimate moments starting from the Mordell integral representation of $E(P)$. Instead, we focus on constructing highly accurate approximations of the logistic function and its derivatives using an basis of exponential functions. This allows us to not just determine an approximation for $E(P)$ but also easily develop a proposition to find the $k^\text{th}$ non-central moment, $k \in \mathbb{Z}^+$.  We will show our approximations, unlike Mordell integral representations,  are numerically stable when $|\mu/\sigma^2| \rightarrow \infty$. We will show our method is faster to compute than standard numerical integration in the programming language R . However, we will also show the number of terms needed for accurate approximation increases with moment order. Unfortunately, in practice the number of moments our method can calculate has an upper bound  around 8. This limits the utility of our method for evaluation of the logistic-normal integral. Nevertheless, getting to the $8^\text{th}$ moment is more than sufficiently high, which we will illustrate with an example, for performing expectation propagation in logistic regression. 

\section{Inferential problems that require logit-normal moments}

To justify our interest in the logit-normal distribution, we will provide two illustrative inferential problems, one frequentist, the other Bayesian, where a collection of logit-normal moments are needed in order to estimate parameters. First though, we need to define the logit-normal distribution and some of its properties.

\subsection{The logit-normal distribution and some of its properties}\label{sec2}
A logit-normal distributed random variable $P$ with realisation $p \sim \text{logit}N(\mu,\sigma^2)$ has probability density (pdf),
\begin{equation} f(p|\mu,\sigma^2) = \frac{1}{\sigma p(1-p)\sqrt{2\pi}}e^{-\frac{(\text{logit}(p)-\mu)^2}{2\sigma^2}}, \quad p \in (0,1), \text{logit}(p) = \log\bigg(\frac{p}{1-p}\bigg) \label{eq1.o} \end{equation} 

To determine the moments $k = 1, 2, \ldots$, of $P$ we will work with the following representations of the $k^\text{th}$ moment,
\begin{eqnarray} E(P^k) &=& 
	\frac{1}{\sigma\sqrt{2\pi}}\int_{-\infty}^{\infty} (1+e^{-x})^{-r}e^{-\frac{z^2}{2}} e^{-\frac{(x-\mu)^2}{2\sigma^2}} {\rm d} x 
	= \frac{1}{\sqrt{2\pi}}\int_{-\infty}^{\infty} (1+e^{-\mu-\sigma z})^{-r}e^{-\frac{z^2}{2}} {\rm d}z, \label{eq4} \end{eqnarray}
where $x \sim N(\mu,\sigma^2), z \sim N(0,1)$. However in practice, to find higher order moments of $P$, $E(P^2), E(P^3), \ldots$, it is easier to use the recurrence relation given in \citet{Holmes2022a}. This is found by taking the derivative of $E(P^k)$ with respect to $\mu$, where $E(P^k)$ is in the form of (\ref{eq4}),  
\begin{eqnarray}
	\frac{{\rm d}E(P^k)}{{\rm d}\mu} &=& \frac{1}{\sqrt{2\pi}}\int_{-\infty}^\infty \frac{ {\rm d} (1+e^{-\mu - \sigma z})^{-k} }{{\rm d}\mu} e^{-\frac{z^2}{2}} {\rm d}z \nonumber \\
	&=& \frac{k}{\sqrt{2\pi}}\int_{-\infty}^\infty \{(1+e^{-\mu - \sigma z})^{-k}- (1+e^{-\mu - \sigma z})^{-(k+1)}\}e^{-\frac{z^2}{2}}  {\rm d}z \label{eq.der} \nonumber \\
	&=& k\{E(P^k) - E(P^{k+1})\}, \label{eq.der1}  \\
	E(P^{k+1})  &=& E(P^k) - \frac{1}{k}	\frac{{\rm d}E(P^k)}{{\rm d}\mu}. \label{eq.der2}		
\end{eqnarray}
Thus the main challenge for determining logit-normal moments is evaluating the first moment or expected value in a form such that applying the recurrence relation in (\ref{eq.der2}) is straightforward for any choice of $k$.

\subsection{A frequentist problem: Some cases of the logistic mixed model}\label{GLMM}

In order to perform inference in frequentist treatments of a logistic mixed model, the marginal likelihood, $L(\boldsymbol{\beta}, \boldsymbol{\Sigma})$ is usually needed. This likelihood, where $\boldsymbol{\beta}$ is a length $p$ vector of fixed effects, and $\boldsymbol{\Sigma}$ is a $q \times q$ positive definite matrix of (co)-variance components is,
\begin{equation} L(\boldsymbol{\beta}, \boldsymbol{\Sigma}) \propto  \int \prod_{i=1}^n {\bf p}_i^{{\bf y}_i}(1-{\bf p}_i)^{1-{\bf y}_i} |{\boldsymbol{\Sigma}}|^{-1/2}e^{-{\bf u}'{\boldsymbol{\Sigma}}^{-1}{\bf u}/2}{\rm d}{\bf u}\label{eqGLMM} \end{equation}  
In (\ref{eqGLMM}), ${\bf p}_i = (1+e^{-{\bf x}_i'\boldsymbol{\beta} - {\bf z}_i'{\bf u}})^{-1}$ \cite{Searle1992, Stiratelli1984} are the elements of the $n$ length vector of success probabilities, where ${\bf x}_i, {\bf z}_i$ correspond to the covariate values for observation $i$ while ${\bf u}$ is a $q$ length vector of random effects. In the special case where ${\bf u}$ corresponds to a random intercept, and all observations; $j =1, \ldots {\bf n}_i$, within each random effect level share the same fixed effect covariate vector, \cite{Pierce1975, Williams1982} then the marginal likelihood simplifies to,
\begin{eqnarray} L(\boldsymbol{\beta}, \boldsymbol{\Sigma}) &=& \prod_{i=1}^q L_i(\boldsymbol{\beta}, \boldsymbol{\Sigma}), \quad \text{where }
	L_i(\boldsymbol{\beta}, \boldsymbol{\Sigma}) = \int {{{\bf n}_i}\choose{ {\bf y}_i}} {\bf p}_i^{{\bf y}_i}(1-{\bf p}_i)^{n_i-{\bf y}_i} \frac{1}{\sqrt{2\pi\sigma^2_u}}e^{-\frac{ {\bf u}_i^2}{2\sigma^2} }{\rm d}{\bf u}. \label{eqGLMM2} \end{eqnarray}  
By the binomial theorem, ${\bf p}_i^{{\bf y}_i}(1-{\bf p}_i)^{n_i-{\bf y}_i}$ can be rewritten as \begin{equation}{\bf p}_i^{{\bf y}_i} \sum_{k=0}^{ {\bf n}_i-{\bf y}_i} {{{\bf n}_i-{\bf y}_i}\choose{k}} (-{\bf p}_i)^k = \sum_{k={\bf y}_i}^{{\bf n}_i} {{{\bf n}_i-{\bf y}_i}\choose{k -{\bf y}_i}} (-1)^{k-{\bf y}_i}{\bf p}_i^k. \label{pbinom}\end{equation}
Substituting (\ref{pbinom}) into (\ref{eqGLMM2}) implies that $L_i(\boldsymbol{\beta}, \boldsymbol{\Sigma})$ as defined in (\ref{eqGLMM2}) corresponds to a linear combination of logit-normal moments from order ${\bf y}_i$ to ${\bf n}_i$ with location parameter ${\bf x}_i\boldsymbol{\beta}$ and scale parameter $\sigma_u$. As the value of ${\bf n}_i$ is arbitrary, this example provides our motivation to find a formula for a logit-normal moment of any positive integer order. 

Note the model we are describing is also known as an observation level random effect model, often seen in ecology \cite{Harrison2015}.

\subsection{An approximate Bayesian problem: Expectation Propagation for logistic regression}\label{EP}

Expectation Propagation (EP) is a method of approximate Bayesian inference introduced by Minka \cite{Minka2001, Minka2001b}. In EP, it is assumed the posterior distribution $p(\boldsymbol{\theta}|{\bf y})$ can be factorised as,
\[ p(\boldsymbol{\theta}|y_1, \ldots y_n) = \prod_{i=0}^n f_i(\boldsymbol{\theta}),  \]

where $i$ represents a data-point. As the posterior is proportional, up to a constant, to likelihood $\times$ prior, the natural choice for $f_i(\boldsymbol{\theta})$ is that when $i=0$, $f_i(\boldsymbol{\theta})$ will be proportional to the prior distribution $p(\boldsymbol{\theta})$, while for $i = 1, 2, \ldots n$, $f_i(\boldsymbol{\theta})$ will be proportional to the likelihood of data-point $i$, $p(y_i|\boldsymbol{\theta})$. The true posterior $ p(\boldsymbol{\theta}|y_1, \ldots y_n)$ is then approximated with the density $g(\boldsymbol{\theta})$ such that  $g(\boldsymbol{\theta}) = \prod_{i=0}^ng_i(\boldsymbol{\theta})$. The parameters of the approximate posterior $g(\boldsymbol{\theta})$ are found by minimising the following Kullback-Leibler divergence $D_{KL}\{cg_{-i}(\boldsymbol{\theta})f_i(\boldsymbol{\theta})||g(\boldsymbol{\theta})\} =$
\begin{equation} \int cg_{-i}(\boldsymbol{\theta})f_i(\boldsymbol{\theta}) \{\log(cg_{-i}(\boldsymbol{\theta})f_i(\boldsymbol{\theta}))-\log(g(\boldsymbol{\theta}))\} {\rm d}\boldsymbol{\theta}, \end{equation}
where $c$ is a normalising constant. 

In a generalised linear model (GLM) with fixed dispersion parameter, such as logistic regression, the parameters $\boldsymbol{\theta}$ are simply the vector of regression coefficients $\boldsymbol{\beta}$. In this case, the obvious choice is to assume the approximate posterior is normal $g(\boldsymbol{\beta}) := \mathcal{N}(\boldsymbol{\beta}|\hat{\boldsymbol{\beta}}, \hat{\boldsymbol{\Sigma}})$. In the resulting moment matching algorithm used to determine $\hat{\boldsymbol{\beta}}, \hat{\boldsymbol{\Sigma}}$ (see \citet{Gelman2014} for details), the following univariate  integrals must be evaluated,
\begin{equation}
	E_k = \int \boldsymbol{\eta}_i^kg_{-i}(\boldsymbol{\eta}_i)f_i(\boldsymbol{\eta}_i) {\rm d}\boldsymbol{\eta}_i, \quad k = 0, 1, 2, \label{eq.EP}
\end{equation}  
where if $i = 1, \ldots, n$, $\boldsymbol{\eta}_i = {\bf x}_i'\boldsymbol{\beta}$, $g_{-i}(\boldsymbol{\eta}_i)$ is a normal density with mean $\mu_i={\bf x}_i'\hat{\boldsymbol{\beta}}$ and variance $\sigma^2_i={\bf x}_i'\hat{\boldsymbol{\Sigma}}{\bf x}_i$ and $f_i(\boldsymbol{\eta}_i)$ is the likelihood term associated with observation $i$. In most cases, including logistic regression, $E_k$ is evaluated using numerical integration methods. This is often emphasised for logistic regression,  for example in \citet{Chopin2017, Kim2018} and \citet{Hall2019}, because the alternative of probit regression is one case where $E_k$ is analytically known. 

To demonstrate that in binary logistic regression, (\ref{eq.EP}) is equivalent to a collection of logit-normal moments, we need to re-write (\ref{eq.EP}) using repeated application of Stein's Lemma \cite{Stein1972}, which states:  

\begin{result}\label{hxExpect0}
	If $X \sim N(\mu,\sigma^2)$, and $h$ a differentiable function applied to $X$, then 
	\begin{eqnarray} \sigma^2E(h'(X)) &=& E((X-\mu)h(X)),
	\end{eqnarray}	
	where $h'(x)$ is the derivative of $h$ with respect to $x$.
\end{result}	

This leads to an alternative set of univariate integrals for solve in the EP moment matching, as stated below.
\begin{proposition}\label{prop1}
	The moment matching step in EP for a GLM problem, when $g(\boldsymbol{\beta})$ is assumed normal, and $\boldsymbol{\eta}_i = {\bf x}_i'\boldsymbol{\beta}$ can be alternatively solved by finding the integrals
	\begin{eqnarray}
		I_k &=& \int g_{-i}(\boldsymbol{\eta}_i)f_i^{(k)}(\boldsymbol{\eta}_i) {\rm d}\boldsymbol{\eta}_i  \quad k = 0, 1, 2 \label{EPeqnew}
	\end{eqnarray}
	where $f_i^{(k)}(\boldsymbol{\eta}_i)$ is the $k^\text{th}$ derivative of $f_i(\boldsymbol{\eta}_i)$. The equations needed to back-transform from (\ref{EPeqnew}) to (\ref{eq.EP}) are $E_0 = I_0, E_1 = \mu_iI_0 + \sigma^2_iI_1, E_2 = (\mu_i^2+\sigma^2_i)I_0 + 2\mu_i\sigma^2_iI_1 + \sigma^4_iI_2$. 
	
	A derivation can be found in the appendix. 
\end{proposition}

For a binary logistic regression, $f_i(\boldsymbol{\eta}_i)$ corresponds to ${\bf p}_i = (1+ e^{-{\bf x}_i'\boldsymbol{\beta}})^{-1}$, a logit-normal random variable with location parameter ${\bf x}_i'\hat{\boldsymbol{\beta}}$ and scale parameter $\sqrt{{\bf x}_i'\hat{\boldsymbol{\Sigma}}{\bf x}_i}$ when $y_i = 1$. When $y_i = 0$,  $f_i(\boldsymbol{\eta}_i)= 1 - {\bf p}_i$, i.e. 1 minus a logit-normal variable. Lastly, we need an additional result, used to justify the reparametrisation trick used in machine learning \cite{Kingma2013}, 

\begin{result}\label{hxExpect}
	If $X \sim N(\mu,\sigma^2)$, and $h$ a differentiable function applied to $X$, then 
	\begin{eqnarray} \frac{{\rm d}E(h(X))}{{\rm d}\mu} &=& E(h'(X)),
	\end{eqnarray}	
	where $h'(x)$ is the derivative of $h$ with respect to $x$.
\end{result}	

Using Result \ref{hxExpect}, we can re-write $I_k$ as $\frac{\rm d}{{\rm d} \mu^k}\big\{ \int g_{-i}(\boldsymbol{\eta}_i)f_i(\boldsymbol{\eta}_i) {\rm d}\boldsymbol{\eta}_i \big\}; k = 1, 2$. For a binary logistic regression, this means determining the approximate EP posterior requires finding a set of $n$ logit-normal first moments, $E({\bf p}_i)$ along with its two derivatives w.r.t $\mu _i = {\bf x}_i'\hat{\boldsymbol{\beta}}$. By the recurrence relation given in (\ref{eq.der1}), this is equivalent to finding the first three moments of $n$ logit normal random variables. Hence, even if analytic evaluation of logit-normal moments of arbitrary order proves impossible, the case of expectation propagation demonstrates determining formulae for a collection of low order logit-normal moments would remain valuable for inferential problems.

\subsection{Numerical implementation issues with previous analytic attempts at finding $E(P)$}

In theory, \citet{Holmes2022a} provided a somewhat user friendly analytic expression for $E(P)$. This states the first moment, or expected value is:
\begin{equation} E(P) = \frac{1}{2} + \frac{\sum_{n=1}^\infty e^{-{\sigma^2n^2}/{2}}\sinh(n\mu)\tanh(\frac{n\sigma^2}{2}) + {2\pi}/{\sigma^2}\sum_{n=1}^{\infty}\frac{e^{-{(2n-1)^2\pi^2}/{2\sigma^2}}\sin\{(2n-1)\pi\mu/{\sigma^2}\}}{\sinh\{{(2n-1)\pi^2}/{\sigma^2}\}}  }{1+2\sum_{n=1}^\infty e^{-{\sigma^2n^2}/{2}}\cosh(n\mu)}. \label{EP.new} \end{equation}

In practice though, the inclusion of $\cosh(n\mu)$ and $\sinh(n\mu)$ terms in (\ref{EP.new}) makes this expression for $E(P)$, and implicitly for all higher order moments unusable in some cases. As $n \rightarrow \infty$, these terms, unless $\mu =0$, will blow out to $\pm \infty$ as  $\cosh(n\mu) \rightarrow \infty \forall \mu \in \mathbb{R}\backslash \{0\}$ and $1$ if $\mu = 0$ and $\sinh(n\mu) \rightarrow \infty \text{ if } \mu > 0, -\infty \text{ if } \mu < 0$ and $0$ if $\mu = 0$. While the $\cosh(n\mu), \sinh(n\mu)$ terms in (\ref{EP.new}) are multiplied by $e^{-{\sigma^2n^2}/{2}}$, which goes to zero as $n \rightarrow \infty$ as $\sigma^2 > 0$, if $|\mu/\sigma^2|$ is large $|\cosh(n\mu)|,|\sinh(n\mu)| \rightarrow \infty$ at a lower value of $n$ than $e^{-{\sigma^2n^2}/{2}} \rightarrow 0$, causing numerical instability. 

For the calculation of moments, \citet{Holmes2022a} developed a truncated version of (\ref{EP.new}) based on the decay of the component series which defined the maximum index $N$ in the truncation as $N = \max(\text{ceiling}(|\mu/\sigma^2 -\sqrt{2\log(1/\epsilon)}/\sigma|),\text{ceiling}(|\mu/\sigma^2 +\sqrt{2\log(1/\epsilon)}/\sigma|), \text{ceiling}(1/2 + \sigma\sqrt{\log(2/\epsilon)/2}/\pi))$, with $\epsilon = 1 \times 10^{-4}$. To illustrate why this truncation method will break down because of the inclusion of $\cosh(n\mu), \sinh(n\mu)$ terms, consider the case where $\mu =1 \Leftrightarrow \text{Median}(P) = 0.7311$ to 4 d.p., $ \forall \sigma^2 \in \mathbb{R}^+$. We wish to calculate $E(P)$ for $\sigma^2= (1, 0.1, 0.01, 0.001) \Leftrightarrow \sigma = (1, 0.3162, 0.1, 0.0316)$ and will use both numerical integration and the truncated sum version of (\ref{EP.new}). Results of the calculation in R , including the largest index used in the truncated sum are given in Table \ref{Table1.EP}. 
\begin{center}
	\begin{table}[ht!]
		\caption{Performance of truncated sum approach of \citet{Holmes2022a} at calculating $E(P)$ when $\mu = 1$}\label{Table1.EP}
		\begin{tabular*}{\textwidth}{@{\extracolsep\fill}lccccccr@{\extracolsep\fill}}%
			\toprule
			$\sigma$	& \multicolumn{3}{c}{Numerical Integration} & \multicolumn{3}{c}{Truncated sum} & $N$ \\ % Table header row
			\midrule
			1 & \multicolumn{3}{c}{0.6967} & \multicolumn{3}{c}{0.6967} & 6 \\ 
			0.3162 & \multicolumn{3}{c}{0.7267} & \multicolumn{3}{c}{0.7267} & 24 \\ 
			0.1 & \multicolumn{3}{c}{0.7306} & \multicolumn{3}{c}{0.7306} & 143 \\
			0.0316 & \multicolumn{3}{c}{0.7310} & \multicolumn{3}{c}{NaN} & 1136 \\ 
			\midrule
		\multicolumn{8}{c}{Performance of Truncated sum at calculating $E(P)$ for varying $N$ when $\mu =1, \sigma^2 =0.001 \Rightarrow E(P) = 0.7310$.} \\
\midrule
$N = 400$ & $N = 450$ & $N = 500$ & $N =550$ & $N = 600$ & $N = 650$ & $N = 700$ & $N =750$ \\ 	
\multicolumn{1}{c}{0.5984} & \multicolumn{1}{c}{0.6103} & \multicolumn{1}{c}{0.6221} & \multicolumn{1}{c}{0.6337} & \multicolumn{1}{c}{0.6452} & \multicolumn{1}{c}{0.6565} & \multicolumn{1}{c}{0.6676} &    \multicolumn{1}{r}{NaN} \\
			\bottomrule
		\end{tabular*}
	\end{table}
\end{center}
Importantly when $\mu=1, \sigma = 0.0316$, where the highest index in the truncated sum is $N =1136$, the representation of $E(P)$ given in (\ref{EP.new}) could not give a meaningful evaluation. Changing the truncation value $N$ used in the calculation of $E(P)$, as shown in the later rows of Table \ref{Table1.EP}, we found the  breakdown value of $N$ occurred before an accurate estimate of $E(P)$ was found.

We could try to get around this by working, when $\mu > 0$, with 
\begin{eqnarray}
	e^{-\frac{\sigma^2n^2}{2}+\log(\sinh(n\mu))+\log(\tanh(\frac{n\sigma^2}{2}))} &=& e^{-\frac{\sigma^2n^2}{2}-\log(2) +n\mu+\log(1-e^{-2n\mu})+\log(\tanh(\frac{n\sigma^2}{2}))}	\nonumber \\
	e^{-\frac{\sigma^2n^2}{2}+\log(\cosh(n\mu))} &=& e^{-\frac{\sigma^2n^2}{2}-\log(2)+n\mu+\log(1+e^{-2n\mu})}\nonumber 
\end{eqnarray}  
and $e^{-\frac{\sigma^2n^2}{2}-\log(2) -n\mu+\log(1-e^{2n\mu})+\log(\tanh(\frac{n\sigma^2}{2}))}, e^{-\frac{\sigma^2n^2}{2}-\log(2)-n\mu+\log(1+e^{2n\mu})}$ when $\mu < 0$. However, the high value of $N$ needed when $\sigma \rightarrow 0$ indicates (\ref{EP.new}) would no longer always out-perform numerical integration on speed. Hence we must conclude the Mordell integral based representation of $E(P)$ in (\ref{EP.new}) is unsuitable for use in inferential problems, such as those discussed in sections \ref{GLMM} and especially \ref{EP}. Hence we need to develop an alternative method for evaluating logit-normal moments.  

\section{A new approach to finding moments of a logit-normal random variable}

Our new approach for determining the moments of the logit-normal distribution is motivated by \citet{Komodromos2024}. This work included an approximation the evidence lower bound (ELBO), the expected value of the joint distribution with respect to the approximate posterior for a variational Bayes implementation of logistic regression. For Bernoulli/binomial likelihoods with a normal distribution for the approximate posterior this bound contains 
\begin{eqnarray} E(S) = E(\log(1+e^X)),\label{softplus.new} \end{eqnarray}
where $X\sim \mathcal{N}(\mu,\sigma^2)$, and $S = \log(1+e^X)$ is also known as the softplus function \cite{Dugas2000}.

Our interest in the softplus function, $s(x)= \log(1+e^x)$, is that its first derivative is the logistic function,
\begin{equation} s'(x) = {e^x}/{(1+e^x)} = (1+e^{-x})^{-1} = p(x). \end{equation}
If $x$ is a realisation of a normal random variable with mean $\mu$ and variance $\sigma^2$, then $p(x)$ is a realisation of a logit-normal random variable. This means that $E(p(X)) = E(P)$, the expected value of a logit normal random variable with location and scale parameters  $\mu$ and $\sigma$ respectively.

\subsection{Re-writing $s(x), p(x)$ into an expectation taking friendly form} 

The properties of the softplus function we need are those used to approximate $E[s(X)]$ in \citet{Komodromos2024}. First, we split $s(x) = \log(1 +e^x)$ into even and odd functions. For this, start by rewriting $s(x)$ as
\begin{eqnarray}
	s(x) = \log(1+e^x) = \log[(e^{-x} + 1)e^x] = \log(1+e^{-x}) + \log(e^x) = x + \log(1 + e^{-x}), \label{otherSPdef}
\end{eqnarray}
then express $s(x)$ as a piece-wise function using $s(x) = \log(1+e^x)$ when $x < 0$ and (\ref{otherSPdef}) when $x \geq 0$, i.e. 
\begin{eqnarray}
	s(x) = \log(1+e^x) &=& \begin{cases} \log(1+e^x) \quad \text{\hspace{7.5mm}If $x < 0$} \\
		x+ \log(1+e^{-x}) \quad \text{If $x \geq 0$} \\
	\end{cases}
\end{eqnarray}
or more compactly 
\begin{eqnarray}
	s(x) = x\mathbb{1}_{x > 0} + \log(1+e^{-|x|}), \label{rampSP}
\end{eqnarray}
where $x$ is an odd, and $\log(1+e^{-|x|})$ is an even function. Taking the derivative of $s(x)$ as written in (\ref{rampSP}), we get an alternative representation of the logistic function, 
\begin{eqnarray}
	p(x) = (1+e^{-x})^{-1} = \mathbb{1}_{x>0} + (-1)^{\mathbb{1}_{x>0}}(1+e^{-|x|})^{-1}. \label{rampimpliedL}
\end{eqnarray}

To complete the transformation of $s(x), p(x)$ into a form where we can take approximate expectations analytically, we used the Maclaurin series of $\log(1+y)$ \cite{Komodromos2024}, 
\begin{eqnarray} \log(1+y) = \sum_{i=1}^\infty \frac{(-1)^{i-1}y^i}{i}. \label{MSlog1y} \end{eqnarray}

The Maclaurin series for $\log(1+y)$ is convergent when $|y| \leq 1$, and as we assume $y = e^{-|x|}, x \in \mathbb{R}$, (\ref{MSlog1y}) is a convergent series  $\forall x$. Hence we will work with the softplus and logistic functions in the following form,
\begin{eqnarray}
	s(x) = \log(1+e^x) &\approx& x\mathbb{1}_{x > 0} + \sum_{i=1}^N \frac{(-1)^{i-1}e^{-i|x|}}{i}, \label{rampSPKomo} \\
	p(x) = (1+e^{-x})^{-1} &\approx& \mathbb{1}_{x > 0} + (-1)^{\mathbb{1}_{x>0}}\sum_{i=1}^N (-1)^{i-1}e^{-i|x|}.\label{ramplogisticKomo}
\end{eqnarray}

This will be our starting point for determining the moments of the logit-normal distribution. 

\subsubsection{Improving the approximation of $p(x)$}

We now will establish a naive copying of \citet{Komodromos2024}, i.e. using (\ref{ramplogisticKomo}), will not lead to a good approximation of $p(x)$ for all $x \in \mathbb{R}$. Consider $p(x)$ when $x = 0$, which we know is equal to 0.5. In this case the approximation of $p(x)$ given in (\ref{ramplogisticKomo}) can be simplified as 
\begin{eqnarray}
	p(x) \approx  (-1)^{\mathbb{1}_{x>0}}\sum_{i=1}^N (-1)^{i-1},\label{ramplogisticKomo2}
\end{eqnarray}
where $\sum_{i=1}^\infty (-1)^{i-1}$ is a divergent series whose partial sums oscillate between $1$ if $N$ is odd, and 0 if $N$ is even. More generally, (\ref{ramplogisticKomo}) contains a modified finite geometric series truncated at $N-1$ with $r = -e^{-|x|}$,
\begin{eqnarray}
	\sum_{i=1}^N (-1)^{i-1}e^{-i|x|} = e^{-|x|}\sum_{i=0}^{N-1} (-e^{-|x|})^i = -r \sum_{i=0}^{N-1}r^i.\label{ramplogisticKomo3}
\end{eqnarray}
To determine $E(P)$, we want to estimate ${-r}/{(1-r)}$, but $-r \sum_{i=0}^{N-1}r^i$ estimates ${-r(1-r^N)}/{(1-r)}$. The resulting error in the approximation of $p(x)$ in (\ref{ramplogisticKomo}) is ${(-1)^{N}e^{-(N+1)|x|}}/{(1+e^{-|x|})}$. The absolute error function of $p(x)$, ${e^{-(N+1)|x|}}/{(1+e^{-|x|})}$, when approximated by (\ref{ramplogisticKomo}), is a product of two functions, $e^{-N|x|}$, ${e^{-|x|}}/{(1+e^{-|x|})}=(1+e^{|x|})^{-1}$, that are monotonely decreasing with respect to $|x|$. This means the absolute error in the approximation of $p(x)$ accumulates in the region near $x \rightarrow 0$ or equivalently $p(x) \rightarrow 0.5$. In turn, estimation of $E(P)$ with (\ref{ramplogisticKomo}) will be particularly hard when $\mu \approx 0$ and $\sigma^2 \rightarrow 0$. This makes the implied estimate of $E(P)$ we could construct by naive following of \citet{Komodromos2024}, like the Mordell integral representation of $E(P)$ in \citet{Holmes2022a}, a particular risk if used in parameter estimation algorithms for large datasets.     

To reduce the error in the approximation of $p(x)$ in the region $(-L,L)$, we propose using the piece-wise function, 
\begin{eqnarray}
	p(x) \approx \mathbb{1}_{x > 0} + (-1)^{\mathbb{1}_{x>0}}\sum_{i=1}^{N} ((-1)^{i-1}\mathbb{1}_{|x| \geq L}+ic_i\mathbb{1}_{|x| < L})e^{-i|x|}, \label{ramplogistic.H}
\end{eqnarray}
where if $|x| \geq L$, we use a Maclaurin series and if $|x| < L$, we use a Chebyshev polynomial interpolation \cite{Lanczos1952}. To calculate the coefficients of the Chebyshev interpolating polynomial, $c_i$, we will use the R package \texttt{pracma} \cite{Pracma2023}, but applied to the function $\log(1+y)$ on the interval $(e^{-L},1)$ where $y =e^x$, then take the derivative. By determining the interpolating polynomial for $\log(1+y)$, rather than $\frac{y}{1+y}$, we ensure the constant term $c_0$ disappears in the approximation of $p(x)$, leaving only $c_i$ associated with $e^{-i|x|}$. 

Details on how to choose optimal $L$ conditional on $N$ will be presented in section \ref{FindingL}. First though, we present formulae for approximating $E(P)$ as well as higher order moments.

\section{(Approximate) moments of a logit-normal random variable}

\subsection{The expected value of a logit-normal random variable}\label{FINDEP}

\begin{proposition}\label{EP.LN2new}
	The expectation of $P$, $E(P)$, where $P$ is distributed $\text{logit}N(\mu, \sigma^2)$ is
	\begin{eqnarray} E(P) &\approx&  \Phi(\mu/\sigma) + \sum_{i=1}^{N} (-1)^{i-1}\left\{e^{\mu i+\sigma^2i^2/2}\Phi\left(\frac{-L-\mu-\sigma^2 i}{\sigma}\right) -e^{-\mu i+\sigma^2i^2/2}\Phi\left(\frac{-L+\mu-\sigma^2 i}{\sigma}\right)  \right\} \nonumber \\
		&&+\sum_{i=1}^{N} ic_ie^{\mu i+\sigma^2i^2/2}\left\{\Phi\left(\frac{-\mu-\sigma^2 i}{\sigma}\right) - \Phi\left(\frac{-L-\mu-\sigma^2 i}{\sigma}\right)\right\}\nonumber \\
		&&-\sum_{i=1}^{N} ic_ie^{-\mu i+\sigma^2i^2/2}\left\{\Phi\left(\frac{\mu-\sigma^2 i}{\sigma}\right)-\Phi\left(\frac{-L+\mu-\sigma^2 i}{\sigma}\right) \right\}  \label{eq.EO}
	\end{eqnarray}	
	
	where $\Phi(\cdot)$ is the standard normal cumulative distribution function. A proof can be found in the appendix.
\end{proposition}

Unlike the Mordell integral based representation of $E(P)$ in (\ref{EP.new}), the expression for $E(P)$ given in Proposition \ref{EP.LN2new} will not run into numerical calculation issues. While $E(P)$ in (\ref{eq.EO}) contains terms, $e^{\pm \mu i+\sigma^2i^2/2}$, which $\rightarrow \infty$ as $i \uparrow$, these are multiplied by terms of the form $\Phi(a-i\sigma)$, with $a$ constant, which $\rightarrow 0$ as $i \uparrow$ at a faster rate than $e^{\pm \mu i+\sigma^2i^2/2}$ $\rightarrow \infty$. This means any numerical $\infty$ that appears in calculation is multiplied by an effective zero.

\subsection{Higher order moments}\label{FINDEPhigh}

Higher order moments of $P$ can be found as a function of $E(P)$ and its derivatives with respect to $\mu$, by applying the recursion formula given in (\ref{eq.der2}) repeatedly. 

By Result \ref{hxExpect}, we just need the derivatives of $p(x)$ from (\ref{ramplogistic.H}) and take expectations to approximate higher order moments. As $p(x)$ is a linear combination of exponentials, the derivatives of $p(x)$ are highly repetitive, with
\begin{eqnarray}
	\frac{{\rm d}^k p(x)}{{\rm d}x^k} \approx \begin{cases} \sum_{i=1}^{N} i^k((-1)^{i-1}\mathbb{1}_{|x| \geq L}+ic_i\mathbb{1}_{|x| < L})e^{-i|x|} \quad \text{\hspace{13 mm} If  $k$ is odd,} \\
		(-1)^{\mathbb{1}_{x>0}}\sum_{i=1}^{N} i^k((-1)^{i-1}\mathbb{1}_{|x| \geq L}+ic_i\mathbb{1}_{|x| < L})e^{-i|x|} \quad \text{If  $k$ is even.} \\
	\end{cases} \label{ramplogistic.diff}
\end{eqnarray}

To find $\frac{{\rm d}^kE(P)}{{\rm d}\mu^k}$, we just change the leading coefficients of all components except $\Phi(\mu/\sigma)$ in the approximation of $E(P)$ from (\ref{eq.EO}),
\begin{eqnarray} \frac{{\rm d}^kE(P)}{{\rm d}\mu^k} &=&  \sum_{i=1}^{N} i^k(-1)^{i-1}\left\{e^{\mu i+\sigma^2i^2/2}\Phi\left(\frac{-L-\mu-\sigma^2 i}{\sigma}\right) +(-1)^{k-1}e^{-\mu i+\sigma^2i^2/2}\Phi\left(\frac{-L+\mu-\sigma^2 i}{\sigma}\right)  \right\} \nonumber \\
	&&+\sum_{i=1}^{N} i^{k+1}c_ie^{\mu i+\sigma^2i^2/2}\left\{\Phi\left(\frac{-\mu-\sigma^2 i}{\sigma}\right) - \Phi\left(\frac{-L-\mu-\sigma^2 i}{\sigma}\right)\right\}\nonumber \\
	&&+(-1)^{k-1}\sum_{i=1}^{N} i^{k+1}c_ie^{-\mu i+\sigma^2i^2/2}\left\{\Phi\left(\frac{\mu-\sigma^2 i}{\sigma}\right)-\Phi\left(\frac{-L+\mu-\sigma^2 i}{\sigma}\right) \right\}.  \label{eq.EOdiff}
\end{eqnarray}

Substituting (\ref{eq.EOdiff}) sequentially into (\ref{eq.der2}), starting with $k=1$ to find the second moment, leads to the following proposition for an approximation of positive integer moments, 	

\begin{proposition}\label{EP.LN2high}
	The $k^\text{th}$, $k \geq 2$, integer moment of $P$, $E(P^k)$, where $P$ is distributed $\text{logit}N(\mu, \sigma^2)$ is	approximately 
	\begin{eqnarray} E(P^k) &\approx&  \Phi(\mu/\sigma) + \frac{1}{(k-1)!}\bigg(\sum_{i=1}^{N} (-1)^{i-1}\prod_{j=1}^{k-1}(j-i) e^{\mu i+\sigma^2i^2/2}\Phi\left(\frac{-L-\mu-\sigma^2 i}{\sigma}\right) \nonumber \\ 
		&&-\sum_{i=1}^{N} (-1)^{i-1}\prod_{j=1}^{k-1}(j+i)e^{-\mu i+\sigma^2i^2/2}\Phi\left(\frac{-L+\mu-\sigma^2 i}{\sigma}\right)  \nonumber \\
		&&+\sum_{i=1}^{N} ic_i\prod_{j=1}^{k-1}(j-i)e^{\mu i+\sigma^2i^2/2}\left\{\Phi\left(\frac{-\mu-\sigma^2 i}{\sigma}\right) - \Phi\left(\frac{-L-\mu-\sigma^2 i}{\sigma}\right)\right\}\nonumber \\
		&&-\sum_{i=1}^{N} ic_i\prod_{j=1}^{k-1}(j+i)e^{-\mu i+\sigma^2i^2/2}\left\{\Phi\left(\frac{\mu-\sigma^2 i}{\sigma}\right)-\Phi\left(\frac{-L+\mu-\sigma^2 i}{\sigma}\right) \right\}\bigg)  \label{eq.EOhigh}
	\end{eqnarray}	
	
	where $\Phi(\cdot)$ is the standard normal cumulative distribution function.
\end{proposition}	

A proof of Proposition \ref{EP.LN2high} can be found in the appendix.

\section{Determining optimal parameters for $N$ and $L$}\label{FindingL}

We have not yet established how accurate our approximations for $E(P), E(P^2), \ldots$, will be, or what optimal choices for $N, L$ are. While increasing $N$ should result in improved accuracy of approximation, we still need to find the optimal $L$ given $N$. To determine optimal $L$, we will work with the underlying functions $p(x), p'(x), \ldots$, as if we can guarantee accurate approximation of the underlying functions for all $x$, we can guarantee accuracy of moment estimates. 

To choose $N, L$, we consider accuracy of approximation up to the third derivative of $p(x)$. This implies we are aiming to achieve accurate approximation up to at least the fourth moment. As the approximate component of $p(x)$ in (\ref{ramplogistic.H}) is symmetric around zero, to determine optimal $L$, we only need to look at $x \in \mathbb{R}^+$ or $x \in \mathbb{R}^-$. As at $|x| = 5, N=5$, the absolute error of the Maclaurin series component of $p(x)$ is $e^{-30}/(1+e^{-5}) = 9.3 \times 10^{-14}$, we will focus approximation on $x \in (-5,0)$ and $N \geq 5$. We then carried optimisation of $L$, using Brent's method \cite{Brent1973} as implemented using the \texttt{optimize} function in R. The cost function we choose to minimise is $\sum_j |p^{(k)}(x_j)-\tilde{p}^{(k)}(x_j)|$, where $p^{(k)}(x_j),\tilde{p}^{(k)}(x_j)$ is the $k^\text{th}; k =0, 1, 2, 3$ derivative of the logistic function $p(x)$, and the approximation of the logistic function given in (\ref{ramplogistic.H}) respectively, with $x_{j+1} - x_{j} = 1 \times 10^{-4}$.

\begin{center}
\begin{table}[h]
	\caption{Optimal $L$ given $N$ with approximation errors.}\label{Tab.error}
	\begin{tabular*}{\textwidth}{@{\extracolsep\fill}lcccccc@{\extracolsep\fill}}
		\toprule
		$N$ & $L, p(x)$ & $\sum |Error|, p(x)$ & $\max|Error|, p(x)$ & $L, p'(x)$ & $\sum |Error|, p'(x)$ & $\max|Error|, p'(x)$\\ 
		\midrule
		5 & 1.685101 & 0.251862 & 0.000113 & 1.586104 & 3.028510 & 0.002872 \\ 
		6 & 1.711008 & 0.034265 & 0.000018 & 1.617302 & 0.497280 & 0.000607 \\ 
		7 & 1.731515 & 0.004680 & 0.000003 & 1.644000 & 0.079812 & 0.000122 \\ 
		8 & 1.749221 & 0.000641 & 0.000000 & 1.668535 & 0.012593 & 0.000023 \\ 
		9 & 1.763723 & 0.000088 & 0.000000 & 1.689421 & 0.001961 & 0.000004 \\ 
		10 & 1.776430 & 0.000012 & 0.000000 & 1.705207 & 0.000302 & 0.000001 \\ 
		11 & 1.787302 & 0.000002 & 0.000000 & 1.721333 & 0.000046 & 0.000000 \\  
		\midrule
		$N$ & $L, p''(x)$ & $\sum |Error|, p''(x)$ & $\max|Error|, p''(x)$&$L, p'''(x)$ &$\sum |Error|, p'''(x)$ & $\max|Error|, p'''(x)$ \\
		\midrule
		5 & 1.439925 & 40.288286 & 0.042179 & 1.300801 & 498.249325 & 0.466538 \\ 
		6 & 1.478410 & 8.452863 & 0.011748 & 1.348403 & 135.405588 & 0.170961 \\ 
		7 & 1.511205 & 1.677197 & 0.002984 & 1.384707 & 33.740484 & 0.054576 \\ 
		8 & 1.541406 & 0.319176 & 0.000710 & 1.419602 & 7.860548 & 0.015913 \\ 
		9 & 1.569604 & 0.058792 & 0.000161 & 1.449400 & 1.736764 & 0.004291 \\ 
		10 & 1.592037 & 0.010550 & 0.000035 & 1.478222 & 0.367689 & 0.001096 \\ 
		11 & 1.611702 & 0.001853 & 0.000007 & 1.501922 & 0.075145 & 0.000265 \\ 
		\bottomrule
	\end{tabular*}
\end{table}
\end{center}

\begin{figure*}[ht]
	\centering
\includegraphics[width=0.8\linewidth]{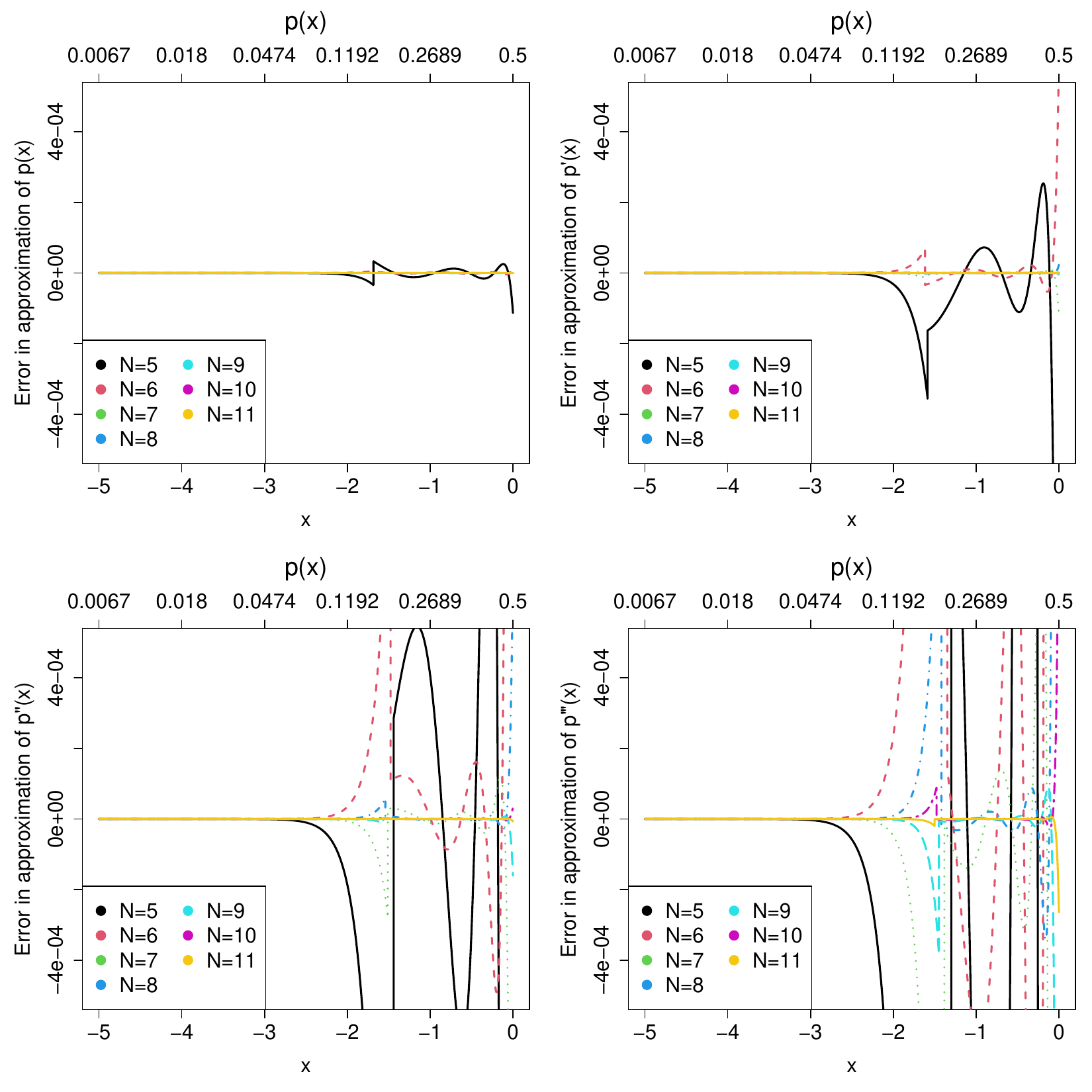}
\caption{Approximation error of the logistic function and its first three derivatives. Note: Error in approximation for $p''(x)$ and $p'''(x)$ is greater than bounds of $y$-axis for lower $N$ for some values of $x$.}\label{Fig1a}
\end{figure*}

In Figure \ref{Fig1a}, we present approximation error for $p(x)$ and its first three derivatives, while in Table \ref{Tab.error}, we present the optimal choice of $L$ for $N=5, \ldots, 11$ along minimised sum of absolute errors and maximum error for the same functions. Note allowing $L$ to vary for each derivative of $p(x)$ reduces the efficiency of any algorithm to calculate moments. This is because, 1. we would need to calculate a new set of evaluations of $\Phi(\cdot)$ for the expected value of each derivative, and 2. we would not be able to use the result of proposition \ref{EP.LN2high} for higher order moments. Hence, to optimise $L$, we fix $N=11$, as we can guarantee the maximum error in the approximation of the four underlying functions needed to evaluate the first four moments of the logit-normal distribution over all $x$ is $<2.7 \times 10^{-4}$, and find $L$ working with the sum of $p(x)$ and its first three derivatives. This resulted in our final choice of $L$ being $1.507306$, with the values of the Chebyshev polynomial coefficients $c_i$ given in Table \ref{tab.optim}. 

\begin{center}
	\begin{table}[h]
		\caption{Chebyshev polynomial coefficients used to approximate $p(x)$ for $x \in (-L,L)$.}\label{tab.optim}
		\begin{tabular*}{\textwidth}{@{\extracolsep\fill}lcccccc ccccc@{\extracolsep\fill}}
			\toprule
		$i$   & 1 & 2 & 3 & 4 & 5 & 6 & 7 & 8 & 9 & 10 & 11 \\ 
$c_i$ & 0.999992 & -0.499904 &  0.332651 & -0.246763 &  0.189081 & -0.139468 &  0.091337 & -0.048596 &  0.018998&  -0.004740 &  0.000559 \\ \midrule
			\bottomrule
		\end{tabular*}
	\end{table}
\end{center}

\section{Testing accuracy and speed of moment calculation}

\subsection{Accuracy}\label{testacc}

Figure \ref{Fig1a}, by showing we have accurately approximated $p(x)$ and its first three derivatives indicates our approach should be accurate for estimating $E(P^k); k = 1, 2, 3, 4$ for any value of $\mu, \sigma$. Nevertheless, to confirm the accuracy of our method, we will estimate $E(P), E(P^2), E(P^3), E(P^4)$ using equations (\ref{eq.EO}), (\ref{eq.EOhigh}) from section \ref{FINDEP} and \ref{FINDEPhigh}, with  values of $L, c_i$ fixed at those given in section \ref{FindingL} and compare to numerical integration estimates. As the logit-normal distribution has the symmetry $P \sim \text{logit}N(\mu,\sigma^2)$ then $1-P \sim \text{logit}N(-\mu,\sigma^2)$,  we will only consider values $\mu \in [0,6]$ in increments of 0.25, which corresponds to $\text{Median}(P) \in [0.5,0.9975]$. For $\sigma$ we will consider the values $(0.001, 0.01, 0.1, 0.25, 0.5, 1, 1.5, 2.5)$. This means we cover, when $\mu$ is near zero, scenarios from near normality for $P$ (low $\sigma$), to near uniform all the way to bi-modality (high $\sigma$) with modes at 0 and 1, as shown in Figure \ref{Fig1c}. 

\begin{figure*}[ht]
	\centering
	\includegraphics[width=0.37\linewidth]{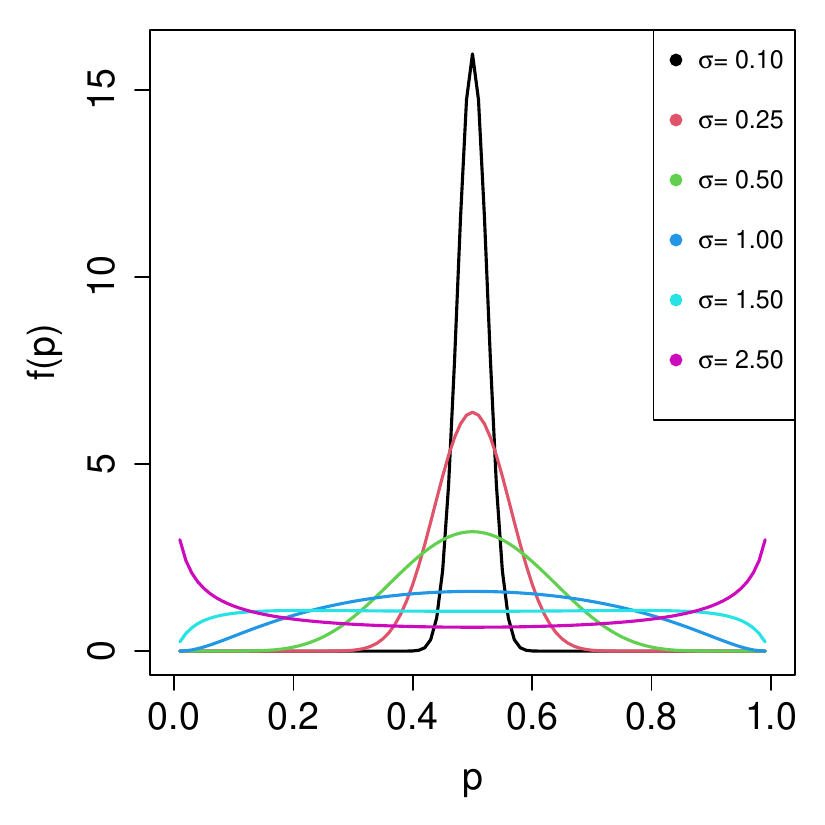}
	\caption{Examples of logit-normal densities with $\mu =0$.}\label{Fig1c}
	\includegraphics[width=0.37\linewidth]{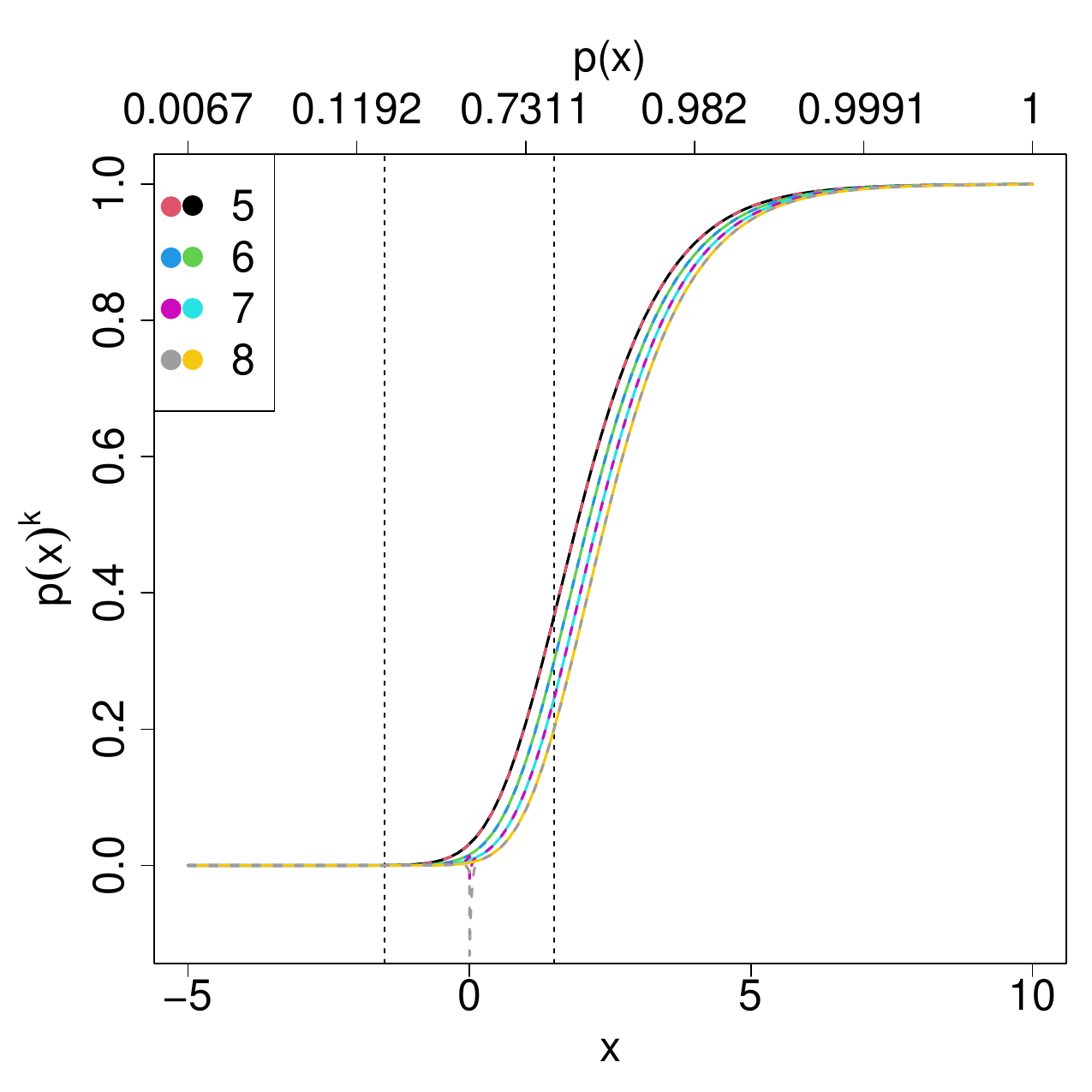}
	\caption{Our implied approximation of $p(x)^k$ for $k = 5, 6, 7, 8$. The vertical dashed lines correspond to $\pm L$, the transition between Maclaurin series and Chebyshev interpolating approximation. Solid lines correspond to the true function,  $p(x)^k$, dashed lines to the approximation $\tilde{p}(x)^k$.}\label{Fig1d}
\end{figure*}

In Tables \ref{tab.accEP}-\ref{tab.accVarP}, comparisons between estimates from our piece-wise approximations of order $N=11$ and numerical integration are provided for the expected value and variance of a logit-normal random variable. Tables showing accuracy for the second to fourth moment, skewness and kurtosis are included in the appendix. For the first four moments and variance, we did obtain accurate estimates for all combinations of $\mu, \sigma$ considered, with the average and worst case error being $6.1 \times 10^{-9}$ and $1.2 \times 10^{-6}$ for $E(P)$, $1.4 \times 10^{-8}$ and $2.2 \times 10^{-6}$ for $\text{Var}(P)$, $2.6 \times 10^{-9}$ and $7.1 \times 10^{-8}$ for $E(P^2)$, $1.2 \times 10^{-8}$ and $2.2 \times 10^{-7}$ for $E(P^3)$ and $5.8 \times 10^{-7}$ and $4.4 \times 10^{-5}$ for $E(P^4)$. 
For skewness and kurtosis, the division by a power of $\text{Var}(P)$ means when $\sigma \leq 0.25$,  estimates obtained are unreliable.

\begin{center}
	\begin{table}[h]\caption{Estimates of $E(P)$ obtained from piecewise approximation of $p(x)$ with error compared to estimate of $E(P)$ obtained from numerical integration.}\label{tab.accEP}
	\scriptsize
	\begin{tabular*}{\textwidth}{@{\extracolsep\fill}lcrcrcrcrcr@{\extracolsep\fill}} \toprule
		$\sigma$	&	$E(P)$	&		Error 			&	$E(P)$	&		Error 	&	$E(P)$	&		Error 		&	$E(P)$	&		Error 		&	$E(P)$	&		Error 		\\ 
		\midrule
		& \multicolumn{2}{c}{$\mu = 0$} & \multicolumn{2}{c}{$\mu = 1.25$} & \multicolumn{2}{c}{$\mu = 2.5$} & \multicolumn{2}{c}{$\mu = 3.75$} & \multicolumn{2}{c}{$\mu = 5.00$}  \\
		0.001	&			0.500000	&	$	3.7\times 10^{-12	}$		&	0.777300	&	$	4.0\times 10^{-11	}$ &			0.924142	&	$	6.8\times 10^{-12	}$			&	0.977023	&	$	7.3\times 10^{-12	}$	&			0.993307	&	$	7.5\times 10^{-12	}$	\\
		0.010		&		0.500000	&	$	3.7\times 10^{-12	}$		&	0.777295	&	$	3.9\times 10^{-11	}$&			0.924139	&	$	6.7\times 10^{-12	}$		&	0.977022	&	$	7.2\times 10^{-12	}$	&			0.993307	&	$	7.3\times 10^{-12	}$	\\
		0.100	&			0.500000	&	$	3.7\times 10^{-12	}$		&	0.776821	&	$	-2.1\times 10^{-11	}$		&	0.923844	&	$	6.7\times 10^{-12	}$			&	0.976915	&	$	7.2\times 10^{-12	}$	&			0.993274	&	$	7.4\times 10^{-12	}$	\\
		0.250 		&	0.500000	&	$	-1.8\times 10^{-13	}$		&	0.774350	&	$	-6.2\times 10^{-10	}$	&	0.922279	&	$	7.8\times 10^{-13	}$	&	0.976346	&	$	7.2\times 10^{-12	}$	&	0.993099	&	$	7.4\times 10^{-12	}$	\\
		0.500		&	0.500000	&	$	-1.1\times 10^{-16	}$		&	0.766062	&	$	-6.0\times 10^{-10	}$&				0.916661	&	$	-1.5\times 10^{-10	}$		&	0.974225	&	$	7.3\times 10^{-12	}$	&		0.992439	&	$	7.4\times 10^{-12	}$	\\
		1.000		&	0.500000	&	$	-1.1\times 10^{-16	}$			&	0.739493	&	$	-3.5\times 10^{-10	}$&				0.894638	&	$	-2.5\times 10^{-10	}$			&	0.964551	&	$	-2.6\times 10^{-11	}$			&	0.989203	&	$	6.5\times 10^{-12	}$	\\
		1.500	&	  0.500000	&	$	0.0 \times 10^{-16}$	&			0.709428	&	$	-2.0\times 10^{-10	}$&				0.862434	&	$	-2.0\times 10^{-10	}$	&			0.945982	&	$	-8.7\times 10^{-11	}$	&			0.981647	&	$	-1.6\times 10^{-11	}$	\\
		2.500		&	0.500000	&	$	-1.1\times 10^{-16	}$			&	0.659245	&	$	-9.3\times 10^{-11	}$&				0.793618	&	$	-3.0\times 10^{-10	}$			&	0.889532	&	$	-5.9\times 10^{-10	}$			&	0.947744	&	$	-3.3\times 10^{-10	}$	\\ 
		\midrule
		& \multicolumn{2}{c}{$\mu = 0.25$} & \multicolumn{2}{c}{$\mu = 1.5$} & \multicolumn{2}{c}{$\mu = 2.75$} & \multicolumn{2}{c}{$\mu = 4.0$} & \multicolumn{2}{c}{$\mu = 5.25$}  \\
		0.001		&	0.562177	&	$	-8.4\times 10^{-11	}$		&	0.817574	&	$	1.8\times 10^{-10	}$&				0.939913	&	$	7.0\times 10^{-12	}$	&			0.982014	&	$	7.4\times 10^{-12	}$	&		0.994780	&	$	7.5\times 10^{-12	}$	\\
		0.010		&	0.562175	&	$	-8.3\times 10^{-11	}$			&	0.817570	&	$	-2.4\times 10^{-9	}$			&	0.939911	&	$	6.9\times 10^{-12	}$			&	0.982013	&	$	7.3\times 10^{-12	}$		&	0.994780	&	$	7.4\times 10^{-12	}$	\\
		0.100		&	0.562024	&	$	-7.0\times 10^{-11	}$			&	0.817102	&	$	-2.6\times 10^{-9	}$		&	0.939665	&	$	7.0\times 10^{-12	}$	&			0.981929	&	$	7.3\times 10^{-12	}$			&	0.994754	&	$	7.4\times 10^{-12	}$	\\
		0.250 		&	0.561248	&	$	-3.4\times 10^{-11	}$		&	0.814651	&	$	-1.4\times 10^{-9	}$ 		&	0.938353	&	$	6.6\times 10^{-12	}$		&	0.981475	&	$	7.3\times 10^{-12	}$		&	0.994617	&	$	7.4\times 10^{-12	}$	\\
		0.500		&	0.558752	&	$	-2.7\times 10^{-11	}$		&	0.806309	&	$	-7.4\times 10^{-10	}$ 			&	0.933597	&	$	-4.9\times 10^{-11	}$			&	0.979781	&	$	7.4\times 10^{-12	}$	&			0.994099	&	$	7.4\times 10^{-12	}$	\\
		1.000			&	0.551493	&	$	-8.6\times 10^{-11	}$			&	0.778527	&	$	-3.7\times 10^{-10	}$ 		&	0.914250	&	$	-1.7\times 10^{-10	}$			&	0.971896	&	$	-1.3\times 10^{-11	}$			&	0.991540	&	$	7.1\times 10^{-12	}$	\\
		1.500			&	0.544085	&	$	-5.1\times 10^{-11	}$		&	0.745691	&	$	-2.2\times 10^{-10	}$ 			&	0.884295	&	$	-1.8\times 10^{-10	}$			&	0.956037	&	$	-6.7\times 10^{-11	}$			&	0.985398	&	$	-9.5\times 10^{-12	}$	\\
		2.500			&	0.532731	&	$	-1.5\times 10^{-11	}$		&	0.688757	&	$	-1.2\times 10^{-10	}$ 			&	0.816045	&	$	-1.0\times 10^{-10	}$			&	0.903957	&	$	-8.0\times 10^{-11	}$		&	0.955654	&	$	-4.5\times 10^{-11	}$	\\ 
		\midrule
		& \multicolumn{2}{c}{$\mu = 0.5$} & \multicolumn{2}{c}{$\mu = 1.75$} & \multicolumn{2}{c}{$\mu = 3.0$} & \multicolumn{2}{c}{$\mu = 4.25$} & \multicolumn{2}{c}{$\mu = 5.5$}  \\
		0.001			&	0.622459	&	$	3.3\times 10^{-13	}$	&	0.851953	&	$	-6.4\times 10^{-10	}$		&	0.952574	&	$	7.1\times 10^{-12	}$	&			0.985936	&	$	7.5\times 10^{-12	}$	&			0.995930	&	$	7.5\times 10^{-12	}$	\\
		0.010	&	0.622457	&	$	4.4\times 10^{-13	}$	&			0.851948	&	$	-6.4\times 10^{-10	}$ 		&	0.952572	&	$	7.0\times 10^{-12	}$	&			0.985936	&	$	7.3\times 10^{-12	}$	&		0.995930	&	$	7.4\times 10^{-12	}$	\\
		0.100		&	0.622173	&	$	-6.5\times 10^{-14	}$			&	0.851510	&	$	-1.2\times 10^{-9	}$ 			&	0.952369	&	$	7.1\times 10^{-12	}$	&		0.985869	&	$	7.3\times 10^{-12	}$	&		0.995910	&	$	7.4\times 10^{-12	}$	\\
		0.250		&		0.620710	&	$	-1.7\times 10^{-11	}$		&	0.849201	&	$	-1.2\times 10^{-9	}$	&	0.951287	&	$	7.0\times 10^{-12	}$		&	0.985510	&	$	7.3\times 10^{-12	}$	&		0.995802	&	$	7.4\times 10^{-12	}$	\\
		0.500		&	0.615976	&	$	-8.1\times 10^{-11	}$		&	0.841228	&	$	-7.1\times 10^{-10	}$ &				0.947330	&	$	-9.1\times 10^{-12	}$		&	0.984163	&	$	7.4\times 10^{-12	}$	&			0.995396	&	$	7.4\times 10^{-12	}$	\\
		1.000			&	0.602027	&	$	-1.7\times 10^{-10	}$		&	0.813571	&	$	-3.7\times 10^{-10	}$ &				0.930676	&	$	-1.4\times 10^{-10	}$		&	0.977793	&	$	-4.2\times 10^{-12	}$			&	0.993380	&	$	7.4\times 10^{-12	}$	\\
		1.500		&	0.587596	&	$	-9.5\times 10^{-11	}$		&	0.779256	&	$	-2.3\times 10^{-10	}$	&			0.903389	&	$	-1.6\times 10^{-10	}$	&			0.964413	&	$	-5.0\times 10^{-11	}$			&	0.988420	&	$	-5.2\times 10^{-10	}$	\\
		2.500		&	0.565237	&	$	-2.8\times 10^{-11	}$		&	0.717074	&	$	-9.0\times 10^{-11	}$ &				0.836848	&	$	-9.9\times 10^{-11	}$		&	0.916909	&	$	-7.4\times 10^{-11	}$			&	0.962543	&	$	-3.9\times 10^{-11	}$	\\ 
		\midrule
		& \multicolumn{2}{c}{$\mu = 0.75$} & \multicolumn{2}{c}{$\mu = 2.0$} & \multicolumn{2}{c}{$\mu = 3.25$} & \multicolumn{2}{c}{$\mu = 4.5$} & \multicolumn{2}{c}{$\mu = 5.75$}  \\
		0.001		&	0.679179	&	$	4.3\times 10^{-11	}$		&	0.880797	&	$	-2.7\times 10^{-11	}$			&	0.962673	&	$	7.2\times 10^{-12	}$	&			0.989013	&	$	7.5\times 10^{-12	}$			&	0.996827	&	$	7.5\times 10^{-12	}$	\\
		0.010		&	0.679175	&	$	4.3\times 10^{-11	}$		&	0.880793	&	$	-2.7\times 10^{-11	}$ 			&	0.962672	&	$	7.1\times 10^{-12	}$	&			0.989013	&	$	7.3\times 10^{-12	}$			&	0.996827	&	$	7.4\times 10^{-12	}$	\\
		0.100		&	0.678790	&	$	2.2\times 10^{-11	}$		&	0.880398	&	$	-6.1\times 10^{-11	}$			&	0.962507	&	$	7.1\times 10^{-12	}$	&			0.988960	&	$	7.3\times 10^{-12	}$			&	0.996812	&	$	7.4\times 10^{-12	}$	\\
		0.250		&		0.676798	&	$	-3.7\times 10^{-13	}$	&	0.878309	&	$	-4.2\times 10^{-10	}$			&	0.961625	&	$	7.1\times 10^{-12	}$		&	0.988676	&	$	7.3\times 10^{-12	}$		&	0.996728	&	$	7.4\times 10^{-12	}$	\\
		0.500		&	0.670292	&	$	-2.0\times 10^{-10	}$		&	0.870994	&	$	-5.4\times 10^{-10	}$		&	0.958377	&	$	3.5\times 10^{-12	}$			&	0.987611	&	$	7.4\times 10^{-12	}$	&			0.996410	&	$	7.4\times 10^{-12	}$	\\
		1.000		&	0.650703	&	$	-2.4\times 10^{-10	}$		&	0.844538	&	$	-3.5\times 10^{-10	}$ 			&	0.944287	&	$	-9.4\times 10^{-11	}$			&	0.982500	&	$	1.5\times 10^{-12	}$			&	0.994825	&	$	7.5\times 10^{-12	}$	\\
		1.500		&	0.629982	&	$	-1.4\times 10^{-10	}$		&	0.809947	&	$	-2.3\times 10^{-10	}$ 			&	0.919892	&	$	8.4\times 10^{-9	}$			&	0.971336	&	$	-3.5\times 10^{-11	}$			&	0.990843	&	$	-4.4\times 10^{-10	}$	\\
		2.500		&	0.597296	&	$	-4.1\times 10^{-11	}$		&	0.744060	&	$	-9.8\times 10^{-11	}$ 		&	0.856019	&	$	-9.6\times 10^{-11	}$			&	0.928465	&	$	-6.6\times 10^{-11	}$			&	0.968506	&	$	-3.3\times 10^{-11	}$	\\ 
		\midrule
		& \multicolumn{2}{c}{$\mu = 1.0$} & \multicolumn{2}{c}{$\mu = 2.25$} & \multicolumn{2}{c}{$\mu = 3.5$} & \multicolumn{2}{c}{$\mu = 4.75$} & \multicolumn{2}{c}{$\mu = 6.0$}  \\
		0.001		&	0.731059	&	$	-1.6\times 10^{-11	}$		&	0.904651	&	$	5.0\times 10^{-12	}$ 		&	0.970688	&	$	7.2\times 10^{-12	}$			&	0.991423	&	$	7.5\times 10^{-12	}$	&			0.997527	&	$	7.6\times 10^{-12	}$	\\
		0.010		&	0.731054	&	$	-1.6\times 10^{-11	}$		&	0.904647	&	$	5.0\times 10^{-12	}$ 		&	0.970686	&	$	7.2\times 10^{-12	}$			&	0.991422	&	$	7.3\times 10^{-12	}$	&		0.997527	&	$	7.4\times 10^{-12	}$	\\
		0.100		&	0.730606	&	$	4.0\times 10^{-12	}$		&	0.904302	&	$	3.3\times 10^{-12	}$ 		&	0.970554	&	$	7.2\times 10^{-12	}$			&	0.991381	&	$	7.4\times 10^{-12	}$	&			0.997515	&	$	7.4\times 10^{-12	}$	\\
		0.250		&		0.728278	&	$	-9.9\times 10^{-11	}$		&	0.902471	&	$	-6.3\times 10^{-11	}$ 	&	0.969842	&	$	7.2\times 10^{-12	}$		&	0.991158	&	$	7.4\times 10^{-12	}$		&	0.997450	&	$	7.4\times 10^{-12	}$	\\
		0.500		&	0.720581	&	$	-3.8\times 10^{-10	}$		&	0.895975	&	$	-3.2\times 10^{-10	}$ 		&	0.967206	&	$	6.7\times 10^{-12	}$		&	0.990318	&	$	7.4\times 10^{-12	}$	&		0.997201	&	$	7.4\times 10^{-12	}$	\\
		1.000		&	0.696735	&	$	-3.1\times 10^{-10	}$		&	0.871495	&	$	-3.0\times 10^{-10	}$		&	0.955459	&	$	-5.1\times 10^{-11	}$			&	0.986241	&	$	4.9\times 10^{-12	}$		&	0.995957	&	$	7.5\times 10^{-12	}$	\\
		1.500			&	0.670739	&	$	-3.1\times 10^{-10	}$	&	0.837674	&	$	-2.2\times 10^{-10	}$ 		&	0.934013	&	$	-1.1\times 10^{-10	}$		&	0.977017	&	$	-2.4\times 10^{-11	}$		&	0.992777	&	$	-9.9\times 10^{-11	}$	\\
		2.500		&	0.628698	&	$	-5.9\times 10^{-11	}$		&	0.769604	&	$	-1.0\times 10^{-10	}$		&	0.873570	&	$	-9.2\times 10^{-11	}$		&	0.938713	&	$	-1.2\times 10^{-6	}$			&	0.973639	&	$	-4.2\times 10^{-10	}$	\\ 
		\bottomrule
	\end{tabular*}
\end{table}
\end{center}

Our second test looked at the applicability of Proposition \ref{EP.LN2high}. In practice, this is a test of how good the underlying approximating functions implied in Proposition \ref{EP.LN2high} are. In Figure \ref{Fig1d}, we provide evidence of our ability to approximate $p(x)^k$ for $k = 5, 6, 7, 8$. For $k=5, 6$, our choice of $c_i, L$ remains good enough to enable accurate estimation of the corresponding logit-normal moment. For $k = 7$ and more so for $k=8$, we would start seeing systematic error appear when estimating the $k^\text{th}$ moment for a logit-normal random variable with non-trivial probability in the region $0.5 \pm \epsilon$ for some $\epsilon > 0$. If we push $k$ higher, approximations will get worse. To show this, we calculated $E(P^{k}); k = 1, 2, \ldots, 10$ using both the proposition and numerical integration when $\mu =1, \sigma = 0.8$. As shown in Table \ref{Tabhigh}, the Proposition result does hold for $k > 4$. However, the increasing error in approximation for higher order derivatives of $p(x)$, shown in Figure \ref{Fig1a}, means approximation accuracy of $E(P^k)$ declines noticeably when $\mu =1, \sigma =0.8$ for $k \geq 8$, if $N =11, L = 1.507306$. This implies our method for finding logit-normal moments is not generally applicable to evaluating the integral needed to find the marginal likelihood of $\boldsymbol{\beta}$, $L(\boldsymbol{\beta})$ that appears in binomial logistic mixed models with only an observation level random effect. 

\begin{center}
\begin{table}[h]\caption{Estimates of $\text{Var}(P)$ obtained from piecewise approximation of $p(x)$ with error compared to estimate of $\text{Var}(P)$ obtained from numerical integration.}\label{tab.accVarP}
	\scriptsize
	\begin{tabular*}{\textwidth}{@{}lcrcrcrcrcr@{}} \toprule
		$\sigma$	&	$\text{Var}(P)$	&		Error 			&	$\text{Var}(P)$	&		Error 	&	$\text{Var}(P)$	&		Error 		&	$\text{Var}(P)$	&		Error 		&	$\text{Var}(P)$	&		Error 		\\ \midrule
		&	\multicolumn{2}{c}{$\mu	=		0$}		&	\multicolumn{2}{c}{$\mu	=		1.25$}		&	\multicolumn{2}{c}{$\mu	=		2.5$}		&	\multicolumn{2}{c}{$\mu	=		3.75$}		&	\multicolumn{2}{c}{$\mu		=		5.00$}	\\
		0.001	&	0.000000	&	$	-7.1\times 10^{-8	}$	&	0.000000	&	$	3.8\times 10^{-10	}$	&	0.000000	&	$	-7.3\times 10^{-12	}$	&	0.000000	&	$	-7.1\times 10^{-12	}$	&	0.000000 &	$	-7.5\times 10^{-12	}$	\\ 	
		0.010	&	0.000006	&	$	-4.5\times 10^{-8	}$	&	0.000003	&	$	3.7\times 10^{-10	}$	&	0.000000	&	$	-7.3\times 10^{-12	}$	&	0.000000	&	$	-7.1\times 10^{-12	}$	&	0.000000 &	$	-7.3\times 10^{-12	}$	\\ 	
		0.100	&	0.000622	&	$	-7.2\times 10^{-9	}$	&	0.000300	&	$	-3.9\times 10^{-10	}$	&	0.000050	&	$	-8.3\times 10^{-12	}$	&	0.000005	&	$	-7.1\times 10^{-12	}$	&	0.000000 &	$	-7.3\times 10^{-12	}$	\\ 	
		0.250	&	0.003789	&	$	-3.3\times 10^{-9	}$	&	0.001884	&	$	-7.4\times 10^{-9	}$	&	0.000325	&	$	-7.3\times 10^{-11	}$	&	0.000034	&	$	-7.1\times 10^{-12	}$	&	0.000003 &	$	-7.3\times 10^{-12	}$	\\ 	
		0.500	&	0.013956	&	$	-1.9\times 10^{-9	}$	&	0.007587	&	$	-7.1\times 10^{-9	}$	&	0.001517	&	$	-1.7\times 10^{-9	}$	&	0.000174	&	$	-8.2\times 10^{-12	}$	&	0.000016&	$	-7.3\times 10^{-12	}$	\\ 	
		1.000	&	0.043379	&	$	-3.5\times 10^{-9	}$	&	0.028821	&	$	-4.7\times 10^{-9	}$	&	0.008965	&	$	-2.9\times 10^{-9	}$	&	0.001529	&	$	-4.2\times 10^{-10	}$	&	0.000176 &	$	-2.1\times 10^{-11	}$	\\ 	
		1.500	&	0.073272	&	$	-4.1\times 10^{-9	}$	&	0.056026	&	$	-3.9\times 10^{-9	}$	&	0.025467	&	$	-2.6\times 10^{-9	}$	&	0.007210	&	$	-1.1\times 10^{-9	}$	&	0.001361 &	$	-2.2\times 10^{-10	}$	\\ 	
		2.500	&	0.118925	&	$	-3.4\times 10^{-9	}$	&	0.104409	&	$	-3.1\times 10^{-9	}$	&	0.070875	&	$	-2.1\times 10^{-9	}$	&	0.037530	&	$	-6.1\times 10^{-10	}$	&	0.015707 & 	$	-2.6\times 10^{-10	}$	\\ \midrule
		&	\multicolumn{2}{c}{$\mu	=		0.25$}		&	\multicolumn{2}{c}{$\mu	=		1.5$}		&	\multicolumn{2}{c}{$\mu	=		2.75$}		&	\multicolumn{2}{c}{$\mu	=		4.0$}		&	\multicolumn{2}{c}{$\mu		=		5.25$}	\\
		0.001	&	0.000000	&	$	2.1\times 10^{-10	}$	&	0.000000	&	$	-4.8\times 10^{-9	}$	&	0.000000	&	$	-6.6\times 10^{-12	}$	&	0.000000	&	$	-7.2\times 10^{-12	}$	&	0.000000 &	$	-7.5\times 10^{-12	}$	\\ 	
		0.010	&	0.000006	&	$	2.0\times 10^{-10	}$	&	0.000002	&	$	-3.1\times 10^{-8	}$	&	0.000000	&	$	-6.6\times 10^{-12	}$	&	0.000000	&	$	-7.1\times 10^{-12	}$	&	0.000000	& $	-7.3\times 10^{-12	}$	\\ 	
		0.100	&	0.000603	&	$	-9.0\times 10^{-10	}$	&	0.000223	&	$	-3.0\times 10^{-8	}$	&	0.000032	&	$	-6.6\times 10^{-12	}$	&	0.000003	&	$	-7.1\times 10^{-12	}$	&	0.000000	& $	-7.3\times 10^{-12	}$	\\ 	
		0.250	&	0.003680	&	$	-2.1\times 10^{-9	}$	&	0.001415	&	$	-1.6\times 10^{-8	}$	&	0.000213	&	$	-1.1\times 10^{-11	}$	&	0.000021	&	$	-7.1\times 10^{-12	}$	&	0.000002	& $	-7.3\times 10^{-12	}$	\\ 	
		0.500	&	0.013608	&	$	-1.8\times 10^{-9	}$	&	0.005877	&	$	-8.6\times 10^{-9	}$	&	0.001015	&	$	-6.3\times 10^{-10	}$	&	0.000109	&	$	-7.3\times 10^{-12	}$	&	0.000010	& $	-7.3\times 10^{-12	}$	\\ 	
		1.000	&	0.042667	&	$	-3.6\times 10^{-9	}$	&	0.024153	&	$	-4.7\times 10^{-9	}$	&	0.006562	&	$	-2.3\times 10^{-9	}$	&	0.001017	&	$	-2.4\times 10^{-10	}$	&	0.000111	& $	-1.3\times 10^{-11	}$	\\ 	
		1.500	&	0.072486	&	$	-4.1\times 10^{-9	}$	&	0.049831	&	$	-3.7\times 10^{-9	}$	&	0.020510	&	$	-2.3\times 10^{-9	}$	&	0.005322	&	$	-8.1\times 10^{-10	}$	&	0.000935	& $	-1.5\times 10^{-10	}$	\\ 	
		2.500	&	0.118306	&	$	-3.4\times 10^{-9	}$	&	0.098614	&	$	-2.9\times 10^{-9	}$	&	0.063641	&	$	-2.2\times 10^{-9	}$	&	0.032114	&	$	-1.4\times 10^{-9	}$	&	0.012846	& $	-6.8\times 10^{-10	}$	\\  \midrule
		&	\multicolumn{2}{c}{$\mu	=		0.5$}		&	\multicolumn{2}{c}{$\mu	=		1.75$}		&	\multicolumn{2}{c}{$\mu	=		3.0$}		&	\multicolumn{2}{c}{$\mu	=		4.25$}		&	\multicolumn{2}{c}{$\mu		=		5.5$}	\\
		0.001	&	0.000000	&	$	5.8\times 10^{-10	}$	&	0.000000	&	$	-7.2\times 10^{-9	}$	&	0.000000	&	$	-6.8\times 10^{-12	}$	&	0.000000	&	$	-7.4\times 10^{-12	}$	&	0.000000	& $	-7.5\times 10^{-12	}$	\\ 	
		0.010	&	0.000006	&	$	5.7\times 10^{-10	}$	&	0.000002	&	$	-7.3\times 10^{-9	}$	&	0.000000	&	$	-6.7\times 10^{-12	}$	&	0.000000	&	$	-7.2\times 10^{-12	}$	&	0.000000	& $	-7.3\times 10^{-12	}$	\\ 	
		0.100	&	0.000550	&	$	1.6\times 10^{-11	}$	&	0.000160	&	$	-1.3\times 10^{-8	}$	&	0.000021	&	$	-6.7\times 10^{-12	}$	&	0.000002	&	$	-7.2\times 10^{-12	}$	&	0.000000 &	$	-7.4\times 10^{-12	}$	\\ 	
		0.250	&	0.003372	&	$	-5.8\times 10^{-10	}$	&	0.001024	&	$	-1.4\times 10^{-8	}$	&	0.000137	&	$	-6.9\times 10^{-12	}$	&	0.000013	&	$	-7.2\times 10^{-12	}$	&	0.000001	& $	-7.4\times 10^{-12	}$	\\ 	
		0.500	&	0.012623	&	$	-2.0\times 10^{-9	}$	&	0.004387	&	$	-8.2\times 10^{-9	}$	&	0.000666	&	$	-1.9\times 10^{-10	}$	&	0.000068	&	$	-7.2\times 10^{-12	}$	&	0.000006	& $	-7.3\times 10^{-12	}$	\\ 	
		1.000	&	0.040604	&	$	-3.8\times 10^{-9	}$	&	0.019649	&	$	-4.5\times 10^{-9	}$	&	0.004696	&	$	-1.6\times 10^{-9	}$	&	0.000667	&	$	-1.3\times 10^{-10	}$	&	0.000069	& $	-9.5\times 10^{-12	}$	\\ 	
		1.500	&	0.070179	&	$	-4.1\times 10^{-9	}$	&	0.043416	&	$	-3.5\times 10^{-9	}$	&	0.016216	&	$	-2.0\times 10^{-9	}$	&	0.003868	&	$	-6.1\times 10^{-10	}$	&	0.000635	& $	9.2\times 10^{-10	}$	\\ 	
		2.500	&	0.116470	&	$	-3.3\times 10^{-9	}$	&	0.092190	&	$	-2.8\times 10^{-9	}$	&	0.056585	&	$	-2.0\times 10^{-9	}$	&	0.027225	&	$	-1.2\times 10^{-9	}$	&	0.010415	& $	-5.8\times 10^{-10	}$	\\ \midrule
		&	\multicolumn{2}{c}{$\mu	=		0.75$}		&	\multicolumn{2}{c}{$\mu	=		2.0$}		&	\multicolumn{2}{c}{$\mu	=		3.25$}		&	\multicolumn{2}{c}{$\mu	=		4.5$}		&	\multicolumn{2}{c}{$\mu		=		5.75$}	\\
		0.001	&	0.000000	&	$	-2.0\times 10^{-10	}$	&	0.000000	&	$	-3.8\times 10^{-10	}$	&	0.000000	&	$	-6.9\times 10^{-12	}$	&	0.000000	&	$	-7.4\times 10^{-12	}$	&	0.000000	& $	-7.5\times 10^{-12	}$	\\ 	
		0.010	&	0.000005	&	$	-2.0\times 10^{-10	}$	&	0.000001	&	$	-3.8\times 10^{-10	}$	&	0.000000	&	$	-6.8\times 10^{-12	}$	&	0.000000	&	$	-7.2\times 10^{-12	}$	&	0.000000 &	$	-7.3\times 10^{-12	}$	\\ 	
		0.100	&	0.000474	&	$	-8.2\times 10^{-11	}$	&	0.000111	&	$	-7.5\times 10^{-10	}$	&	0.000013	&	$	-6.9\times 10^{-12	}$	&	0.000001	&	$	-7.2\times 10^{-12	}$	&	0.000000	& $	-7.4\times 10^{-12	}$	\\ 	
		0.250	&	0.002922	&	$	-1.8\times 10^{-10	}$	&	0.000717	&	$	-4.8\times 10^{-9	}$	&	0.000087	&	$	-6.9\times 10^{-12	}$	&	0.000008	&	$	-7.2\times 10^{-12	}$	&	0.000001	& $	-7.4\times 10^{-12	}$	\\ 	
		0.500	&	0.011152	&	$	-2.8\times 10^{-9	}$	&	0.003167	&	$	-6.2\times 10^{-9	}$	&	0.000431	&	$	-4.9\times 10^{-11	}$	&	0.000042	&	$	-7.2\times 10^{-12	}$	&	0.000004	& $	-7.4\times 10^{-12	}$	\\ 	
		1.000	&	0.037397	&	$	-4.2\times 10^{-9	}$	&	0.015536	&	$	-4.1\times 10^{-9	}$	&	0.003291	&	$	-1.1\times 10^{-9	}$	&	0.000432	&	$	-7.2\times 10^{-11	}$	&	0.000043	& $	-8.1\times 10^{-12	}$	\\ 	
		1.500	&	0.066503	&	$	-4.0\times 10^{-9	}$	&	0.037068	&	$	-3.3\times 10^{-9	}$	&	0.012592	&	$	-1.7\times 10^{-8	}$	&	0.002770	&	$	-4.5\times 10^{-10	}$	&	0.000426	& $	8.0\times 10^{-10	}$	\\ 	
		2.500	&	0.113474	&	$	-3.3\times 10^{-9	}$	&	0.085310	&	$	-2.7\times 10^{-9	}$	&	0.049823	&	$	-1.9\times 10^{-9	}$	&	0.022869	&	$	-1.1\times 10^{-9	}$	&	0.008373	& $	-4.8\times 10^{-10	}$	\\ \midrule
		&	\multicolumn{2}{c}{$\mu	=		1.0$}		&	\multicolumn{2}{c}{$\mu	=		2.25$}		&	\multicolumn{2}{c}{$\mu	=		3.5$}		&	\multicolumn{2}{c}{$\mu	=		4.75$}		&	\multicolumn{2}{c}{$\mu		=		6.0$}	\\
		%		0.001	&	0.000000	&	$	-8.4\times 10^{-11	}$	&	0.000000	&	$	-2.5\times 10^{-11	}$	&	0.000000	&	$	-7.0\times 10^{-12	}$	&	0.000000	&	$	-7.4\times 10^{-12	}$	&	0.000000	& $	-7.5\times 10^{-12	}$	\\ 	
		0.010	&	0.000004	&	$	-8.4\times 10^{-11	}$	&	0.000001	&	$	-2.5\times 10^{-11	}$	&	0.000000	&	$	-7.0\times 10^{-12	}$	&	0.000000	&	$	-7.3\times 10^{-12	}$	&	0.000000	& $	-7.4\times 10^{-12	}$	\\ 	
		0.100	&	0.000386	&	$	-3.9\times 10^{-11	}$	&	0.000075	&	$	-4.4\times 10^{-11	}$	&	0.000008	&	$	-7.0\times 10^{-12	}$	&	0.000001	&	$	-7.3\times 10^{-12	}$	&	0.000000	& $	-7.4\times 10^{-12	}$	\\ 	
		0.250	&	0.002403	&	$	-1.3\times 10^{-9	}$	&	0.000488	&	$	-7.8\times 10^{-10	}$	&	0.000055	&	$	-7.0\times 10^{-12	}$	&	0.000005	&	$	-7.3\times 10^{-12	}$	&	0.000000 &	$	-7.4\times 10^{-12	}$	\\ 	
		0.500	&	0.009403	&	$	-4.7\times 10^{-9	}$	&	0.002220	&	$	-3.7\times 10^{-9	}$	&	0.000275	&	$	-1.5\times 10^{-11	}$	&	0.000026	&	$	-7.3\times 10^{-12	}$	&	0.000002 &	$	-7.4\times 10^{-12	}$	\\ 	
		1.000	&	0.033352	&	$	-4.5\times 10^{-9	}$	&	0.011955	&	$	-3.6\times 10^{-9	}$	&	0.002263	&	$	-7.3\times 10^{-10	}$	&	0.000277	&	$	-3.8\times 10^{-11	}$	&	0.000027	& $	-7.6\times 10^{-12	}$	\\ 	
		1.500	&	0.061689	&	$	-3.8\times 10^{-9	}$	&	0.031025	&	$	-3.0\times 10^{-9	}$	&	0.009608	&	$	-1.3\times 10^{-9	}$	&	0.001955	&	$	-3.2\times 10^{-10	}$	&	0.000283	& $	-1.1\times 10^{-9	}$	\\ 	
		2.500	&	0.109412	&	$	-3.2\times 10^{-9	}$	&	0.078150	&	$	-2.5\times 10^{-9	}$	&	0.043448	&	$	-1.7\times 10^{-9	}$	&	0.019038	&	$	2.2\times 10^{-6	}$	&	0.006674	& $3.6\times 10^{-10	}$	\\ 
		\bottomrule
	\end{tabular*}
\end{table}

\begin{table}[h]
	\caption{Estimates of $E(P^k)$ up to 10, using the result of Proposition \ref{EP.LN2high} compared to Numerical Integration when $\mu = 1, \sigma =0.8$.}\label{Tabhigh}
	\centering
	\begin{tabular*}{\textwidth}{@{}lccc ccc cccc@{}} \toprule 
		Estimation method & $E(P)$ & $E(P^2)$ & $E(P^3)$ & $E(P^4)$ & $E(P^5)$ & $E(P^6)$ & $E(P^7)$ & $E(P^8)$ & $E(P^9)$ & $E(P^{10})$  \\  \midrule
		Numerical Integration & 0.707042 & 0.522576 & 0.399138 & 0.312699 & 0.249999 & 0.203222 & 0.167511 & 0.139720 & 0.117736 & 0.100102  \\ 
		Proposition 2 & 0.707042 & 0.522576 & 0.399138 & 0.312699 & 0.249998 & 0.203176 & 0.167343 & 0.138362 & 0.112458 & 0.078556 \\ 
		\bottomrule
	\end{tabular*}
\end{table}
\end{center}

\subsection{Speed}

To evaluate speed of calculation, we have chosen to focus on $E(P), E(P^2), E(P^3), E(P^4), \text{Var}(P)$. Because of the highly sequential nature of the component functions needed to calculate our estimates of logit-normal moments, the marginal cost of calculating multiple moments simultaneously drops as the number of moments needed increases. This means we will provide two speed comparisons. The first comparison consists of calculating a single moment, $E(P)$, using our method and numerical integration. The second will be a comparison of calculating $E(P), E(P^2), E(P^3), E(P^4), \text{Var}(P)$ using our method and numerical integration, with results presented in Table \ref{TabspeedEP}. For speed testing, we implemented our approach in R but written in C++ code using the package \texttt{Rcpp} \cite{Eddelbuettel2011}, while numerical integration was done using the \texttt{integrate} function in R. To determine the time distribution, the functions were evaluated 10,000 times and timed using the R package \texttt{microbenchmark} \cite{Mersmann2025}. The values for $\mu, \sigma$ were $\mu = 0.2, 1.2, \sigma = 0.1, 1, 2.5$. This was done to see if speed of evaluation was affected by the choice of parameter values.

\begin{center}
\begin{table}[h]
	\centering\caption{Quantiles of time taken to execute a single evaluation of logit-normal moments based on piecewise approximation and numerical integration. Times are reported in $10^{-6}$ seconds.}\label{TabspeedEP}
	\begin{tabular*}{\textwidth}{@{}l ccc ccc ccc ccc@{}} \toprule
		& \multicolumn{6}{c}{Calculating $E(P)$} &  \multicolumn{6}{c}{Calculating $E(P), \ldots, E(P^4), \text{Var}(P)$} \\
		& \multicolumn{3}{c}{$\mu	=0.2$} & \multicolumn{3}{c}{$\mu	=1.2$} &  \multicolumn{3}{c}{$\mu	=0.2$} & \multicolumn{3}{c}{$\mu	=1.2$} \\
		$\sigma$ & 0.1 & 1 & 2.5 &  0.1 & 1 & 2.5 &  0.1 & 1 & 2.5 & 0.1 & 1 & 2.5 \\
		\midrule 
		Quantile & \multicolumn{12}{c}{Piecewise approximation} \\ 
		\midrule
		10\% & 5.00 & 5.30 & 5.40 & 4.50 & 5.20 & 5.40 & 5.10 & 5.40 & 5.50 & 4.60 & 5.30 & 5.60 \\ 
		25\% & 5.20 & 5.50 & 5.60 & 4.70 & 5.40 & 5.60 & 5.30 & 5.60 & 5.70 & 4.80 & 5.50 & 5.80 \\ 
		50\% & 5.50 & 5.70 & 5.80 & 4.90 & 5.70 & 5.90 & 5.60 & 5.80 & 5.90 & 5.00 & 5.80 & 6.00 \\ 
		75\% & 5.80 & 6.10 & 6.20 & 5.30 & 6.10 & 6.30 & 6.00 & 6.20 & 6.30 & 5.40 & 6.20 & 6.40 \\ 
		90\% & 7.10 & 7.40 & 7.60 & 6.30 & 7.30 & 7.60 & 7.30 & 7.50 & 7.70 & 6.60 & 7.50 & 7.80 \\ 
		\midrule
		Quantile & \multicolumn{12}{c}{Numerical integration} \\ 
		\midrule 
		10\% & 28.00 & 47.79 & 47.60 & 27.60 & 37.50 & 47.40 & 118.10 & 190.60 & 218.70 & 117.10 & 168.70 & 218.60 \\ 
		25\% & 29.00 & 49.40 & 49.10 & 28.58 & 38.80 & 49.00 & 122.10 & 197.30 & 226.30 & 120.90 & 174.40 & 225.90 \\ 
		50\% & 30.20 & 51.40 & 51.20 & 29.80 & 40.40 & 51.00 & 127.10 & 205.70 & 236.10 & 126.10 & 182.30 & 235.80 \\ 
		75\% & 32.10 & 54.70 & 54.30 & 31.90 & 43.10 & 54.12 & 137.00 & 223.50 & 257.70 & 136.10 & 198.80 & 258.10 \\ 
		90\% & 39.10 & 69.70 & 69.60 & 39.50 & 53.50 & 68.71 & 170.80 & 266.60 & 307.81 & 171.00 & 237.11 & 306.80 \\ 
		\bottomrule
	\end{tabular*}
\end{table}
\end{center}

In Table \ref{TabspeedEP}, we see evaluation time for $E(P)$ using numerical integration is somewhat sensitive to choice of $\sigma$, with evaluation being faster for low $\sigma$. This was less so for our approach based on piecewise approximation of $p(x)$. Our approach for finding $E(P)$ was roughly 5.5-6 times faster when $\sigma = 0.1$ and was roughly 7-9 times faster  when $\sigma =1, 2.5$. When calculating the first four moments and variance simultaneously, the low marginal cost of obtaining higher order moments and the variance using our approach becomes apparent. On average, obtaining  $E(P^2), E(P^3), E(P^4), \text{Var}(P)$ in addition to $E(P)$ increased run time by 2 to 3 \%, or $0.1-0.2 \times 10^{-6}$ s. Using standard numerical integration to obtain $E(P^2), E(P^3), E(P^4), \text{Var}(P)$ in addition to $E(P)$ caused a linear increase in coast, as numerical solving of four rather than one integral was needed. Our approach to finding $E(P), E(P^2), E(P^3), E(P^4), \text{Var}(P)$ was on average between 23-25 times faster when $\sigma =0.1$, 31-35 times faster when $\sigma =1$ and 39-40 times faster when $\sigma=2.5$. The low cost of finding higher order moments/expectation with respect to $\mu$ of $E(P)$ using our piecewise approximation to $p(x)$ indicates our approach would be particularly useful in algorithms, such as logistic regression models where parameters are estimated using expectation propagation, where we need moments of multiple order simultaneously.

\section{An inferential application that will work: Expectation propagation for logistic regression}

Our method for approximating logit-normal moments has been shown to have a practical upper limit of 7/8 before losing accuracy. However there are many inferential problems where logit-normal moments under order 7/8 are needed. As an example, we will now demonstrate our results are sufficient to remove the need to use numerical integration when fitting logistic regression in using expectation propagation for many choices of prior. As the illustrative dataset, we chose the {\it lbw} dataset from the {\it COUNT} package in R, originally given in \cite{Hosmer2000}.

\subsection{Dataset and variables}

The {\it lbw} dataset consists of 189 observations with 10 variables. For the response variable, we have chosen to use \texttt{low}, an indicator (1 = yes, 0 = no) whether a birth was classed as low-weight or not . For our illustrative example, we have chosen to consider only 3 of the possible predictor variables. These are \texttt{Smoke}, a binary variable indicating if the mother had a history of smoking (1 = yes, 0 = no), \texttt{lwt}, a continuous variable giving maternal weight at last menstrual period in pounds, and \texttt{Hypertension}, a binary variable indicating if the mother had a history of hypertension (1 = yes, 0 = no). For the continuous variable \texttt{lwt}, we will standardise the covariate values to create \texttt{Standardised Maternal Weight} before running the algorithm,. 

\subsection{Estimation methods}

Our primary interest in implementing expectation propagation (EP) is to demonstrate that for logistic regression, rewriting the moment matching step into the form of Proposition \ref{prop1} leads to faster implementation, as we know the required logit-normal moments. This means we will implement two versions of EP, written in R code. The first, EP-NI, will implement sequential-EP, as described in \cite{Gelman2014} using numerical integration to solve the moment matching step. The second, EP-LN, will implement sequential-EP using our formula for $E(P), \frac{{\rm d}E(P)}{{\rm d} \mu}, \frac{{\rm d}^2E(P)}{{\rm d} \mu^2}$ given in (\ref{eq.EO}), (\ref{eq.EOdiff}) to avoid numerical integration. For more details on the particulars of the algorithm, particularly for accounting that $f_i(\boldsymbol{\eta}_i)$ switches between ${\bf p}_i$ and $1-{\bf p}_i$ when  ${\bf y}_i$ changes from 1 to 0,  see the appendix. While not of main interest, we will include a MCMC comparison to demonstrate the accuracy of expectation propagation at posterior approximation. This comparison will be with a Gibbs sampler using P{\'o}lya-Gamma data augmentation \cite{Polson2013}. For all MCMC analyses, 3 chains will be run for 30,000 iterations with the first 15,000 iterations removed as burn-in. Convergence will be checked using Gelman-Rubin statistics \cite{Gelman1992}.

For the prior on the regression coefficients, we choose $p(\boldsymbol{\beta}) := \mathcal{N}(\boldsymbol{\beta}|\tilde{\boldsymbol{\beta}}, \tilde{\boldsymbol{\Sigma}})$, where $\tilde{\boldsymbol{\beta}} = (0, -0.1, 0.1, 0.2)$, and $\tilde{\boldsymbol{\Sigma}}_{ij} = 1.3\times 0.5^{|i-j|}$. There are two motivations for this. The first is with a normal prior, the moment matching step required in the EP algorithm can skip over the prior \cite{Chopin2017}. Second, in applications of EP to determine the marginal likelihood in frequentist probit mixed models \cite{Hall2019}, the prior equivalent is also normal. Moreover \cite{Hall2019} state the extension of their algorithm to a logistic mixed model requires solving univariate integrals equivalent to (\ref{eq.EP}) with $f_i(\boldsymbol{\eta}_i) = (1-{\bf p}_i)$. Hence showing our approach works for logistic regression with a normal prior implies our work could be extended to the logistic examples mentioned in \cite{Hall2019}.

\subsection{Results}

In Figure \ref{Fig1e}, we present posterior density plots estimated using MCMC and EP-NI, EP-LN. As expected, very high accuracy of approximation was obtained. This was confirmed by calculating the accuracy diagnostic from \cite{Luts2015}, 
\begin{equation} \text{accuracy}(g(\boldsymbol{\beta}_i)) = 100 \times \bigg( 1 - 0.5 \int_{-\infty}^\infty |g(\boldsymbol{\beta}_i)-p(\boldsymbol{\beta}_i|{\bf y})| d\boldsymbol{\beta}_i \bigg),\end{equation} 
where $g(\boldsymbol{\beta})$ is the approximate posterior and $p(\boldsymbol{\beta}|{\bf y})$ is the \lq true\rq\hspace{0.5mm} posterior as estimated with MCMC samples. For the 4 parameters in our test model, accuracy of EP approximation was $0.9882$ for the \texttt{Intercept}, $0.9926$ for \texttt{Smoke}, $0.9788$ for \texttt{Standardised Maternal Weight} and $0.9934$ for \texttt{Hypertension}.

\begin{center}
\begin{table}[h]
	\centering\caption{Estimates of $\hat{\boldsymbol{\beta}}, \hat{\boldsymbol{\Sigma}}$ obtained from EP-LN, alongside relative error between EP-NI and EP-LN. Relative error was calculated the difference between NI and LN estimates divided by the LN estimate. Note \texttt{Maternal} has been used to denote \texttt{Standardised Maternal Weight} and in the row labels, the first letter denotes each coefficient.}\label{TabEPNILN}
	\begin{tabular*}{\textwidth}{@{}l rr rr rr rr rr@{}} \toprule
		& \multicolumn{2}{c}{} & \multicolumn{8}{c}{Coefficient $j$} \\
		\multicolumn{2}{c}{Coefficient} & & \multicolumn{2}{c}{\texttt{I} =\texttt{Intercept}} &  \multicolumn{2}{c}{\texttt{S} =\texttt{Smoke}} & \multicolumn{2}{c}{\texttt{M} = \texttt{Maternal}} & \multicolumn{2}{c}{\texttt{H} =\texttt{Hypertension}}  \\
		$i$ & $\hat{\boldsymbol{\beta}}_i$,  & relative error & $\hat{\boldsymbol{\Sigma}}_{i,j}$,   & relative error & $\hat{\boldsymbol{\Sigma}}_{i,j}$,   & relative error & $\hat{\boldsymbol{\Sigma}}_{i,j}$,   & relative error & $\hat{\boldsymbol{\Sigma}}_{i,j}$,   & relative error \\
		\midrule
		\texttt{I} & -1.071865 & $-7.1 \times 10^{-9}$ & 0.041916 & $1.1 \times 10^{-7}$ &  &  &  & &  &  \\ 
		\texttt{S} & 0.456130 & $-3.7 \times 10^{-8}$ & -0.035515 & $1.1 \times 10^{-7}$ & 0.089018 & $-4.6 \times 10^{-8}$ &  &  &  &  \\ 
		\texttt{M} & -0.456085 & $-6.5 \times 10^{-9}$ & 0.004408 & $1.1 \times 10^{-6}$ & 0.003743 & $-1.6 \times 10^{-6}$ & 0.033009 & $8.2 \times 10^{-8}$ &  &   \\
		\texttt{H} & 1.166203 & $2.0 \times 10^{-8}$ & -0.021236 & $2.1 \times 10^{-7}$ & -0.002845 & $-2.8 \times 10^{-6}$ & -0.018205 & $-1.3 \times 10^{-7}$ & 0.300135 & $6.0 \times 10^{-8}$  \\ 
		\bottomrule
	\end{tabular*}
\end{table}
\end{center}

Focusing on the comparison between EP-NI and EP-LN, we see, as expected, that the two algorithms are effectively equivalent, as shown in Table \ref{TabEPNILN}. The relative error in estimates of $\hat{\boldsymbol{\beta}}, \hat{\boldsymbol{\Sigma}}$ between EP-NI and EP-LN are in the range $10^{-8} - 10^{-9}$ for elements of $\hat{\boldsymbol{\beta}}$,  and $10^{-6} - 10^{-8}$ for elements of $\hat{\boldsymbol{\Sigma}}$. 

\begin{figure}[h]
	\centering
	\includegraphics[width=0.8\linewidth]{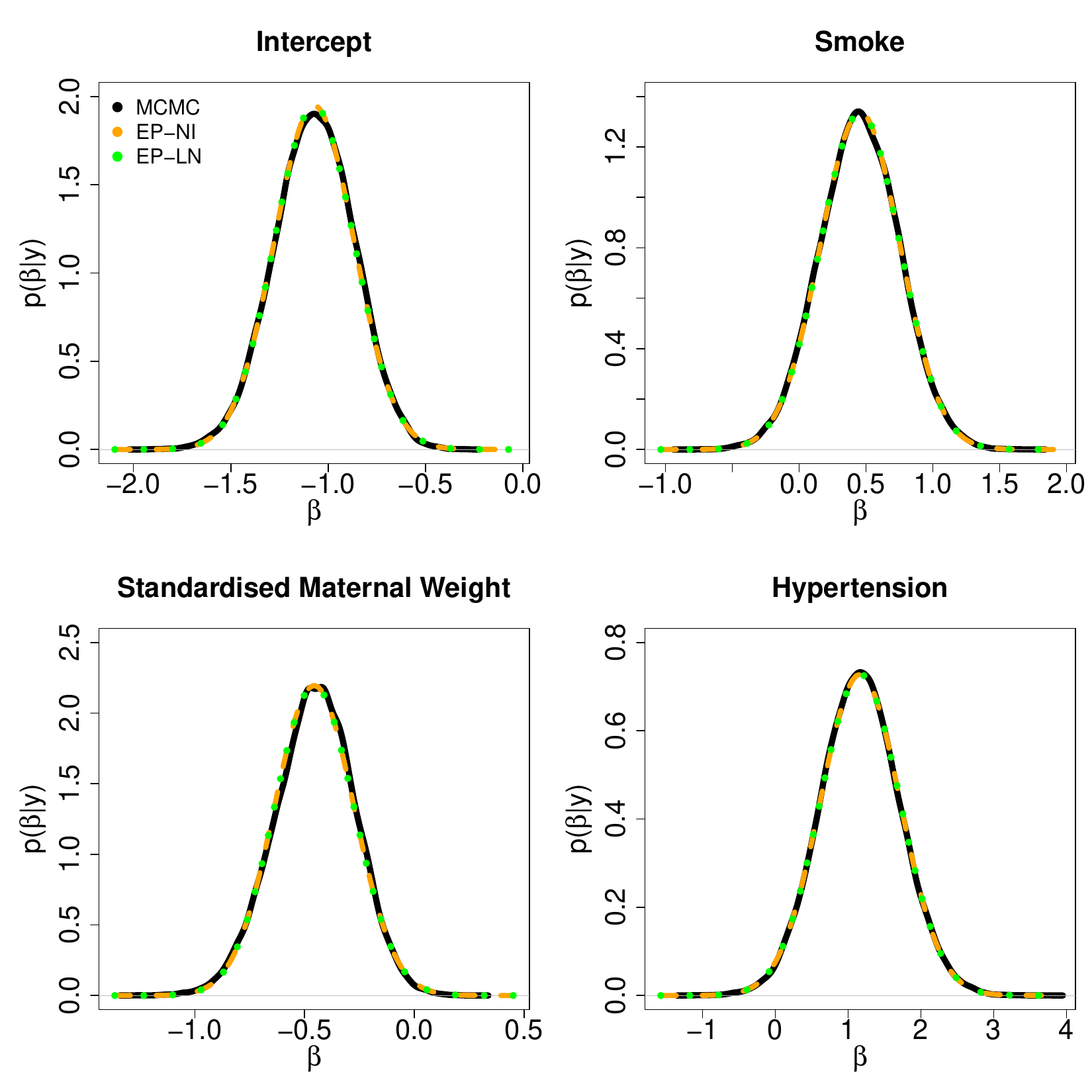}
	\caption{(Approximate) posterior densities estimated using MCMC, EP-NI, EP-LN for a logistic regression fitted to low birth weight data.}\label{Fig1e}
\end{figure}

Lastly, to benchmark the computational speed gain obtained by using EP-LN over EP-NI, we ran both versions of EP 1,000 times, with timings determined using the R package \texttt{microbenchmark} \cite{Mersmann2025}. Results are provided in Table \ref{TabspeedEPNILN}. This shows that for a sample of size 189, replacing numerical integration with our approximations of logit-normal moments reduced computation time by about 58 \%.

\begin{center}
\begin{table}[h]
	\centering\caption{Quantiles of time taken to fit logistic regression to the low birth weight data, assuming a normal prior using expectation propagation. EP-NI implements the sequential EP algorithm as described in \citet{Gelman2014}. EP-LN implements a modified version where the moment matching step is solved using functions of logit-normal moments rather than numerical integration. Times are reported in $10^{-3}$ seconds.}\label{TabspeedEPNILN}
	\begin{tabular*}{\textwidth}{@{\extracolsep\fill}lcc cc cc cc cc@{\extracolsep\fill}}\toprule
		Quantile & \multicolumn{2}{c}{10 \%} & \multicolumn{2}{c}{25 \%} & \multicolumn{2}{c}{50 \%} &  \multicolumn{2}{c}{75 \%}  & \multicolumn{2}{c}{90 \%}   \\
		& EP-NI & EP-LN 	& EP-NI & EP-LN  	& EP-NI & EP-LN  	& EP-NI & EP-LN 	& EP-NI & EP-LN \\
		\midrule 
 times $10^{-3} s$ & 612 & 256  & 631 & 264  & 663 & 277  &  714 & 301  & 790 & 335   \\
		\bottomrule
	\end{tabular*}
\end{table}
\end{center}

\newpage

\section{Conclusions}

We would argue  this paper fulfils the conclusion of \citet{Holmes2022a} that \lq the possibility of either analytic positive integer moments, or at least a highly accurate and simple approximation could be determined, through further study that utilises the similarity between the log-normal and logit-normal distributions\rq \hspace{0.5 mm}. Through examples, we have shown our estimates of moments up to order 4 do not suffer the numerical issues that the Mordell integral based representation of $E(P)$ from \citet{Holmes2022a} has when $\sigma$ is small, and thus maintain a high level of accuracy in all settings considered. 

We have shown our method for estimating low order moments, up to the fourth, of logit-normal random variables is faster than standard numerical integration in R, especially if multiple moments are needed simultaneously. As error in the underlying approximating function accumulates with increasing moment order, a finite limit exists in practice on the highest moment we can calculate accurately. Referring to the implied approximating function of the $k^\text{th}$ power of the logistic function, this limit is around 7 or 8. However, the highest moment accurately approximated will also depend on $\mu, \sigma$, as error in the underlying function $p(x)^k$ first appears for $x \approx 0$. 

The inability of our approach to estimate moments of arbitrarily high order means our work has limited application as an alternative to quadrature to finding the marginal likelihood in frequentist treatments of the logistic mixed model. However, while not initially obvious, we have demonstrated our approach to finding logit-normal moments does enable faster, numerical integration free implementation of expectation propagation for logistic regression. This is because only knowledge of the equivalent to logit-normal moments of order $1, 2, 3$ are required to solve the expectation propagation moment matching steps as described in \cite{Gelman2014} related to the likelihood in logistic regression. We do note expectation propagation algorithms for regression type problems with normal priors have been considered for finding the marginal likelihood for fairly general frequentist treatments of generalised linear mixed models \cite{Hall2019}. This indicates our illustrative application may have indirect frequentist as well as Bayesian algorithm uses.  Lastly, we suspect our results, apart from enabling faster implementation of expectation propagation, can be utilised in other versions of approximate Bayesian inference, such as variational Bayes methods for logistic and negative binomial regression. These applications are the focus of future work.   

\appendix

%\section{Supplementary material}

%In the separate supplementary materials document, the following can be found:
%\begin{itemize}
%	\item A proof of Proposition \ref{prop1}.
%	\item More detail on Proposition \ref{EP.LN2new}. \vspace{-2mm}
%	\item A proof of  Proposition  \ref{EP.LN2high}. \vspace{-2mm}
%	\item Tables demonstrating the accuracy of our approach for calculating logit-normal moments up to order 4. This includes skewness and kurtosis. \vspace{-2mm}
%	\item A description of the implementation of expectation propagation for logistic regression using logit-normal moments.
%	\item Code for calculating logit-normal moments.  
%\end{itemize} 

%For further code, see \url{https://github.com/johnbholmes/Logit-normal-moments-with-EP-applications}. This  contains R code for implementing expectation propagation, and a Gibbs sampler for MCMC based comparison.

%\backmatter
\bmsubsection*{Author Contributions}

John Holmes proposed the study and potential illustrative applications. Ness Arps was involved with parameter testing. Marco Reale was involved in manuscript refinement.

%\bmsubsection*{Acknowledgments}
%This is acknowledgment text. \cite{Kenamond2013} Provide text here. This is acknowledgment text. Provide text here. This is acknowledgment text. Provide text here. This is acknowledgment text. Provide text here. This is acknowledgment text. Provide text here. This is acknowledgment text. Provide text here. This is acknowledgment text. Provide text here. This is acknowledgment text. Provide text here. This is acknowledgment text. Provide text here.

\bmsubsection*{Financial Disclosure}

None reported.

\bmsubsection*{Conflicts of Interest}

The authors declare no conflicts of interest.

%\bibliography{wileyNJD-Chicago}
%\bibliography{../../Annals_of_Applied_statistics_template/References.bib}
%\bibliography{LNDistreference.bbl}
\bibliography{Inferential applications of moments of the logit-normal distribution - arxiv.bbl}
%\bmsubsection*{Supporting Information}

%Additional supporting information can be found online in the Supporting Information
%section.

\newpage

\appendix

\section{Expectation Propagation as described in \textbf{\textit{Bayesian Data Analysis}}}

One of the inferential applications we considered as motivational for researching logit-normal moments was Expectation Propagation (EP) \citep{Minka2001, Minka2001b} for logistic regression. Specifically, our motivation comes from the EP algorithm for a regression problem with fixed dispersion parameter described in \cite{Gelman2014}, which we outline below.

\subsection{EP for a regression problem with fixed dispersion parameter}

As we are considering a regression type problem, we assume the approximate posterior is normal $g(\boldsymbol{\beta}) := \mathcal{N}(\boldsymbol{\beta}|\hat{\boldsymbol{\beta}}, \hat{\boldsymbol{\Sigma}})$. This choice of $g$ does admit the factorisation $g(\boldsymbol{\beta}) = \prod_{i=0}^ng_i(\boldsymbol{\beta})$, which we need as the Kullback-Leibler divergence to be minimised is
\begin{equation} D_{KL}\{cg_{-i}(\boldsymbol{\beta})f_i(\boldsymbol{\beta})||g(\boldsymbol{\beta})\} = \int cg_{-i}(\boldsymbol{\beta})f_i(\boldsymbol{\beta}) \{\log(cg_{-i}(\boldsymbol{\beta})f_i(\boldsymbol{\beta}))-\log(g(\boldsymbol{\beta}))\} {\rm d}\boldsymbol{\beta}, \end{equation}
where $c$ is a normalising constant. 

The steps needed to perform this minimisation are given below. As we assume a normal prior, we let $g_0(\boldsymbol{\beta})$ be equal to the prior. Using a normal prior also means $g_0(\boldsymbol{\beta})$ will never need updating \citep{Chopin2017}, so moment matching only needs to be applied to the $n$ terms related to the likelihood. 
\begin{itemize}
	\item For $t = 1, \ldots,$
	\begin{itemize}
		\item For $i = 1, \ldots, n$
		\begin{itemize}
			\item[1] Compute the (natural) parameters of the cavity distribution $g_{-i}(\boldsymbol{\beta})$,
			
			\item[] $\hat{\boldsymbol{\Sigma}}_{-i}^{-1} = \hat{\boldsymbol{\Sigma}}^{-1} -\hat{\boldsymbol{\Sigma}}_i^{-1}$, $\hat{\boldsymbol{\Sigma}}_{-i}^{-1}\hat{\boldsymbol{\beta}}_{-i} =\hat{\boldsymbol{\Sigma}}^{-1}\hat{\boldsymbol{\beta}} - \hat{\boldsymbol{\Sigma}}_i^{-1}\hat{\boldsymbol{\beta}}_i \Rightarrow \hat{\boldsymbol{\beta}}_{-i} = \hat{\boldsymbol{\Sigma}}_{-i}(\hat{\boldsymbol{\Sigma}}^{-1}\hat{\boldsymbol{\beta}} - \hat{\boldsymbol{\Sigma}}_i^{-1}\hat{\boldsymbol{\beta}}_i)$. 
			
			\item[2]  Compute the parameters of $g_{-i}(\boldsymbol{\eta}_i)$, where $\boldsymbol{\eta}_i = {\bf x}_i'\hat{\boldsymbol{\beta}}$ and ${\bf x}_i$ are the covariate values, including intercept, for the $i^\text{th}$ observation, 
			
			\item[] $M_{-i} = {\bf x}_i'\hat{\boldsymbol{\beta}}_{-i}$, $V_{-i} = {\bf x}_i'\hat{\boldsymbol{\Sigma}}_{-i}{\bf x}_i$ 
			\item[3] Construct the un-normalised tilted distribution and calculate,
			\begin{eqnarray} E_k = \int_{-\infty}^{\infty} \boldsymbol{\eta}_i^k g_{-i}(\boldsymbol{\eta}_i)f_i(\boldsymbol{\eta}_i) d\boldsymbol{\eta}_i \quad\text{for $k = 0, 1, 2$}. \nonumber  \end{eqnarray} 
			Then moment match and set $M = \frac{E_1}{E_0}$ and $V = \frac{E_2}{E_0}-M^2$. Note in practice, we will need to specify finite bounds on the integral. The simplest choice would be $M_{-i} \pm \delta\sqrt{V_{-i}}$ for suitably large $\delta$.	 
			\item[4] Determine the natural parameters of $g_i(\boldsymbol{\eta}_i)$,
			
			\item[] $M_i/V_i = M/V - M_{-i}/V_{-i}$,\hspace{2 mm} $1/V_i = 1/V - 1/V_{-i}$.
			
			\item[5] Transform the natural parameters found in step 4 to those of $g_i(\boldsymbol{\beta})$,
			
			\item[] $\hat{\boldsymbol{\Sigma}}_{i}^{-1}\hat{\boldsymbol{\beta}}_{i} = {\bf x}_iM_i/V_i, \hat{\boldsymbol{\Sigma}}_{i}^{-1} = {\bf x}_i(1/V_i){\bf x}_i'$ 
			
			\item[6] Update the natural parameters of $g(\boldsymbol{\beta})$,
			
			\item[] $\hat{\boldsymbol{\Sigma}}^{-1} =  \hat{\boldsymbol{\Sigma}}_i^{-1}+ \hat{\boldsymbol{\Sigma}}_{-i}^{-1}$ and $\hat{\boldsymbol{\Sigma}}^{-1}\hat{\boldsymbol{\beta}}= \hat{\boldsymbol{\Sigma}}_i^{-1}\hat{\boldsymbol{\beta}}_i + \hat{\boldsymbol{\Sigma}}_{-i}^{-1}\hat{\boldsymbol{\beta}}_{-i}$.\\
			
		\end{itemize}
	\end{itemize}
	Stop once estimates have converged.
\end{itemize}

Some versions, known as parallel EP, skip step 6, updating $g(\boldsymbol{\beta})$ only after $g_i(\boldsymbol{\beta})$ has been updated for all $i$. In our example, we will implement sequential EP, which does not skip step 6.

\subsubsection{Proving Proposition 1}

Our interest for speeding up inference is focused on avoiding the numerical integrations in step 3,
\begin{eqnarray} E_k = \int_{-\infty}^{\infty} \boldsymbol{\eta}_i^k g_{-i}(\boldsymbol{\eta}_i)f_i(\boldsymbol{\eta}_i) d\boldsymbol{\eta}_i \quad\text{for $k = 0, 1, 2$}. \label{EQEK} \end{eqnarray}  
To achieve this for logistic regression, this requires re-writing (\ref{EQEK}) as shown in Proposition \ref{prop1}.

\begin{proof} As $g_{-i}(\boldsymbol{\eta}_i)$ is a normal density with mean ${\bf x}_i\hat{\boldsymbol{\beta}}$ and variance ${{\bf x}_i'\hat{\boldsymbol{\Sigma}}{\bf x}_i}$, the equations in (\ref{EQEK}) correspond to the expected value of some function with respect to a normal random variable. We will now apply Stein's Lemma \citep{Stein1972} (see Result \ref{hxExpect0} from main paper) to rewrite (\ref{EQEK}) in terms of derivatives of $f_i(\boldsymbol{\eta}_i)$. We will assume that $f_i(\boldsymbol{\eta}_i)$ is twice differentiable.

	For $E_0$, no change is required, as $E_0 = E(f_i(\boldsymbol{\eta}_i))$. To re-write $E_1$, we just re-arrange Result \ref{hxExpect0} to $E(Xh(X))= \sigma^2E(h'(X)) + \mu E(h(X))$. Hence 
	\begin{eqnarray}
		E_1 = M_{-i}E(f_i(\boldsymbol{\eta}_i)) + V_{-i}E(f_i'(\boldsymbol{\eta}_i)) 
	\end{eqnarray}
	
	Re-writing $E_2$ is more complicated. First let $\tilde{h}(X) = Xh(X)$, and apply Result \ref{hxExpect0}, noting $\tilde{h}'(X) = h(X) + Xh'(X)$,
	\begin{eqnarray} \sigma^2E((Xh(X))') &=& E((X-\mu)Xh(X)) \nonumber \\
		\sigma^2E(h(X)) + \sigma^2E(Xh'(X)) &=& E(X^2h(X)) - \mu E(Xh(X)). \label{eqp2} 	
	\end{eqnarray}	
	
	Next substitute the formula for $E_1$ into parts of both the left and right hand side of (\ref{eqp2}), but with $h'(X)$ replacing $h(X)$ in the substitution on the left hand side.
	\begin{eqnarray} 
		\sigma^2E(h(X)) + \sigma^2E(Xh'(X)) &=& \sigma^2E(h(X)) + \sigma^2(\sigma^2E(h''(X)) + \mu E(h'(X)) ) \nonumber \\
		&=& \sigma^2E(h(X)) + \sigma^4E(h''(X)) + \sigma^2\mu E(h'(X)) ) \label{eqp3} \\
		E(X^2h(X)) - \mu E(Xh(X)) &=& 	E(X^2h(X)) - \mu (\sigma^2E(h'(X)) + \mu E(h(X))) \nonumber \\
		&=& E(X^2h(X)) - \mu\sigma^2E(h'(X)) - \mu^2 E(h(X)). \label{eqp4}
	\end{eqnarray}
	As (\ref{eqp3}-\ref{eqp4}) are equal by (\ref{eqp2}), the result for $E_2$ comes by re-arranging to make $E(X^2h(X))$ the subject,
	\begin{eqnarray}
		E(X^2h(X)) &=& (\mu^2 + \sigma^2)E(h(X)) + 2\mu\sigma^2E(h'(X)) + \sigma^4E(h''(X)), 
	\end{eqnarray}
	hence \begin{eqnarray}
		E_2 = (M_{-i}^2+ V_{-i})E(f_i(\boldsymbol{\eta}_i)) + 2M_{-i}V_{-i}E(f_i'(\boldsymbol{\eta}_i)) +  V_{-i}^2E(f_i''(\boldsymbol{\eta}_i)).
	\end{eqnarray}
\end{proof} 
Lastly, we can use result \ref{hxExpect} (see main paper),

This allows us to re-write $E_1, E_2$ as 
\begin{eqnarray}
	E_1 &=& M_{-i}E(f_i(\boldsymbol{\eta}_i)) + V_{-i}\frac{ {\rm d}E(f_i(\boldsymbol{\eta}_i))}{{\rm d}M_{-i}} \\ 
	E_2 &=& (M_{-i}^2+ V_{-i})E(f_i(\boldsymbol{\eta}_i)) + 2M_{-i}V_{-i}\frac{ {\rm d}E(f_i(\boldsymbol{\eta}_i))}{{\rm d}M_{-i}} +  V_{-i}^2\frac{ {\rm d}^2E(f_i(\boldsymbol{\eta}_i))}{{\rm d}M_{-i}^2}. \label{eqlat}
\end{eqnarray}

\subsubsection{Implications in the binary logistic regression case} 	

For a binary logistic regression, with $\boldsymbol{\eta}_i$ being normal, $f_i(\boldsymbol{\eta}_i)$ corresponds to ${\bf p}_i = (1+ e^{-\boldsymbol{\eta}_i})^{-1}$, a logit-normal random variable with location ${\bf x}_i\hat{\boldsymbol{\beta}}$ and scale $({\bf x}_i'\hat{\boldsymbol{\Sigma}}{\bf x}_i)^{0.5}$ parameters when $y_i = 1$. If $y_i = 0$,  $f_i(\boldsymbol{\eta}_i)= 1 - {\bf p}_i$. In addition, the logit-normal distribution has the symmetry, if $P \sim \text{logit}\mathcal{N}(\mu,\sigma^2)$ then $1- P \sim \text{logit}\mathcal{N}(-\mu,\sigma^2)$. Exploiting this symmetry, for binary logistic regression the integrals in (\ref{EQEK}) needed for moment matching can be compactly written as
\begin{eqnarray}
	E_0 = E({\bf p}_i), \quad E_1 = (2{\bf y}_i-1)E(\boldsymbol{\eta}_i{\bf p}_i), \quad E_2 = E(\boldsymbol{\eta}_i^2{\bf p}_i) \label{bineq3}
\end{eqnarray}	
where ${\bf p}_i \sim \text{logit}\mathcal{N}((2{\bf y}_i-1)M_{-i},V_{-i}), \boldsymbol{\eta}_i \sim \mathcal{N}((2{\bf y}_i-1)M_{-i},V_{-i})$. Applying the alternative form of $E_0, E_1, E_2$ given in Proposition \ref{prop1}, specifically (\ref{eqlat}), we have shown implementing EP moment matching for logistic regression just requires a logit-normal first moment along with its first two derivatives w.r.t. $(2{\bf y}_i-1)M_{-i}$. 

% Main text

\section{The expected value (first moment) of a logit-normal random variable}\label{IntroA}

We propose that the expected value of a logit-normal random variable with location parameter $\mu$ and scale parameter $\sigma$ is very accurately approximated by the result of Proposition \ref{EP.LN2new}. This result is obtained by determining the expected value of our proposed approximation of the logistic function given in (\ref{ramplogistic.H})

\begin{proof}
	To determine the result given in Proposition \ref{EP.LN2new}, start by replacing $(1+e^{-x})^{-1}$ with the result of (\ref{ramplogistic.H}). This leads to the following representation of $E(P)$:
	\begin{eqnarray}  E(P) &\approx&\frac{1}{\sigma\sqrt{2\pi}}\int_{-\infty}^{\infty} \bigg\{ \mathbb{1}_{x > 0} + (-1)^{\mathbb{1}_{x>0}}\sum_{i=1}^{N} ((-1)^{i-1}\mathbb{1}_{|x| \geq L}+ic_i\mathbb{1}_{|x| < L})e^{-i|x|}\bigg\} e^{-\frac{(x-\mu)^2}{2\sigma^2}} {\rm d} x, \nonumber \\
		&=& \frac{1}{\sigma\sqrt{2\pi}}\bigg\{ \int_{0}^{\infty} e^{-\frac{(x-\mu)^2}{2\sigma^2}} {\rm d} x + \sum_{i=1}^{N}(-1)^{i-1} \int_{-\infty}^{-L} e^{ix}e^{-\frac{(x-\mu)^2}{2\sigma^2}} {\rm d} x + \sum_{i=1}^{N}ic_i \int_{-L}^{0} e^{ix}e^{-\frac{(x-\mu)^2}{2\sigma^2}} {\rm d} x \nonumber \\
		&&- \sum_{i=1}^{N}(-1)^{i-1} \int_{L}^{\infty} e^{-ix}e^{-\frac{(x-\mu)^2}{2\sigma^2}} {\rm d} x - \sum_{i=1}^{N}ic_i \int_{0}^{L} e^{-ix}e^{-\frac{(x-\mu)^2}{2\sigma^2}} {\rm d} x\bigg\}. \label{eq3.2} \end{eqnarray}
	
	The five component integrals in (\ref{eq3.2}) are examples of $E(h(X)), X \sim \mathcal{N}(\mu,\sigma^2)$, that are known analytically. The specific terms are:
	\begin{itemize}	
		\item[1.] The probability $\Pr(X\geq 0)$ when $X$ is normal,
		\begin{eqnarray} \frac{1}{\sigma\sqrt{2\pi}}\int_{0}^{\infty} e^{-\frac{(x-\mu)^2}{2\sigma^2}} {\rm d} x = 1- \Phi((0-\mu)/\sigma) = \Phi(\mu/\sigma). \label{eq.int1} \end{eqnarray}
	\end{itemize}
	
	The remaining component integrals, denoted 2-5, are examples of a moment generating function of a normal distribution truncated at $(a, b)$ multiplied by $\Phi((b-\mu)/\sigma)-\Phi((a-\mu)/\sigma)$. This integral is given in (\ref{eq.logNtrunc}), 
	\begin{eqnarray} \frac{1}{\sigma\sqrt{2\pi}}\int_{a}^{b}e^{tx} e^{-\frac{(x-\mu)^2}{2\sigma^2}} {\rm d} x &=& \frac{1}{\sigma\sqrt{2\pi}}\int_{a}^{b} e^{-\frac{(x^2-2(\mu+t\sigma)x+\mu^2)}{2\sigma^2}} {\rm d} x \nonumber \\
		&=& \frac{1}{\sigma\sqrt{2\pi}}\int_{a}^{b} e^{-\frac{(x^2-2(\mu+t\sigma^2)x+(\mu+t\sigma^2)^2 -2t\mu\sigma^2 - t^2(\sigma^2)^2)}{2\sigma^2}} {\rm d} x \nonumber \\
		&=& e^{t\mu+t^2\sigma^2/2}\times \frac{1}{\sigma\sqrt{2\pi}}\int_{a}^{b} e^{-\frac{(x-(\mu+t\sigma^2))^2}{2\sigma^2}} {\rm d} x \nonumber \\
		&=& e^{t\mu+t^2\sigma^2/2}\{\Phi((b-\mu-t\sigma^2)/\sigma) -\Phi((a-\mu-t\sigma^2)/\sigma) \}. \label{eq.logNtrunc}\end{eqnarray}
	
	Hence the value of the four remaining definite integrals appearing in $E(P)$ are:
	\begin{eqnarray}2. \hspace{5.6 mm} \frac{1}{\sigma\sqrt{2\pi}}\int_{-\infty}^{-L} e^{ix}e^{-\frac{(x-\mu)^2}{2\sigma^2}} {\rm d} x &=& e^{i\mu+i^2\sigma^2/2}\{\Phi((-L-\mu-i\sigma^2)/\sigma) -\Phi((-\infty-\mu-i\sigma^2)/\sigma) \} \nonumber \\ &=& e^{i\mu+i^2\sigma^2/2}\Phi((-L-\mu-i\sigma^2)/\sigma).\label{eq.int2} \\
		3. \hspace{6.1 mm} \frac{1}{\sigma\sqrt{2\pi}}\int_{-L}^{-0} e^{ix}e^{-\frac{(x-\mu)^2}{2\sigma^2}} {\rm d} x &=& e^{i\mu+i^2\sigma^2/2}\{\Phi((0-\mu-i\sigma^2)/\sigma) -\Phi((-L-\mu-i\sigma^2)/\sigma) \} \nonumber \\
		&=& e^{i\mu+i^2\sigma^2/2}\{\Phi((-\mu-i\sigma^2)/\sigma) -\Phi((-L-\mu-i\sigma^2)/\sigma) \}. \label{eq.int3}\\
		4. \hspace{5 mm} \frac{1}{\sigma\sqrt{2\pi}}\int_{L}^{\infty} e^{-ix}e^{-\frac{(x-\mu)^2}{2\sigma^2}} {\rm d} x &=& e^{-i\mu+i^2\sigma^2/2}\{\Phi((\infty-\mu+i\sigma^2)/\sigma) -\Phi((L-\mu+i\sigma^2)/\sigma) \} \nonumber \\ 
		&=& e^{-i\mu+i^2\sigma^2/2}\{1-\Phi((L-\mu+i\sigma^2)/\sigma)\} \nonumber \\
		&=& e^{-i\mu+i^2\sigma^2/2}\Phi((-L+\mu-i\sigma^2)/\sigma). \label{eq.int4}\\
		5. \hspace{5.5 mm} \frac{1}{\sigma\sqrt{2\pi}}\int_{0}^{L} e^{-ix}e^{-\frac{(x-\mu)^2}{2\sigma^2}} {\rm d} x 
		&=& e^{-i\mu+i^2\sigma^2/2}\{\Phi((L-\mu+i\sigma^2)/\sigma) -\Phi((0-\mu+i\sigma^2)/\sigma) \} \nonumber \\
		&=& e^{-i\mu+i^2\sigma^2/2}\{ 1-\Phi((-L+\mu-i\sigma^2)/\sigma) -(1-\Phi((\mu-i\sigma^2)/\sigma)) \} \nonumber \\
		&=& e^{-i\mu+i^2\sigma^2/2}\{\Phi((\mu-i\sigma^2)/\sigma) -\Phi((-L+\mu-i\sigma^2)/\sigma) \}.\label{eq.int5} \end{eqnarray}
	
	Substituting the result of (\ref{eq.int1}) and (\ref{eq.int2}-\ref{eq.int5}) into (\ref{eq3.2}) gives the result stated in Proposition \ref{EP.LN2new}. 
	
\end{proof}

\section{The $k^\text{th}$ moment of a logit-normal random variable}\label{Intro2}

We propose the $k^\text{th}$ moment of a logit-normal random variable can be found using the result of Proposition \ref{EP.LN2high} (see main paper). We now will prove this result using induction. As part of the proof, we will use the formula for the expected value of the derivatives of $E(P)$ with respect to $\mu$ given in the main text. As a reminder, this showed that derivatives of the logistic function can be approximated by (\ref{ramplogistic.diff}) and the expected value of $\frac{{\rm d}^kE(P)}{{\rm d}\mu^k}$, can be approximated using the same results for component integrals from (\ref{eq.int2}-\ref{eq.int5}) as shown in (\ref{eq.EOdiff}).

\begin{proof}
1. Proving the result holds for $k = 2$. 
		
To find $E(P^2)$, we need $E(P)$ and $\frac{{\rm d}E(P)}{{\rm d}\mu}$. $E(P)$ was given in Proposition \ref{EP.LN2new}, and $\frac{{\rm d}E(P)}{{\rm d}\mu}$ can be found by setting $k=1$ in (\ref{eq.EOdiff}), and corresponds to
		\begin{eqnarray} \frac{{\rm d}E(P)}{{\rm d}\mu} &=&  \sum_{i=1}^{N} i(-1)^{i-1}\left\{e^{\mu i+\sigma^2i^2/2}\Phi\left(\frac{-L-\mu-\sigma^2 i}{\sigma}\right) +e^{-\mu i+\sigma^2i^2/2}\Phi\left(\frac{-L+\mu-\sigma^2 i}{\sigma}\right)  \right\} \nonumber \\
			&&+\sum_{i=1}^{N} i^2c_ie^{\mu i+\sigma^2i^2/2}\left\{\Phi\left(\frac{-\mu-\sigma^2 i}{\sigma}\right) - \Phi\left(\frac{-L-\mu-\sigma^2 i}{\sigma}\right)\right\}\nonumber \\
			&&+\sum_{i=1}^{N} i^{2}c_ie^{-\mu i+\sigma^2i^2/2}\left\{\Phi\left(\frac{\mu-\sigma^2 i}{\sigma}\right)-\Phi\left(\frac{-L+\mu-\sigma^2 i}{\sigma}\right) \right\}.  \label{eq.EO1diff}
		\end{eqnarray}
		
		Putting these together, $E(P^2)$ is approximately:
		\begin{eqnarray} E(P^2) &=& E(P)- \frac{{\rm d}E(P)}{{\rm d}\mu} \nonumber \\
			&\approx&  \Phi(\mu/\sigma) + \sum_{i=1}^{N} (-1)^{i-1}(1-i)e^{\mu i+\sigma^2i^2/2}\Phi\left(\frac{-L-\mu-\sigma^2 i}{\sigma}\right) \nonumber \\
			&&- \sum_{i=1}^{N} (-1)^{i-1}(1+i)e^{-\mu i+\sigma^2i^2/2}\Phi\left(\frac{-L+\mu-\sigma^2 i}{\sigma}\right)  \nonumber \\
			&&+\sum_{i=1}^{N} i(1-i)c_ie^{\mu i+\sigma^2i^2/2}\left\{\Phi\left(\frac{-\mu-\sigma^2 i}{\sigma}\right) - \Phi\left(\frac{-L-\mu-\sigma^2 i}{\sigma}\right)\right\}\nonumber \\
			&&-\sum_{i=1}^{N} i(1+i)c_ie^{-\mu i+\sigma^2i^2/2}\left\{\Phi\left(\frac{\mu-\sigma^2 i}{\sigma}\right)-\Phi\left(\frac{-L+\mu-\sigma^2 i}{\sigma}\right) \right\}  \label{eq.EP2}
		\end{eqnarray}	
		
		If $k=2$, then 	$\prod_{j=1}^{k-1}(j-i), \prod_{j=1}^{k-1}(j+i)$ are equal to $1-i$ and $1+i$ respectively, and ${1}/{(k-1)!} =1$. Hence  (\ref{eq.EP2}) satisfies Proposition \ref{EP.LN2high}.
		
2. Assume Proposition \ref{EP.LN2high} holds for $k$, and show this means $E(P^{k+1})$ must satisfy the proposition.
		
		To find $E(P^{k+1})$, we need $E(P^k)$ and ${{\rm d}E(P^k)}/{{\rm d}\mu}$. We assume $E(P^k)$ follows the result of Proposition \ref{EP.LN2high}. This leaves us needing to find ${{\rm d}E(P^k)}/{{\rm d}\mu}$. To find this derivative, we use the re-parametrisation trick, ${{\rm d}E(P^k)}/{{\rm d}\mu}= E({{\rm d}p(x)^k}/{{\rm d}x})$. 
		
		Undoing the expectation step for $P^k$, the underlying function $p(x)^k$ corresponds to
		\begin{eqnarray} p(x)^k &\approx& \mathbb{1}_{x>0} + \frac{(-1)^{\mathbb{1}_{x>0}}}{(k-1)!}\sum_{i=1}^{N} \bigg(\prod_{j=1}^{k-1}(j-i(-1)^{\mathbb{1}_{x>0}})\bigg)\{(-1)^{i-1}\mathbb{1}_{|x| \geq L} +ic_i\mathbb{1}_{|x| < L}\} e^{-i|x|},  \label{pk}
		\end{eqnarray}	
		which follows as 
		\begin{itemize}
			\item Negative terms before the sums in the expectation of $P^k$ correspond to integrals of a function w.r.t a normal density taken over positive real numbers, and positive terms to integrals of a function w.r.t a normal density taken over negative real numbers. Hence $(-1)^{\mathbb{1}_{x>0}}$ appears before the sum in the underlying function.
			\item The product term contained $j+i$ when the integral was over positive real numbers, and $j-i$ when the integral was over negative real numbers. Hence $\prod_{j=1}^{k-1}(j-i(-1)^{\mathbb{1}_{x>0}})$ must appear in the underlying function.
			\item The term $\Phi(\mu/\sigma)$ implies an integration of a normal density over all positive real numbers. Hence $\mathbb{1}_{x>0}$ appears in the underlying function. 
		\end{itemize}
		
		Having determined $p(x)^k$, ${{\rm d}p(x)^k}/{{\rm d}x}$ is
		\begin{eqnarray} \frac{{\rm d}p(x)^k}{{\rm d}x} &\approx&  \frac{1}{(k-1)!}\sum_{i=1}^{N}i \bigg(\prod_{j=1}^{k-1}(j-i(-1)^{\mathbb{1}_{x>0}})\bigg)\{(-1)^{i-1}\mathbb{1}_{|x| \geq L} +ic_i\mathbb{1}_{|x| < L}\} e^{-i|x|},  \label{Epk}
		\end{eqnarray}	 
		where $(-1)^{\mathbb{1}_{x>0}}$ disappears as the derivative of $e^{-i|x|}$ is $ie^{-i|x|}$ if $x \leq 0$ and $-ie^{-i|x|}$ if $x \geq 0$. Note: while strictly speaking $e^{-i|x|}$ would not be differentiable at 0, our function is designed as piecewise with transition point at 0.
		
		Hence, ${{\rm d}E(P^k)}/{{\rm d}\mu}$ can be evaluated using the component integrals (\ref{eq.int2}-\ref{eq.int5}), first used to determine $E(P)$ as
		\begin{eqnarray} \frac{{\rm d}E(P^k)}{{\rm d}\mu} &=&  \frac{1}{(k-1)!}\bigg(\sum_{i=1}^{N} i(-1)^{i-1}\prod_{j=1}^{k-1}(j-i)e^{\mu i+\sigma^2i^2/2}\Phi\left(\frac{-L-\mu-\sigma^2 i}{\sigma}\right) \nonumber \\
			&& +\sum_{i=1}^{N} i(-1)^{i-1}\prod_{j=1}^{k-1}(j+i)e^{-\mu i+\sigma^2i^2/2}\Phi\left(\frac{-L+\mu-\sigma^2 i}{\sigma}\right)  \nonumber \\
			&&+\sum_{i=1}^{N} i^{2}c_i\prod_{j=1}^{k-1}(j-i)e^{\mu i+\sigma^2i^2/2}\left\{\Phi\left(\frac{-\mu-\sigma^2 i}{\sigma}\right) - \Phi\left(\frac{-L-\mu-\sigma^2 i}{\sigma}\right)\right\}\nonumber \\
			&&+\sum_{i=1}^{N} i^{2}c_i\prod_{j=1}^{k-1}(j+i)e^{-\mu i+\sigma^2i^2/2}\left\{\Phi\left(\frac{\mu-\sigma^2 i}{\sigma}\right)-\Phi\left(\frac{-L+\mu-\sigma^2 i}{\sigma}\right) \right\}\bigg).  \label{eq.EPkdiff}
		\end{eqnarray}  
		
		Lastly, we know by the recursion formula (see main text) that for a logit-normal random variable, $	E(P^{k+1})= E(P^k) - \frac{1}{k}	\frac{{\rm d}E(P^k)}{{\rm d}\mu}$. This means that $E(P^{k+1})$ is 
		\begin{eqnarray}
			E(P^{k+1}) &=& E(P^k) - \frac{1}{k}	\frac{{\rm d}E(P^k)}{{\rm d}\mu} \nonumber \\
			&\approx&  \Phi(\mu/\sigma) + \frac{1}{(k-1)!}\bigg(\sum_{i=1}^{N} (-1)^{i-1}\prod_{j=1}^{k-1}(j-i) e^{\mu i+\sigma^2i^2/2}\Phi\left(\frac{-L-\mu-\sigma^2 i}{\sigma}\right) \nonumber \\ 
			&&-\sum_{i=1}^{N} (-1)^{i-1}\prod_{j=1}^{k-1}(j+i)e^{-\mu i+\sigma^2i^2/2}\Phi\left(\frac{-L+\mu-\sigma^2 i}{\sigma}\right)  \nonumber \\
			&&+\sum_{i=1}^{N} ic_i\prod_{j=1}^{k-1}(j-i)e^{\mu i+\sigma^2i^2/2}\left\{\Phi\left(\frac{-\mu-\sigma^2 i}{\sigma}\right) - \Phi\left(\frac{-L-\mu-\sigma^2 i}{\sigma}\right)\right\}\nonumber \\
			&&-\sum_{i=1}^{N} ic_i\prod_{j=1}^{k-1}(j+i)e^{-\mu i+\sigma^2i^2/2}\left\{\Phi\left(\frac{\mu-\sigma^2 i}{\sigma}\right)-\Phi\left(\frac{-L+\mu-\sigma^2 i}{\sigma}\right) \right\}\bigg) \nonumber \\
			&& - \frac{1}{k}\frac{1}{(k-1)!}\bigg(\sum_{i=1}^{N} i(-1)^{i-1}\prod_{j=1}^{k-1}(j-i)e^{\mu i+\sigma^2i^2/2}\Phi\left(\frac{-L-\mu-\sigma^2 i}{\sigma}\right) \nonumber \\
			&& +\sum_{i=1}^{N} i(-1)^{i-1}\prod_{j=1}^{k-1}(j+i)e^{-\mu i+\sigma^2i^2/2}\Phi\left(\frac{-L+\mu-\sigma^2 i}{\sigma}\right)  \nonumber \\
			&&+\sum_{i=1}^{N} i^{2}c_i\prod_{j=1}^{k-1}(j-i)e^{\mu i+\sigma^2i^2/2}\left\{\Phi\left(\frac{-\mu-\sigma^2 i}{\sigma}\right) - \Phi\left(\frac{-L-\mu-\sigma^2 i}{\sigma}\right)\right\}\nonumber \\
			&&+\sum_{i=1}^{N} i^{2}c_i\prod_{j=1}^{k-1}(j+i)e^{-\mu i+\sigma^2i^2/2}\left\{\Phi\left(\frac{\mu-\sigma^2 i}{\sigma}\right)-\Phi\left(\frac{-L+\mu-\sigma^2 i}{\sigma}\right) \right\}\bigg). \label{eqEPK}
		\end{eqnarray}
		
		Grouping like terms in (\ref{eqEPK}), we just need to simplify,
		\begin{eqnarray}
			\frac{1}{(k-1)!}\prod_{j=1}^{k-1}(j-i) - \frac{1}{k!}i\prod_{j=1}^{k-1}(j-i) &=& \frac{k-i}{k!}\prod_{j=1}^{k-1}(j-i) = \frac{1}{k!}\prod_{j=1}^{k}(j-i)  \\
			-\frac{1}{(k-1)!}\prod_{j=1}^{k-1}(j+i) - \frac{1}{k!}i\prod_{j=1}^{k-1}(j+i) &=& \frac{-k-i}{k!}\prod_{j=1}^{k-1}(j+i) = -\frac{1}{k!}\prod_{j=1}^{k}(j+i),  
		\end{eqnarray}
		giving us  
		\begin{eqnarray}
			E(P^{k+1}) 
			&\approx&  \Phi(\mu/\sigma) + \frac{1}{k!}\bigg(\sum_{i=1}^{N} (-1)^{i-1}\prod_{j=1}^{k}(j-i) e^{\mu i+\sigma^2i^2/2}\Phi\left(\frac{-L-\mu-\sigma^2 i}{\sigma}\right) \nonumber \\ 
			&&-\sum_{i=1}^{N} (-1)^{i-1}\prod_{j=1}^{k}(j+i)e^{-\mu i+\sigma^2i^2/2}\Phi\left(\frac{-L+\mu-\sigma^2 i}{\sigma}\right)  \nonumber \\
			&&+\sum_{i=1}^{N} ic_i\prod_{j=1}^{k}(j-i)e^{\mu i+\sigma^2i^2/2}\left\{\Phi\left(\frac{-\mu-\sigma^2 i}{\sigma}\right) - \Phi\left(\frac{-L-\mu-\sigma^2 i}{\sigma}\right)\right\}\nonumber \\
			&&-\sum_{i=1}^{N} ic_i\prod_{j=1}^{k}(j+i)e^{-\mu i+\sigma^2i^2/2}\left\{\Phi\left(\frac{\mu-\sigma^2 i}{\sigma}\right)-\Phi\left(\frac{-L+\mu-\sigma^2 i}{\sigma}\right) \right\}\bigg)  \label{eqEPKfinal}
		\end{eqnarray}
		which corresponds to Proposition \ref{EP.LN2high}, if $k$ is replaced with $k+1$. Hence by induction Proposition  \ref{EP.LN2high} holds for all $k \geq 2$.

\end{proof}

\section{Testing accuracy of approximation for higher order moments}\label{CheckingResults}

In the main text, we showed that our approach provided accurate estimates of the expected value and variance of a logit-normal random variable for a wide range of values of $\mu, \sigma$. In this section, we will show our approximations maintain accuracy up to at least the fourth moment. This is unsurprising as our approximation of the underlying functions is accurate up to the third derivative of the logistic function. We also include calculations for skewness and kurtosis. This have been put in to highlight when the numerator moments are divided by a variance near zero, numerical instability will result, even if the component moments are well estimated.

\subsection{Second, third and fourth moments}  

\begin{center}
\begin{table}[h!]\caption{Estimates of $E(P^2)$ obtained from piecewise approximation of $p(x)$ with error compared to estimate of $E(P^2)$ obtained from numerical integration.}\label{tab.accEP2}
	\scriptsize
	\begin{tabular*}{\textwidth}{@{\extracolsep\fill}lcr cr cr cr cr@{\extracolsep\fill}} \toprule
		$\sigma$	&	$E(P^2)$	&		Error 			&	$E(P^2)$	&		Error 	&	$E(P^2)$	&		Error 		&	$E(P^2)$	&		Error 		&	$E(P^2)$	&		Error 		\\ 
		\midrule
		& \multicolumn{2}{c}{$\mu = 0$} & \multicolumn{2}{c}{$\mu = 1.25$} & \multicolumn{2}{c}{$\mu = 2.5$} & \multicolumn{2}{c}{$\mu = 3.75$} & \multicolumn{2}{c}{$\mu = 5.00$}  \\
		0.001	&	0.250000	&	$	-7.1\times 10^{-8}$	&	0.604195	&	$	4.4\times 10^{-10}$	&	0.854038	&	$	5.2\times 10^{-12}$	&	0.954573	&	$	7.1\times 10^{-12}$	&	0.986659	&	$	7.5\times 10^{-12}$	\\ 
		0.010	&	0.250006	&	$	-4.5\times 10^{-8}$	&	0.604191	&	$	4.4\times 10^{-10}$	&	0.854033	&	$	5.2\times 10^{-12}$	&	0.954571	&	$	7.1\times 10^{-12}$	&	0.986658	&	$	7.3\times 10^{-12}$	\\ 
		0.100	&	0.250622	&	$	-7.2\times 10^{-9}$	&	0.603751	&	$	-4.2\times 10^{-10}$	&	0.853538	&	$	4.0\times 10^{-12}$	&	0.954369	&	$	7.1\times 10^{-12}$	&	0.986594	&	$	7.3\times 10^{-12}$	\\ 
		0.250	&	0.253789	&	$	-3.3\times 10^{-9}$	&	0.601501	&	$	-8.4\times 10^{-9}$	&	0.850924	&	$	-7.2\times 10^{-11}$	&	0.953285	&	$	7.1\times 10^{-12}$	&	0.986249	&	$	7.3\times 10^{-12}$	\\ 
		0.500	&	0.263956	&	$	-1.9\times 10^{-9}$	&	0.594438	&	$	-8.0\times 10^{-9}$	&	0.841786	&	$	-2.0\times 10^{-9}$	&	0.949288	&	$	6.0\times 10^{-12}$	&	0.984951	&	$	7.4\times 10^{-12}$	\\ 
		1.000	&	0.293379	&	$	-3.5\times 10^{-9}$	&	0.575671	&	$	-5.2\times 10^{-9}$	&	0.809342	&	$	-3.4\times 10^{-9}$	&	0.931887	&	$	-4.7\times 10^{-10}$	&	0.978699	&	$	-8.0\times 10^{-12}$	\\ 
		1.500	&	0.323272	&	$	-4.1\times 10^{-9}$	&	0.559314	&	$	-4.1\times 10^{-9}$	&	0.769259	&	$	-3.0\times 10^{-9}$	&	0.902092	&	$	-1.2\times 10^{-9}$	&	0.964992	&	$	-2.5\times 10^{-10}$	\\ 
		2.500	&	0.368925	&	$	-3.4\times 10^{-9}$	&	0.539014	&	$	-3.2\times 10^{-9}$	&	0.700705	&	$	-2.5\times 10^{-9}$	&	0.828797	&	$	-1.7\times 10^{-9}$	&	0.913925	&	$	-8.9\times 10^{-10}$	\\ 
		\midrule
		&	\multicolumn{2}{c}{$\mu	=		0.25$}		&	\multicolumn{2}{c}{$\mu	=		1.5$}		&	\multicolumn{2}{c}{$\mu	=		2.75$}		&	\multicolumn{2}{c}{$\mu	=		4.0$}		&	\multicolumn{2}{c}{$\mu		=		5.25$}	\\
		0.001	&	0.316043	&	$	1.2\times 10^{-10}$	&	0.668428	&	$	-4.6\times 10^{-9	}$	&	0.883437	&	$	6.5\times 10^{-12}$	&	0.964351	&	$	7.2\times 10^{-12}$	&	0.989587	&	$	7.5\times 10^{-12}$	\\ 
		0.010	&	0.316047	&	$	1.1\times 10^{-10}$	&	0.668423	&	$	-3.5\times 10^{-8	}$	&	0.883433	&	$	6.5\times 10^{-12}$	&	0.964349	&	$	7.1\times 10^{-12}$	&	0.989587	&	$	7.3\times 10^{-12}$	\\ 
		0.100	&	0.316474	&	$	-9.8\times 10^{-10}$	&	0.667878	&	$	-3.4\times 10^{-8	}$	&	0.883002	&	$	6.4\times 10^{-12}$	&	0.964187	&	$	7.1\times 10^{-12}$	&	0.989536	&	$	7.3\times 10^{-12}$	\\ 
		0.250	&	0.318679	&	$	-2.1\times 10^{-9}$	&	0.665071	&	$	-1.8\times 10^{-8	}$	&	0.880719	&	$	1.9\times 10^{-12}$	&	0.963315	&	$	7.1\times 10^{-12}$	&	0.989265	&	$	7.3\times 10^{-12}$	\\ 
		0.500	&	0.325812	&	$	-1.8\times 10^{-9}$	&	0.656012	&	$	-9.8\times 10^{-9	}$	&	0.872619	&	$	-7.2\times 10^{-10}$	&	0.960079	&	$	7.2\times 10^{-12}$	&	0.988242	&	$	7.4\times 10^{-12}$	\\ 
		1.000	&	0.346812	&	$	-3.7\times 10^{-9}$	&	0.630256	&	$	-5.3\times 10^{-9	}$	&	0.842416	&	$	-2.6\times 10^{-9}$	&	0.945599	&	$	-2.6\times 10^{-10}$	&	0.983263	&	$	1.3\times 10^{-12}$	\\ 
		1.500	&	0.368515	&	$	-4.2\times 10^{-9}$	&	0.605885	&	$	-4.0\times 10^{-9	}$	&	0.802489	&	$	-2.6\times 10^{-9}$	&	0.919330	&	$	-9.4\times 10^{-10}$	&	0.971945	&	$	-1.7\times 10^{-10}$	\\ 
		2.500	&	0.402108	&	$	-3.4\times 10^{-9}$	&	0.573000	&	$	-3.1\times 10^{-9	}$	&	0.729570	&	$	-2.4\times 10^{-9}$	&	0.849251	&	$	-1.5\times 10^{-9}$	&	0.926121	&	$	-7.6\times 10^{-10}$	\\  \midrule
		&	\multicolumn{2}{c}{$\mu	=		0.5$}		&	\multicolumn{2}{c}{$\mu	=		1.75$}		&	\multicolumn{2}{c}{$\mu	=		3.0$}		&	\multicolumn{2}{c}{$\mu	=		4.25$}		&	\multicolumn{2}{c}{$\mu		=		5.5$}	\\
		0.001	&	0.387456	&	$	5.8\times 10^{-10}$	&	0.725824	&	$	-8.3\times 10^{-9}$	&	0.907397	&	$	6.7\times 10^{-12}$	&	0.972071	&	$	7.4\times 10^{-12}$	&	0.991876	&	$	7.5\times 10^{-12}$	\\ 
		0.010	&	0.387458	&	$	5.7\times 10^{-10}$	&	0.725818	&	$	-8.4\times 10^{-9}$	&	0.907394	&	$	6.7\times 10^{-12}$	&	0.972069	&	$	7.2\times 10^{-12}$	&	0.991876	&	$	7.3\times 10^{-12}$	\\ 
		0.100	&	0.387649	&	$	1.6\times 10^{-11}$	&	0.725228	&	$	-1.5\times 10^{-8}$	&	0.907028	&	$	6.7\times 10^{-12}$	&	0.971939	&	$	7.2\times 10^{-12}$	&	0.991836	&	$	7.4\times 10^{-12}$	\\ 
		0.250	&	0.388653	&	$	-6.0\times 10^{-10}$	&	0.722166	&	$	-1.6\times 10^{-8}$	&	0.905084	&	$	6.4\times 10^{-12}$	&	0.971243	&	$	7.2\times 10^{-12}$	&	0.991623	&	$	7.4\times 10^{-12}$	\\ 
		0.500	&	0.392049	&	$	-2.1\times 10^{-9}$	&	0.712052	&	$	-9.4\times 10^{-9}$	&	0.898100	&	$	-2.1\times 10^{-10}$	&	0.968645	&	$	7.3\times 10^{-12}$	&	0.990820	&	$	7.4\times 10^{-12}$	\\ 
		1.000	&	0.403041	&	$	-4.0\times 10^{-9}$	&	0.681546	&	$	-5.1\times 10^{-9}$	&	0.870854	&	$	-1.9\times 10^{-9}$	&	0.956745	&	$	-1.4\times 10^{-10}$	&	0.986873	&	$	5.2\times 10^{-12}$	\\ 
		1.500	&	0.415449	&	$	-4.2\times 10^{-9}$	&	0.650656	&	$	-3.9\times 10^{-9}$	&	0.832328	&	$	-2.3\times 10^{-9}$	&	0.933962	&	$	-7.1\times 10^{-10}$	&	0.977610	&	$	-1.1\times 10^{-10}$	\\ 
		2.500	&	0.435962	&	$	-3.4\times 10^{-9}$	&	0.606385	&	$	-3.0\times 10^{-9}$	&	0.756900	&	$	-2.2\times 10^{-9}$	&	0.867946	&	$	-1.3\times 10^{-9}$	&	0.936904	&	$	-6.5\times 10^{-10}$	\\ 
		\midrule
		&	\multicolumn{2}{c}{$\mu	=		0.75$}		&	\multicolumn{2}{c}{$\mu	=		2.0$}		&	\multicolumn{2}{c}{$\mu	=		3.25$}		&	\multicolumn{2}{c}{$\mu	=		4.5$}		&	\multicolumn{2}{c}{$\mu		=		5.75$}	\\
		0.001	&	0.461284	&	$	-1.4\times 10^{-10}$	&	0.775803	&	$	-4.2\times 10^{-10}$	&	0.926740	&	$	6.9\times 10^{-12}$	&	0.978147	&	$	7.4\times 10^{-12}$	&	0.993665	&	$	7.5\times 10^{-12}$	\\ 
		0.010	&	0.461283	&	$	-1.4\times 10^{-10}$	&	0.775798	&	$	-4.3\times 10^{-10}$	&	0.926736	&	$	6.8\times 10^{-12}$	&	0.978146	&	$	7.2\times 10^{-12}$	&	0.993664	&	$	7.3\times 10^{-12}$	\\ 
		0.100	&	0.461229	&	$	-5.3\times 10^{-11}$	&	0.775211	&	$	-8.6\times 10^{-10}$	&	0.926432	&	$	6.9\times 10^{-12}$	&	0.978043	&	$	7.2\times 10^{-12}$	&	0.993633	&	$	7.4\times 10^{-12}$	\\ 
		0.250	&	0.460978	&	$	-1.8\times 10^{-10}$	&	0.772143	&	$	-5.6\times 10^{-9}$	&	0.924809	&	$	6.8\times 10^{-12}$	&	0.977489	&	$	7.2\times 10^{-12}$	&	0.993467	&	$	7.4\times 10^{-12}$	\\ 
		0.500	&	0.460443	&	$	-3.1\times 10^{-9}$	&	0.761797	&	$	-7.1\times 10^{-9}$	&	0.918918	&	$	-4.2\times 10^{-11}$	&	0.975418	&	$	7.4\times 10^{-12}$	&	0.992837	&	$	7.4\times 10^{-12}$	\\ 
		1.000	&	0.460811	&	$	-4.5\times 10^{-9}$	&	0.728780	&	$	-4.7\times 10^{-9}$	&	0.894970	&	$	-1.3\times 10^{-9}$	&	0.965738	&	$	-6.9\times 10^{-11}$	&	0.989719	&	$	6.7\times 10^{-12}$	\\ 
		1.500	&	0.463380	&	$	-4.2\times 10^{-9}$	&	0.693082	&	$	-3.6\times 10^{-9}$	&	0.858792	&	$	-1.9\times 10^{-9}$	&	0.946264	&	$	-5.2\times 10^{-10}$	&	0.982197	&	$	-6.9\times 10^{-11}$	\\ 
		2.500	&	0.470236	&	$	-3.3\times 10^{-9}$	&	0.638936	&	$	-2.8\times 10^{-9}$	&	0.782592	&	$	-2.0\times 10^{-9}$	&	0.884917	&	$	-1.2\times 10^{-9}$	&	0.946376	&	$	-5.5\times 10^{-10}$	\\ \midrule
		&	\multicolumn{2}{c}{$\mu	=		1.0$}		&	\multicolumn{2}{c}{$\mu	=		2.25$}		&	\multicolumn{2}{c}{$\mu	=		3.5$}		&	\multicolumn{2}{c}{$\mu	=		4.75$}		&	\multicolumn{2}{c}{$\mu		=		6.0$}	\\
		0.001	&	0.534447	&	$	-1.1\times 10^{-10}$	&	0.818393	&	$	-1.6\times 10^{-11}$	&	0.942235	&	$	7.0\times 10^{-12}$	&	0.982919	&	$	7.4\times 10^{-12}$	&	0.995061	&	$	7.5\times 10^{-12}$	\\ 
		0.010	&	0.534444	&	$	-1.1\times 10^{-10}$	&	0.818387	&	$	-1.6\times 10^{-11}$	&	0.942232	&	$	7.0\times 10^{-12}$	&	0.982918	&	$	7.3\times 10^{-12}$	&	0.995061	&	$	7.4\times 10^{-12}$	\\ 
		0.100	&	0.534171	&	$	-3.3\times 10^{-11}$	&	0.817836	&	$	-3.8\times 10^{-11}$	&	0.941983	&	$	7.0\times 10^{-12}$	&	0.982836	&	$	7.3\times 10^{-12}$	&	0.995036	&	$	7.4\times 10^{-12}$	\\ 
		0.250	&	0.532792	&	$	-1.5\times 10^{-9}$	&	0.814942	&	$	-8.9\times 10^{-10}$	&	0.940649	&	$	7.0\times 10^{-12}$	&	0.982398	&	$	7.3\times 10^{-12}$	&	0.994906	&	$	7.4\times 10^{-12}$	\\ 
		0.500	&	0.528640	&	$	-5.3\times 10^{-9}$	&	0.804991	&	$	-4.2\times 10^{-9}$	&	0.935763	&	$	-1.8\times 10^{-12}$	&	0.980755	&	$	7.4\times 10^{-12}$	&	0.994413	&	$	7.4\times 10^{-12}$	\\ 
		1.000	&	0.518791	&	$	-4.9\times 10^{-9}$	&	0.771459	&	$	-4.1\times 10^{-9}$	&	0.915164	&	$	-8.3\times 10^{-10}$	&	0.972949	&	$	-2.9\times 10^{-11}$	&	0.991958	&	$	7.3\times 10^{-12}$	\\ 
		1.500	&	0.511579	&	$	-4.2\times 10^{-9}$	&	0.732723	&	$	-3.3\times 10^{-9}$	&	0.881989	&	$	-1.5\times 10^{-9}$	&	0.956517	&	$	-3.7\times 10^{-10}$	&	0.985890	&	$	-1.3\times 10^{-9}$	\\ 
		2.500	&	0.504674	&	$	-3.3\times 10^{-9}$	&	0.670439	&	$	-2.7\times 10^{-9}$	&	0.806572	&	$	-1.8\times 10^{-9}$	&	0.900219	&	$	-1.0\times 10^{-9}$	&	0.954646	&	$	-4.6\times 10^{-10}$ \\
		\bottomrule
	\end{tabular*}
\end{table}
\end{center}

\subsubsection{Third moment}	
\begin{center}	
\begin{table}[h!]\caption{Estimates of $E(P^3)$ obtained from piecewise approximation of $p(x)$ with error compared to estimate of $E(P^3)$ obtained from numerical integration.}\label{tab.accEP3}
	\scriptsize
	\begin{tabular*}{\textwidth}{@{\extracolsep\fill}lcr cr cr cr cr@{\extracolsep\fill}} \toprule
		$\sigma$	&	$E(P^3)$	&		Error 			&	$E(P^3)$	&		Error 	&	$E(P^3)$	&		Error 		&	$E(P^3)$	&		Error 		&	$E(P^3)$	&		Error 		\\ 
		\midrule
		& \multicolumn{2}{c}{$\mu = 0$} & \multicolumn{2}{c}{$\mu = 1.25$} & \multicolumn{2}{c}{$\mu = 2.5$} & \multicolumn{2}{c}{$\mu = 3.75$} & \multicolumn{2}{c}{$\mu = 5.00$}  \\
		0.001	&	0.125000	&	$	-1.1\times 10^{-7	}$	&	0.469641	&	$	-3.2\times 10^{-9	}$	&	0.789252	&	$	-1.9\times 10^{-12	}$	&	0.932640	&	$	6.9\times 10^{-12	}$	&	0.980056	&	$	7.4\times 10^{-12	}$	\\ 
		0.010	&	0.125009	&	$	-6.7\times 10^{-8	}$	&	0.469639	&	$	-3.2\times 10^{-9	}$	&	0.789246	&	$	-2.0\times 10^{-12	}$	&	0.932637	&	$	6.9\times 10^{-12	}$	&	0.980055	&	$	7.2\times 10^{-12	}$	\\ 
		0.100	&	0.125933	&	$	-1.1\times 10^{-8	}$	&	0.469472	&	$	-3.4\times 10^{-9	}$	&	0.788628	&	$	-1.0\times 10^{-11	}$	&	0.932347	&	$	6.9\times 10^{-12	}$	&	0.979960	&	$	7.3\times 10^{-12	}$	\\ 
		0.250	&	0.130684	&	$	-5.0\times 10^{-9	}$	&	0.468658	&	$	-5.4\times 10^{-8	}$	&	0.785386	&	$	-5.3\times 10^{-10	}$	&	0.930803	&	$	6.9\times 10^{-12	}$	&	0.979449	&	$	7.3\times 10^{-12	}$	\\ 
		0.500	&	0.145933	&	$	-2.8\times 10^{-9	}$	&	0.466537	&	$	-5.2\times 10^{-8	}$	&	0.774340	&	$	-1.3\times 10^{-8	}$	&	0.925154	&	$	-1.9\times 10^{-12	}$	&	0.977535	&	$	7.4\times 10^{-12	}$	\\ 
		1.000	&	0.190069	&	$	-5.3\times 10^{-9	}$	&	0.463836	&	$	-3.2\times 10^{-8	}$	&	0.738415	&	$	-2.2\times 10^{-8	}$	&	0.901605	&	$	-3.3\times 10^{-9	}$	&	0.968469	&	$	-9.9\times 10^{-11	}$	\\ 
		1.500	&	0.234908	&	$	-6.2\times 10^{-9	}$	&	0.464784	&	$	-2.0\times 10^{-8	}$	&	0.699445	&	$	-1.8\times 10^{-8	}$	&	0.864833	&	$	-7.9\times 10^{-9	}$	&	0.949646	&	$	-1.7\times 10^{-9	}$	\\ 
		2.500	&	0.303387	&	$	-5.1\times 10^{-9	}$	&	0.470603	&	$	-9.5\times 10^{-9	}$	&	0.641524	&	$	-1.1\times 10^{-8	}$	&	0.786117	&	$	-8.6\times 10^{-9	}$	&	0.888054	&	$	-5.1\times 10^{-9	}$	\\ 
		\midrule
		&	\multicolumn{2}{c}{$\mu	=		0.25$}		&	\multicolumn{2}{c}{$\mu	=		1.5$}		&	\multicolumn{2}{c}{$\mu	=		2.75$}		&	\multicolumn{2}{c}{$\mu	=		4.0$}		&	\multicolumn{2}{c}{$\mu		=		5.25$}	\\
		0.001	&	0.177672	&	$	1.3\times 10^{-8	}$	&	0.546490	&	$	2.7\times 10^{-08	}$	&	0.830354	&	$	5.8\times 10^{-12	}$	&	0.947006	&	$	7.1\times 10^{-12	}$	&	0.984421	&	$	7.5\times 10^{-12	}$	\\ 
		0.010	&	0.177680	&	$	1.2\times 10^{-8	}$	&	0.546486	&	$	-2.0\times 10^{-7	}$	&	0.830349	&	$	5.7\times 10^{-12	}$	&	0.947004	&	$	7.0\times 10^{-12	}$	&	0.984421	&	$	7.3\times 10^{-12	}$	\\ 
		0.100	&	0.178543	&	$	-6.5\times 10^{-9	}$	&	0.546088	&	$	-2.2\times 10^{-7	}$	&	0.829786	&	$	5.3\times 10^{-12	}$	&	0.946769	&	$	7.0\times 10^{-12	}$	&	0.984346	&	$	7.3\times 10^{-12	}$	\\ 
		0.250	&	0.182969	&	$	-6.9\times 10^{-9	}$	&	0.544081	&	$	-1.2\times 10^{-7	}$	&	0.826822	&	$	-2.6\times 10^{-11	}$	&	0.945511	&	$	7.0\times 10^{-12	}$	&	0.983943	&	$	7.3\times 10^{-12	}$	\\ 
		0.500	&	0.197026	&	$	-4.9\times 10^{-9	}$	&	0.538053	&	$	-6.5\times 10^{-8	}$	&	0.816526	&	$	-4.9\times 10^{-9	}$	&	0.940878	&	$	6.2\times 10^{-12	}$	&	0.982429	&	$	7.4\times 10^{-12	}$	\\ 
		1.000	&	0.236755	&	$	-1.1\times 10^{-8	}$	&	0.524085	&	$	-3.3\times 10^{-8	}$	&	0.780988	&	$	-1.7\times 10^{-8	}$	&	0.920884	&	$	-1.9\times 10^{-9	}$	&	0.975159	&	$	-3.6\times 10^{-11	}$	\\ 
		1.500	&	0.276117	&	$	-9.7\times 10^{-9	}$	&	0.514345	&	$	-2.1\times 10^{-8	}$	&	0.739409	&	$	-1.6\times 10^{-8	}$	&	0.887551	&	$	5.3\times 10^{-8	}$	&	0.959401	&	$	-1.1\times 10^{-9	}$	\\ 
		2.500	&	0.334990	&	$	-6.1\times 10^{-9	}$	&	0.505548	&	$	-1.0\times 10^{-8	}$	&	0.673335	&	$	-1.1\times 10^{-8	}$	&	0.810109	&	$	-7.9\times 10^{-9	}$	&	0.903176	&	$	-4.5\times 10^{-9	}$	\\  
		\midrule
		&	\multicolumn{2}{c}{$\mu	=		0.5$}		&	\multicolumn{2}{c}{$\mu	=		1.75$}		&	\multicolumn{2}{c}{$\mu	=		3.0$}		&	\multicolumn{2}{c}{$\mu	=		4.25$}		&	\multicolumn{2}{c}{$\mu		=		5.5$}	\\
		0.001	&	0.241175	&	$	2.0\times 10^{-9	}$	&	0.618367	&	$	-5.7\times 10^{-8	}$	&	0.864363	&	$	6.4\times 10^{-12	}$	&	0.958400	&	$	7.3\times 10^{-12	}$	&	0.987839	&	$	7.5\times 10^{-12	}$	\\ 
		0.010	&	0.241182	&	$	2.0\times 10^{-9	}$	&	0.618362	&	$	-5.8\times 10^{-8	}$	&	0.864358	&	$	6.4\times 10^{-12	}$	&	0.958398	&	$	7.1\times 10^{-12	}$	&	0.987839	&	$	7.3\times 10^{-12	}$	\\ 
		0.100	&	0.241869	&	$	-3.4\times 10^{-10	}$	&	0.617810	&	$	-1.0\times 10^{-7	}$	&	0.863865	&	$	6.4\times 10^{-12	}$	&	0.958209	&	$	7.1\times 10^{-12	}$	&	0.987780	&	$	7.3\times 10^{-12	}$	\\ 
		0.250	&	0.245393	&	$	-2.8\times 10^{-9	}$	&	0.614985	&	$	-1.0\times 10^{-7	}$	&	0.861255	&	$	4.7\times 10^{-12	}$	&	0.957195	&	$	7.1\times 10^{-12	}$	&	0.987463	&	$	7.3\times 10^{-12	}$	\\ 
		0.500	&	0.256636	&	$	-8.9\times 10^{-9	}$	&	0.606103	&	$	-6.3\times 10^{-8	}$	&	0.852035	&	$	-1.4\times 10^{-9	}$	&	0.953437	&	$	7.2\times 10^{-12	}$	&	0.986270	&	$	7.4\times 10^{-12	}$	\\ 
		1.000	&	0.288603	&	$	-1.7\times 10^{-8	}$	&	0.582812	&	$	-3.3\times 10^{-8	}$	&	0.818423	&	$	-1.2\times 10^{-8	}$	&	0.936736	&	$	-1.0\times 10^{-9	}$	&	0.980474	&	$	-9.3\times 10^{-12	}$	\\ 
		1.500	&	0.320344	&	$	-1.3\times 10^{-8	}$	&	0.563295	&	$	-2.1\times 10^{-8	}$	&	0.776046	&	$	-1.4\times 10^{-8	}$	&	0.907121	&	$	-4.7\times 10^{-9	}$	&	0.967424	&	$	-7.5\times 10^{-10	}$	\\ 
		2.500	&	0.367749	&	$	-7.2\times 10^{-9	}$	&	0.540375	&	$	-1.1\times 10^{-8	}$	&	0.703849	&	$	-1.0\times 10^{-8	}$	&	0.832295	&	$	-7.2\times 10^{-9	}$	&	0.916684	&	$	-3.8\times 10^{-9	}$	\\ 
		\midrule
		&	\multicolumn{2}{c}{$\mu	=		0.75$}		&	\multicolumn{2}{c}{$\mu	=		2.0$}		&	\multicolumn{2}{c}{$\mu	=		3.25$}		&	\multicolumn{2}{c}{$\mu	=		4.5$}		&	\multicolumn{2}{c}{$\mu		=		5.75$}	\\
		0.001	&	0.313294	&	$	-4.0\times 10^{-9	}$	&	0.683325	&	$	-3.0\times 10^{-9	}$	&	0.892147	&	$	6.6\times 10^{-12	}$	&	0.967400	&	$	7.3\times 10^{-12	}$	&	0.990512	&	$	7.5\times 10^{-12	}$	\\ 
		0.010	&	0.313298	&	$	-3.9\times 10^{-9	}$	&	0.683319	&	$	-3.0\times 10^{-9	}$	&	0.892143	&	$	6.6\times 10^{-12	}$	&	0.967399	&	$	7.2\times 10^{-12	}$	&	0.990512	&	$	7.3\times 10^{-12	}$	\\ 
		0.100	&	0.313720	&	$	-1.2\times 10^{-9	}$	&	0.682689	&	$	-6.0\times 10^{-9	}$	&	0.891722	&	$	6.6\times 10^{-12	}$	&	0.967247	&	$	7.2\times 10^{-12	}$	&	0.990466	&	$	7.3\times 10^{-12	}$	\\ 
		0.250	&	0.315905	&	$	-1.1\times 10^{-9	}$	&	0.679429	&	$	-3.8\times 10^{-8	}$	&	0.889487	&	$	6.5\times 10^{-12	}$	&	0.966436	&	$	7.2\times 10^{-12	}$	&	0.990217	&	$	7.3\times 10^{-12	}$	\\ 
		0.500	&	0.323076	&	$	-1.8\times 10^{-8	}$	&	0.668848	&	$	-4.8\times 10^{-8	}$	&	0.881482	&	$	-3.3\times 10^{-10	}$	&	0.963416	&	$	7.3\times 10^{-12	}$	&	0.989280	&	$	7.4\times 10^{-12	}$	\\ 
		1.000	&	0.344614	&	$	-2.3\times 10^{-8	}$	&	0.638745	&	$	-3.1\times 10^{-8	}$	&	0.850805	&	$	-8.5\times 10^{-9	}$	&	0.949650	&	$	-5.0\times 10^{-10	}$	&	0.984681	&	$	1.4\times 10^{-12	}$	\\ 
		1.500	&	0.367003	&	$	-1.6\times 10^{-8	}$	&	0.610889	&	$	-2.1\times 10^{-8	}$	&	0.809179	&	$	-1.2\times 10^{-8	}$	&	0.923801	&	$	-3.4\times 10^{-9	}$	&	0.973974	&	$	-4.8\times 10^{-10	}$	\\ 
		2.500	&	0.401438	&	$	-8.1\times 10^{-9	}$	&	0.574815	&	$	-1.1\times 10^{-8	}$	&	0.732901	&	$	-9.8\times 10^{-9	}$	&	0.852665	&	$	-3.8\times 10^{-9	}$	&	0.928670	&	$	-3.3\times 10^{-9	}$	\\ 
		\midrule
		&	\multicolumn{2}{c}{$\mu	=		1.0$}		&	\multicolumn{2}{c}{$\mu	=		2.25$}		&	\multicolumn{2}{c}{$\mu	=		3.5$}		&	\multicolumn{2}{c}{$\mu	=		4.75$}		&	\multicolumn{2}{c}{$\mu		=		6.0$}	\\
		0.001	&	0.390712	&	$	3.1\times 10^{-9	}$	&	0.740359	&	$	-1.5\times 10^{-10	}$	&	0.914616	&	$	6.8\times 10^{-12	}$	&	0.974488	&	$	7.4\times 10^{-12	}$	&	0.992601	&	$	7.5\times 10^{-12	}$	\\ 
		0.010	&	0.390713	&	$	3.1\times 10^{-9	}$	&	0.740353	&	$	-1.5\times 10^{-10	}$	&	0.914612	&	$	6.8\times 10^{-12	}$	&	0.974486	&	$	7.2\times 10^{-12	}$	&	0.992600	&	$	7.3\times 10^{-12	}$	\\ 
		0.100	&	0.390832	&	$	1.1\times 10^{-9	}$	&	0.739706	&	$	-3.0\times 10^{-10	}$	&	0.914261	&	$	6.8\times 10^{-12	}$	&	0.974366	&	$	7.2\times 10^{-12	}$	&	0.992564	&	$	7.4\times 10^{-12	}$	\\ 
		0.250	&	0.391483	&	$	-9.4\times 10^{-9	}$	&	0.736336	&	$	-6.2\times 10^{-9	}$	&	0.912387	&	$	6.7\times 10^{-12	}$	&	0.973722	&	$	7.2\times 10^{-12	}$	&	0.992369	&	$	7.4\times 10^{-12	}$	\\ 
		0.500	&	0.393962	&	$	-3.4\times 10^{-8	}$	&	0.725109	&	$	-2.8\times 10^{-8	}$	&	0.905600	&	$	-5.5\times 10^{-11	}$	&	0.971310	&	$	7.4\times 10^{-12	}$	&	0.991634	&	$	7.4\times 10^{-12	}$	\\ 
		1.000	&	0.403509	&	$	-2.8\times 10^{-8	}$	&	0.690856	&	$	-2.7\times 10^{-8	}$	&	0.878400	&	$	-5.6\times 10^{-9	}$	&	0.960088	&	$	-2.3\times 10^{-10	}$	&	0.988000	&	$	5.5\times 10^{-12	}$	\\ 
		1.500	&	0.415404	&	$	-1.8\times 10^{-8	}$	&	0.656459	&	$	-2.0\times 10^{-8	}$	&	0.838757	&	$	-9.9\times 10^{-9	}$	&	0.937877	&	$	-2.4\times 10^{-9	}$	&	0.979287	&	$	-3.0\times 10^{-10	}$	\\ 
		2.500	&	0.435809	&	$	-8.9\times 10^{-9	}$	&	0.608611	&	$	-1.1\times 10^{-8	}$	&	0.760358	&	$	-9.2\times 10^{-9	}$	&	0.871238	&	$	-5.8\times 10^{-9	}$	&	0.939237	&	$	-2.8\times 10^{-9	}$	\\
		\bottomrule
	\end{tabular*}
\end{table}
\end{center}
\newpage	

\subsubsection{Fourth moment}

\begin{center}
\begin{table}[h!]\caption{Estimates of $E(P^4)$ obtained from piecewise approximation of $p(x)$ with error compared to estimate of $E(P^4)$ obtained from numerical integration.}\label{tab.accEP4}
	\scriptsize 
	\begin{tabular*}{\textwidth}{@{\extracolsep\fill}lcr cr cr cr cr@{\extracolsep\fill}} \toprule
		$\sigma$	&	$E(P^4)$	&		Error 			&	$E(P^4)$	&		Error 	&	$E(P^4)$	&		Error 		&	$E(P^4)$	&		Error 		&	$E(P^4)$	&		Error 		\\ 
		\midrule
		& \multicolumn{2}{c}{$\mu = 0$} & \multicolumn{2}{c}{$\mu = 1.25$} & \multicolumn{2}{c}{$\mu = 2.5$} & \multicolumn{2}{c}{$\mu = 3.75$} & \multicolumn{2}{c}{$\mu = 5.00$}  \\
		0.001	&	0.062456	&	$	-4.4\times 10^{-5	}$	&	0.365052	&	$	-1.5\times 10^{-8	}$	&	0.729381	&	$	-3.3\times 10^{-11	}$	&	0.911210	&	$	6.8\times 10^{-12	}$	&	0.973496	&	$	7.4\times 10^{-12	}$	\\ 
		0.010	&	0.062476	&	$	-3.3\times 10^{-5	}$	&	0.365054	&	$	-1.5\times 10^{-8	}$	&	0.729374	&	$	-3.4\times 10^{-11	}$	&	0.911206	&	$	6.7\times 10^{-12	}$	&	0.973495	&	$	7.2\times 10^{-12	}$	\\ 
		0.100	&	0.063427	&	$	-6.8\times 10^{-6	}$	&	0.365237	&	$	-3.2\times 10^{-8	}$	&	0.728696	&	$	-7.3\times 10^{-11	}$	&	0.910839	&	$	6.7\times 10^{-12	}$	&	0.973370	&	$	7.2\times 10^{-12	}$	\\ 
		0.250	&	0.068223	&	$	-2.8\times 10^{-6	}$	&	0.366230	&	$	-2.9\times 10^{-7	}$	&	0.725165	&	$	-2.7\times 10^{-9	}$	&	0.908882	&	$	6.7\times 10^{-12	}$	&	0.972699	&	$	7.2\times 10^{-12	}$	\\ 
		0.500	&	0.083944	&	$	-1.4\times 10^{-6	}$	&	0.369890	&	$	-3.3\times 10^{-7	}$	&	0.713441	&	$	-6.7\times 10^{-8	}$	&	0.901792	&	$	-3.8\times 10^{-11	}$	&	0.970190	&	$	7.4\times 10^{-12	}$	\\ 
		1.000	&	0.131594	&	$	-7.6\times 10^{-7	}$	&	0.383118	&	$	-4.8\times 10^{-7	}$	&	0.678263	&	$	-1.4\times 10^{-7	}$	&	0.873391	&	$	-1.8\times 10^{-8	}$	&	0.958497	&	$	-5.3\times 10^{-10	}$	\\ 
		1.500	&	0.182192	&	$	-5.5\times 10^{-7	}$	&	0.398645	&	$	-4.4\times 10^{-7	}$	&	0.644062	&	$	-2.1\times 10^{-7	}$	&	0.832344	&	$	-6.1\times 10^{-8	}$	&	0.935361	&	$	-1.0\times 10^{-8	}$	\\ 
		2.500	&	0.262113	&	$	-3.5\times 10^{-7	}$	&	0.424049	&	$	-3.3\times 10^{-7	}$	&	0.598329	&	$	-2.4\times 10^{-7	}$	&	0.752938	&	$	-1.4\times 10^{-7	}$	&	0.866777	&	$	-6.5\times 10^{-8	}$	\\
		\midrule
		&	\multicolumn{2}{c}{$\mu	=		0.25$}		&	\multicolumn{2}{c}{$\mu	=		1.5$}		&	\multicolumn{2}{c}{$\mu	=		2.75$}		&	\multicolumn{2}{c}{$\mu	=		4.0$}		&	\multicolumn{2}{c}{$\mu		=		5.25$}	\\
		0.001	&	0.099883	&	$	-1.5\times 10^{-8	}$	&	0.446796	&	$	-2.4\times 10^{-8	}$	&	0.780461	&	$	3.8\times 10^{-12	}$	&	0.929973	&	$	7.0\times 10^{-12	}$	&	0.979282	&	$	7.4\times 10^{-12	}$	\\ 
		0.010	&	0.099893	&	$	-1.9\times 10^{-8	}$	&	0.446794	&	$	-1.1\times 10^{-6	}$	&	0.780455	&	$	3.8\times 10^{-12	}$	&	0.929970	&	$	6.9\times 10^{-12	}$	&	0.979281	&	$	7.2\times 10^{-12	}$	\\ 
		0.100	&	0.100917	&	$	-3.4\times 10^{-7	}$	&	0.446654	&	$	-1.2\times 10^{-6	}$	&	0.779806	&	$	1.8\times 10^{-12	}$	&	0.929668	&	$	6.9\times 10^{-12	}$	&	0.979183	&	$	7.3\times 10^{-12	}$	\\ 
		0.250	&	0.106173	&	$	-1.7\times 10^{-6	}$	&	0.446000	&	$	-6.2\times 10^{-7	}$	&	0.776407	&	$	-1.5\times 10^{-10	}$	&	0.928057	&	$	6.9\times 10^{-12	}$	&	0.978652	&	$	7.3\times 10^{-12	}$	\\ 
		0.500	&	0.122943	&	$	-1.2\times 10^{-6	}$	&	0.444530	&	$	-3.5\times 10^{-7	}$	&	0.764844	&	$	-2.5\times 10^{-8	}$	&	0.922162	&	$	1.7\times 10^{-12	}$	&	0.976660	&	$	7.4\times 10^{-12	}$	\\ 
		1.000	&	0.170878	&	$	-7.5\times 10^{-7	}$	&	0.444461	&	$	-4.0\times 10^{-7	}$	&	0.727644	&	$	-1.0\times 10^{-7	}$	&	0.897574	&	$	-9.7\times 10^{-9	}$	&	0.967219	&	$	-2.2\times 10^{-10	}$	\\ 
		1.500	&	0.219043	&	$	-5.5\times 10^{-7	}$	&	0.448498	&	$	-4.0\times 10^{-7	}$	&	0.688289	&	$	-1.7\times 10^{-7	}$	&	0.859412	&	$	-4.5\times 10^{-8	}$	&	0.947610	&	$	-6.9\times 10^{-9	}$	\\ 
		2.500	&	0.292034	&	$	-3.5\times 10^{-7	}$	&	0.458973	&	$	-3.1\times 10^{-7	}$	&	0.631745	&	$	-2.2\times 10^{-7	}$	&	0.779319	&	$	-1.2\times 10^{-7	}$	&	0.884109	&	$	-5.5\times 10^{-8	}$	\\  
		\midrule
		&	\multicolumn{2}{c}{$\mu	=		0.5$}		&	\multicolumn{2}{c}{$\mu	=		1.75$}		&	\multicolumn{2}{c}{$\mu	=		3.0$}		&	\multicolumn{2}{c}{$\mu	=		4.25$}		&	\multicolumn{2}{c}{$\mu		=		5.5$}	\\
		0.001	&	0.150122	&	$	-3.3\times 10^{-8	}$	&	0.526820	&	$	-2.8\times 10^{-7	}$	&	0.823370	&	$	6.0\times 10^{-12	}$	&	0.944921	&	$	7.2\times 10^{-12	}$	&	0.983819	&	$	7.4\times 10^{-12	}$	\\ 
		0.010	&	0.150132	&	$	-3.3\times 10^{-8	}$	&	0.526816	&	$	-2.9\times 10^{-7	}$	&	0.823364	&	$	6.0\times 10^{-12	}$	&	0.944919	&	$	7.0\times 10^{-12	}$	&	0.983818	&	$	7.3\times 10^{-12	}$	\\ 
		0.100	&	0.151122	&	$	-8.1\times 10^{-9	}$	&	0.526418	&	$	-5.1\times 10^{-7	}$	&	0.822775	&	$	5.9\times 10^{-12	}$	&	0.944674	&	$	7.0\times 10^{-12	}$	&	0.983740	&	$	7.3\times 10^{-12	}$	\\ 
		0.250	&	0.156186	&	$	-3.8\times 10^{-7	}$	&	0.524421	&	$	-5.3\times 10^{-7	}$	&	0.819669	&	$	-2.4\times 10^{-12	}$	&	0.943363	&	$	7.0\times 10^{-12	}$	&	0.983322	&	$	7.3\times 10^{-12	}$	\\ 
		0.500	&	0.172137	&	$	-8.7\times 10^{-7	}$	&	0.518570	&	$	-3.2\times 10^{-7	}$	&	0.808886	&	$	-7.2\times 10^{-9	}$	&	0.938532	&	$	6.7\times 10^{-12	}$	&	0.981748	&	$	7.4\times 10^{-12	}$	\\ 
		1.000	&	0.216440	&	$	-7.1\times 10^{-7	}$	&	0.506114	&	$	-3.2\times 10^{-7	}$	&	0.771920	&	$	-7.1\times 10^{-8	}$	&	0.917665	&	$	-5.2\times 10^{-9	}$	&	0.974180	&	$	-7.9\times 10^{-11	}$	\\ 
		1.500	&	0.259610	&	$	-5.4\times 10^{-7	}$	&	0.498773	&	$	-3.5\times 10^{-7	}$	&	0.729508	&	$	-1.4\times 10^{-7	}$	&	0.883015	&	$	4.6\times 10^{-8	}$	&	0.957764	&	$	-4.4\times 10^{-9	}$	\\ 
		2.500	&	0.323398	&	$	-3.5\times 10^{-7	}$	&	0.494138	&	$	-3.0\times 10^{-7	}$	&	0.664104	&	$	-2.0\times 10^{-7	}$	&	0.803925	&	$	-1.1\times 10^{-7	}$	&	0.899710	&	$	-4.5\times 10^{-8	}$	\\ 
		\midrule
		&	\multicolumn{2}{c}{$\mu	=		0.75$}		&	\multicolumn{2}{c}{$\mu	=		2.0$}		&	\multicolumn{2}{c}{$\mu	=		3.25$}		&	\multicolumn{2}{c}{$\mu	=		4.5$}		&	\multicolumn{2}{c}{$\mu		=		5.75$}	\\
		0.001	&	0.212783	&	$	-3.3\times 10^{-9	}$	&	0.601871	&	$	-1.5\times 10^{-8	}$	&	0.858846	&	$	6.4\times 10^{-12	}$	&	0.956771	&	$	7.2\times 10^{-12	}$	&	0.987370	&	$	7.5\times 10^{-12	}$	\\ 
		0.010	&	0.212791	&	$	-3.2\times 10^{-9	}$	&	0.601865	&	$	-1.5\times 10^{-8	}$	&	0.858841	&	$	6.3\times 10^{-12	}$	&	0.956769	&	$	7.1\times 10^{-12	}$	&	0.987369	&	$	7.3\times 10^{-12	}$	\\ 
		0.100	&	0.213603	&	$	1.9\times 10^{-10	}$	&	0.601295	&	$	-3.0\times 10^{-8	}$	&	0.858325	&	$	6.3\times 10^{-12	}$	&	0.956572	&	$	7.1\times 10^{-12	}$	&	0.987308	&	$	7.3\times 10^{-12	}$	\\ 
		0.250	&	0.217764	&	$	-3.6\times 10^{-8	}$	&	0.598380	&	$	-1.9\times 10^{-7	}$	&	0.855592	&	$	5.9\times 10^{-12	}$	&	0.955516	&	$	7.1\times 10^{-12	}$	&	0.986979	&	$	7.3\times 10^{-12	}$	\\ 
		0.500	&	0.230937	&	$	-5.4\times 10^{-7	}$	&	0.589322	&	$	-2.4\times 10^{-7	}$	&	0.845942	&	$	-1.7\times 10^{-9	}$	&	0.951601	&	$	7.3\times 10^{-12	}$	&	0.985739	&	$	7.4\times 10^{-12	}$	\\ 
		1.000	&	0.267701	&	$	-6.5\times 10^{-7	}$	&	0.566519	&	$	-2.5\times 10^{-7	}$	&	0.810896	&	$	-4.7\times 10^{-8	}$	&	0.934180	&	$	-2.6\times 10^{-9	}$	&	0.979708	&	$	-2.4\times 10^{-11	}$	\\ 
		1.500	&	0.303466	&	$	-5.2\times 10^{-7	}$	&	0.548641	&	$	-3.0\times 10^{-7	}$	&	0.767373	&	$	-1.1\times 10^{-7	}$	&	0.903362	&	$	-2.3\times 10^{-8	}$	&	0.966114	&	$	-2.8\times 10^{-9	}$	\\ 
		2.500	&	0.356008	&	$	-3.5\times 10^{-7	}$	&	0.529268	&	$	-2.8\times 10^{-7	}$	&	0.695202	&	$	-1.8\times 10^{-7	}$	&	0.826710	&	$	-9.1\times 10^{-8	}$	&	0.913659	&	$	-3.7\times 10^{-8	}$	\\ 
		\midrule
		&	\multicolumn{2}{c}{$\mu	=		1.0$}		&	\multicolumn{2}{c}{$\mu	=		2.25$}		&	\multicolumn{2}{c}{$\mu	=		3.5$}		&	\multicolumn{2}{c}{$\mu	=		4.75$}		&	\multicolumn{2}{c}{$\mu		=		6.0$}	\\
		0.001	&	0.285633	&	$	9.0\times 10^{-9	}$	&	0.669766	&	$	-7.5\times 10^{-10	}$	&	0.887806	&	$	6.6\times 10^{-12	}$	&	0.966129	&	$	7.3\times 10^{-12	}$	&	0.990146	&	$	7.5\times 10^{-12	}$	\\ 
		0.010	&	0.285639	&	$	9.0\times 10^{-9	}$	&	0.669760	&	$	-7.6\times 10^{-10	}$	&	0.887802	&	$	6.6\times 10^{-12	}$	&	0.966127	&	$	7.1\times 10^{-12	}$	&	0.990146	&	$	7.3\times 10^{-12	}$	\\ 
		0.100	&	0.286161	&	$	3.4\times 10^{-9	}$	&	0.669101	&	$	-1.5\times 10^{-9	}$	&	0.887362	&	$	6.6\times 10^{-12	}$	&	0.965970	&	$	7.2\times 10^{-12	}$	&	0.990098	&	$	7.3\times 10^{-12	}$	\\ 
		0.250	&	0.288866	&	$	-5.3\times 10^{-8	}$	&	0.665698	&	$	-3.1\times 10^{-8	}$	&	0.885025	&	$	6.5\times 10^{-12	}$	&	0.965126	&	$	7.2\times 10^{-12	}$	&	0.989840	&	$	7.3\times 10^{-12	}$	\\ 
		0.500	&	0.297692	&	$	-3.6\times 10^{-7	}$	&	0.654724	&	$	-1.4\times 10^{-7	}$	&	0.876653	&	$	-3.0\times 10^{-10	}$	&	0.961981	&	$	7.3\times 10^{-12	}$	&	0.988865	&	$	7.4\times 10^{-12	}$	\\ 
		1.000	&	0.323692	&	$	-5.7\times 10^{-7	}$	&	0.624280	&	$	-1.9\times 10^{-7	}$	&	0.844632	&	$	-2.9\times 10^{-8	}$	&	0.947630	&	$	-1.2\times 10^{-9	}$	&	0.984083	&	$	-3.4\times 10^{-12	}$	\\ 
		1.500	&	0.350041	&	$	-4.8\times 10^{-7	}$	&	0.597310	&	$	-2.5\times 10^{-7	}$	&	0.801677	&	$	-8.1\times 10^{-8	}$	&	0.920714	&	$	-1.6\times 10^{-8	}$	&	0.972932	&	$	-1.7\times 10^{-9	}$	\\ 
		2.500	&	0.389640	&	$	-3.4\times 10^{-7	}$	&	0.564087	&	$	-2.6\times 10^{-7	}$	&	0.724861	&	$	-1.6\times 10^{-7	}$	&	0.847657	&	$	-7.7\times 10^{-8	}$	&	0.926047	&	$	-3.0\times 10^{-8	}$	\\ 
		\bottomrule
	\end{tabular*}
\end{table}	
\end{center}
\newpage

	\subsection{Skewness and Kurtosis}

\subsubsection{Skewness}	

\begin{center}
\begin{table}[h!]\caption{Estimates of Skewness$(P)$ obtained from piecewise (P.W) approximation of $p(x)$ compared to estimate of Skewness$(P)$ obtained from numerical integration (N.I).}\label{tab.accskew}
	\scriptsize
	\begin{tabular*}{\textwidth}{@{\extracolsep\fill}lcr cr cr cr cr@{\extracolsep\fill}} \toprule
		$\sigma$	&	P.W.	&		N.I. 			&	P.W.	&		N.I.  	&	P.W.	&		N.I.  		&	P.W.	&		N.I.  		&	P.W.	&		N.I.  		\\ 
		\midrule
		& \multicolumn{2}{c}{$\mu = 0$} & \multicolumn{2}{c}{$\mu = 1.25$} & \multicolumn{2}{c}{$\mu = 2.5$} & \multicolumn{2}{c}{$\mu = 3.75$} & \multicolumn{2}{c}{$\mu = 5.00$}  \\
		0.001	&	NaN	&	-0.059301	&	-793.892	&	-0.667082	&	-14.2326	&	-17.0125	&	0.000000	&	-600.488	&	0.000000	&	-19977.2	\\										
		0.010	&	0.000000	&	-0.000059	&	-0.815861	&	-0.017308	&	-0.039774	&	-0.042376	&	-0.028628	&	-0.637634	&	-0.028711	&	-24.6356	\\										
		0.100	&	0.000000	&	0.000000	&	-0.165709	&	-0.165241	&	-0.254678	&	-0.254666	&	-0.287411	&	-0.288009	&	-0.297573	&	-0.321755	\\										
		0.250	&	0.000000	&	0.000000	&	-0.398809	&	-0.398371	&	-0.637851	&	-0.637794	&	-0.734077	&	-0.734111	&	-0.765221	&	-0.766599	\\										
		0.500	&	0.000000	&	0.000000	&	-0.700880	&	-0.700826	&	-1.258130	&	-1.257990	&	-1.569970	&	-1.569970	&	-1.693560	&	-1.693680	\\										
		1.000	&	0.000000	&	0.000000	&	-0.918570	&	-0.918565	&	-1.995830	&	-1.995810	&	-3.285780	&	-3.285750	&	-4.575250	&	-4.575230	\\										
		1.500	&	0.000000	&	0.000000	&	-0.867341	&	-0.867340	&	-1.947800	&	-1.947800	&	-3.545410	&	-3.545400	&	-6.125330	&	-6.125310	\\										
		2.500	&	0.000000	&	0.000000	&	-0.664027	&	-0.664027	&	-1.434290	&	-1.434290	&	-2.460820	&	-2.460820	&	-4.005940	&	-4.005940	\\	\midrule									
		&	\multicolumn{2}{c}{$\mu	=		0.25$}		&	\multicolumn{2}{c}{$\mu	=		1.5$}		&	\multicolumn{2}{c}{$\mu	=		2.75$}		&	\multicolumn{2}{c}{$\mu	=		4.0$}		&	\multicolumn{2}{c}{$\mu		=		5.25$}	\\
		0.001	&	830.004	&	-0.088667	&	16962.5	&	-1.212320	&	-1.370750	&	-34.1905	&	0.000000	&	-1246.16	&	-3.171340	&	-36844.1	\\										
		0.010	&	0.792044	&	-0.003818	&	-36.2603	&	-0.020275	&	-0.027774	&	-0.060459	&	-0.028887	&	-1.298900	&	-0.026951	&	-51.8039	\\										
		0.100	&	-0.037189	&	-0.036857	&	-0.232955	&	-0.189570	&	-0.264387	&	-0.264418	&	-0.290521	&	-0.291768	&	-0.298491	&	-0.349454	\\										
		0.250	&	-0.086747	&	-0.086732	&	-0.462361	&	-0.460887	&	-0.665807	&	-0.665802	&	-0.743540	&	-0.743611	&	-0.768076	&	-0.770979	\\										
		0.500	&	-0.144372	&	-0.144371	&	-0.828815	&	-0.828719	&	-1.340510	&	-1.340420	&	-1.605970	&	-1.605980	&	-1.705690	&	-1.705930	\\										
		1.000	&	-0.178212	&	-0.178212	&	-1.117090	&	-1.117080	&	-2.238940	&	-2.238920	&	-3.555900	&	-3.555870	&	-4.797870	&	-4.797860	\\										
		1.500	&	-0.167221	&	-0.167221	&	-1.058540	&	-1.058540	&	-2.213630	&	-2.213620	&	-3.964980	&	-3.965120	&	-6.809510	&	-6.809490	\\										
		2.500	&	-0.129654	&	-0.129654	&	-0.805732	&	-0.805732	&	-1.612900	&	-1.612900	&	-2.716830	&	-2.716830	&	-4.415180	&	-4.415180	\\	\midrule									
		&	\multicolumn{2}{c}{$\mu	=		0.5$}		&	\multicolumn{2}{c}{$\mu	=		1.75$}		&	\multicolumn{2}{c}{$\mu	=		3.0$}		&	\multicolumn{2}{c}{$\mu	=		4.25$}		&	\multicolumn{2}{c}{$\mu		=		5.5$}	\\
		0.001	&	72.334	&	-0.138492	&	-46345.8	&	-2.266410	&	-0.139677	&	-69.3929	&	0.000000	&	-2572.76	&	-3.333680	&	-63848.1	\\										
		0.010	&	0.062054	&	-0.007484	&	-19.0441	&	-0.023401	&	-0.027295	&	-0.096421	&	-0.029146	&	-2.686260	&	-0.036663	&	-108.968	\\										
		0.100	&	-0.072677	&	-0.072649	&	-0.243807	&	-0.210444	&	-0.272221	&	-0.272289	&	-0.292968	&	-0.295576	&	-0.299210	&	-0.406728	\\										
		0.250	&	-0.171576	&	-0.171567	&	-0.517992	&	-0.515918	&	-0.688702	&	-0.688706	&	-0.751028	&	-0.751177	&	-0.770311	&	-0.776431	\\										
		0.500	&	-0.287744	&	-0.287740	&	-0.949656	&	-0.949515	&	-1.412370	&	-1.412320	&	-1.635374	&	-1.635390	&	-1.715290	&	-1.715800	\\										
		1.000	&	-0.357816	&	-0.357814	&	-1.323300	&	-1.323290	&	-2.490940	&	-2.490920	&	-3.823459	&	-3.823420	&	-5.002180	&	-5.002180	\\										
		1.500	&	-0.335975	&	-0.335974	&	-1.259940	&	-1.259940	&	-2.502700	&	-2.502700	&	-4.427979	&	-4.427970	&	-7.559880	&	-7.559780	\\										
		2.500	&	-0.260084	&	-0.260084	&	-0.952517	&	-0.952517	&	-1.802980	&	-1.802980	&	-2.995759	&	-2.995760	&	-4.869750	&	-4.869740	\\	\midrule									
		&	\multicolumn{2}{c}{$\mu	=		0.75$}		&	\multicolumn{2}{c}{$\mu	=		2.0$}		&	\multicolumn{2}{c}{$\mu	=		3.25$}		&	\multicolumn{2}{c}{$\mu	=		4.5$}		&	\multicolumn{2}{c}{$\mu		=		5.75$}	\\
		0.001	&	-350.547	&	-0.225599	&	-1742.91	&	-4.364530	&	-0.014357	&	-141.899	&	-0.173062	&	-5210.67	&	-7.019330	&	-102217	\\										
		0.010	&	-0.354979	&	-0.010974	&	-1.691060	&	-0.027217	&	-0.027780	&	-0.169836	&	-0.029414	&	-5.599510	&	-0.035089	&	-229.040000	\\										
		0.100	&	-0.106529	&	-0.106428	&	-0.231351	&	-0.228045	&	-0.278486	&	-0.278625	&	-0.294889	&	-0.300357	&	-0.299776	&	-0.526737	\\										
		0.250	&	-0.252743	&	-0.252739	&	-0.564751	&	-0.563484	&	-0.707254	&	-0.707262	&	-0.756932	&	-0.757244	&	-0.772057	&	-0.784975	\\										
		0.500	&	-0.429064	&	-0.429053	&	-1.062140	&	-1.061970	&	-1.474100	&	-1.474070	&	-1.659170	&	-1.659200	&	-1.722860	&	-1.723920	\\										
		1.000	&	-0.540187	&	-0.540185	&	-1.538120	&	-1.538110	&	-2.750680	&	-2.750660	&	-4.084890	&	-4.084860	&	-5.186650	&	-5.186660	\\										
		1.500	&	-0.507820	&	-0.507819	&	-1.473680	&	-1.473680	&	-2.818470	&	-2.818480	&	-4.939170	&	-4.939160	&	-8.379630	&	-8.379480	\\										
		2.500	&	-0.392076	&	-0.392075	&	-1.105430	&	-1.105430	&	-2.006300	&	-2.006300	&	-3.300950	&	-3.300950	&	-5.376360	&	-5.376360	\\	\midrule									
		&	\multicolumn{2}{c}{$\mu	=		1.0$}		&	\multicolumn{2}{c}{$\mu	=		2.25$}		&	\multicolumn{2}{c}{$\mu	=		3.5$}		&	\multicolumn{2}{c}{$\mu	=		4.75$}		&	\multicolumn{2}{c}{$\mu		=		6.0$}	\\
		0.001	&	435.085	&	-0.380924	&	-151.884	&	-8.572380	&	0.000000	&	-291.578	&	-0.722125	&	-10357.6	&	14.7975	&	-149558	\\										
		0.010	&	0.415796	&	-0.014243	&	-0.176634	&	-0.032803	&	-0.028238	&	-0.321606	&	-0.029240	&	-11.7310	&	-0.014794	&	-480.461	\\										
		0.100	&	-0.137300	&	-0.137453	&	-0.242961	&	-0.242667	&	-0.283467	&	-0.283755	&	-0.296394	&	-0.307884	&	-0.300220	&	-0.779660	\\										
		0.250	&	-0.328785	&	-0.328731	&	-0.604286	&	-0.603922	&	-0.722166	&	-0.722183	&	-0.761576	&	-0.762231	&	-0.773422	&	-0.800708	\\										
		0.500	&	-0.567190	&	-0.567164	&	-1.165220	&	-1.165050	&	-1.526339	&	-1.526330	&	-1.678300	&	-1.678350	&	-1.728810	&	-1.731050	\\										
		1.000	&	-0.726673	&	-0.726670	&	-1.762190	&	-1.762180	&	-3.016413	&	-3.016380	&	-4.336620	&	-4.336590	&	-5.350560	&	-5.350600	\\										
		1.500	&	-0.684374	&	-0.684373	&	-1.702100	&	-1.702100	&	-3.164690	&	-3.164680	&	-5.503380	&	-5.503360	&	-9.270860	&	-9.271500	\\										
		2.500	&	-0.526441	&	-0.526441	&	-1.265600	&	-1.265600	&	-2.224833	&	-2.224830	&	-3.636210	&	-3.635680	&	-5.942780	&	-5.942770	\\
		\bottomrule
	\end{tabular*}
\end{table}	
\end{center}
	
\newpage
\subsubsection{Kurtosis}

\begin{center}
\begin{table}[h!]\caption{Estimates of Kurtosis$(P)$ obtained from piecewise (P.W.) approximation of $p(x)$ compared to estimate of Kurtosis$(P)$ obtained from numerical integration (N.I.).}\label{tab.acckurt} 
	\scriptsize 
	\begin{tabular*}{\textwidth}{@{\extracolsep\fill}lcr cr cr cr cr@{\extracolsep\fill}} \toprule
		$\sigma$	&	P.W.	&		N.I. 			&	P.W.	&		N.I.  	&	P.W.	&		N.I.  		&	P.W.	&		N.I.  		&	P.W.	&		N.I.  		\\ 
		\midrule
		0.001	&	-616668	&	121.617	&	-420787	&	2991.06	&	-636577	&	224125	&	-3496.87	&	259482	&	0.000000	&	276097	\\																																		
		0.010	&	-856080	&	3.011660	&	-442.290	&	3.301710	&	-61.0697	&	25.3095	&	3.321160	&	2653.74	&	0.000000	&	367302	\\																																		
		0.100	&	0.000000	&	0.000000	&	2.780680	&	3.034590	&	3.096950	&	3.111810	&	3.145340	&	3.404200	&	3.158360	&	39.0431	\\																																		
		0.250	&	0.000000	&	0.000000	&	3.149120	&	3.191330	&	3.674980	&	3.685030	&	3.956290	&	3.962050	&	4.053710	&	4.842450	\\																																		
		0.500	&	0.000000	&	0.000000	&	3.493000	&	3.496450	&	5.572690	&	5.584860	&	7.476750	&	7.477540	&	8.420690	&	8.449350	\\																																		
		1.000	&	0.000000	&	0.000000	&	3.374880	&	3.375370	&	8.339780	&	8.340770	&	21.0096	&	21.0128	&	45.1782	&	45.1852	\\																																		
		1.500	&	0.000000	&	0.000000	&	2.804230	&	2.804360	&	6.926430	&	6.926670	&	19.8409	&	19.8416	&	62.0938	&	62.0965	\\																																		
		2.500	&	0.000000	&	0.000000	&	2.016500	&	2.016530	&	3.925600	&	3.925650	&	8.745550	&	8.745630	&	21.1813	&	21.1815	\\	\midrule																																	
		&	\multicolumn{2}{c}{$\mu	=		0.25$}		&	\multicolumn{2}{c}{$\mu	=		1.5$}		&	\multicolumn{2}{c}{$\mu	=		2.75$}		&	\multicolumn{2}{c}{$\mu	=		4.0$}		&	\multicolumn{2}{c}{$\mu		=		5.25$}	\\																								
		0.001	&	-118555	&	204.640	&	-43547	&	6636.970000	&	-75346.0	&	568364	&	-4562.89	&	684956&	183213	&	624387	\\																																		
		0.010	&	-1246.01	&	3.020000	&	-122042	&	3.669950	&	-4.591260	&	59.6883	&	2.280830	&	7062.33	&	0.000000	&	990585	\\																																		
		0.100	&	2.076390	&	2.982990	&	-8.284760	&	3.051840	&	3.118250	&	3.125320	&	3.148970	&	3.837920	&	3.167470	&	100.084	\\																																		
		0.250	&	2.776590	&	2.899890	&	3.147170	&	3.296550	&	3.757980	&	3.759480	&	3.985380	&	4.000660	&	4.062870	&	6.191260	\\																																		
		0.500	&	2.660590	&	2.667180	&	3.848050	&	3.853160	&	6.015620	&	6.025370	&	7.738920	&	7.739670	&	8.520460	&	8.597360	\\																																		
		1.000	&	2.184700	&	2.185100	&	3.982070	&	3.982600	&	10.088500	&	10.08970	&	24.9067	&	24.9104	&	51.1256	&	51.1332	\\																																		
		1.500	&	1.858520	&	1.858620	&	3.290810	&	3.290950	&	8.477790	&	8.478090	&	24.7707	&	24.7796	&	78.6087	&	78.6125	\\																																		
		2.500	&	1.520380	&	1.520410	&	2.260920	&	2.260950	&	4.574260	&	4.574310	&	10.3673	&	10.3674	&	25.5622	&	25.5624	\\	\midrule																																	
		&	\multicolumn{2}{c}{$\mu	=		0.5$}		&	\multicolumn{2}{c}{$\mu	=		1.75$}		&	\multicolumn{2}{c}{$\mu	=		3.0$}		&	\multicolumn{2}{c}{$\mu	=		4.25$}		&	\multicolumn{2}{c}{$\mu		=		5.5$}	\\																								
		0.001	&	-117410	&	367.811	&	-162063	&	15294.4	&	-9062.18	&	146063	&	12013.9	&	179503	&	0.000000	&	12996	\\																																		
		0.010	&	-1180.64	&	3.036400	&	-49317.8	&	4.543980	&	2.057180	&	149.056	&	3.603180	&	18893.0	&	-164.434	&	267003	\\																																		
		0.100	&	2.966820	&	2.990770	&	-5.592240	&	3.068600	&	3.128020	&	3.142510	&	3.151890	&	4.994570	&	3.130600	&	265.422	\\																																		
		0.250	&	2.908710	&	2.941710	&	3.173400	&	3.402740	&	3.822740	&	3.823240	&	4.008710	&	4.049440	&	4.070870	&	9.821820	\\																																		
		0.500	&	2.768840	&	2.774210	&	4.245400	&	4.253280	&	6.436280	&	6.442750	&	7.960550	&	7.962110	&	8.600340	&	8.807310	\\																																		
		1.000	&	2.323680	&	2.324090	&	4.751680	&	4.752270	&	12.186100	&	12.1877	&	29.3036	&	29.3083	&	57.2219	&	57.2303	\\																																		
		1.500	&	1.968580	&	1.968680	&	3.913120	&	3.913280	&	10.423400	&	10.4238	&	31.0307	&	31.0267	&	99.6597	&	99.6602	\\																																		
		2.500	&	1.579730	&	1.579750	&	2.564490	&	2.564520	&	5.351500	&	5.351550	&	12.3297	&	12.3299	&	30.9772	&	30.9775	\\	\midrule																																	
		&	\multicolumn{2}{c}{$\mu	=		0.75$}		&	\multicolumn{2}{c}{$\mu	=		2.0$}		&	\multicolumn{2}{c}{$\mu	=		3.25$}		&	\multicolumn{2}{c}{$\mu	=		4.5$}		&	\multicolumn{2}{c}{$\mu		=		5.75$}	\\																								
		0.001	&	310570	&	702.869	&	-542399	&	36588.3	&	-1065.43	&	379098	&	-63706.5	&	460157	&	887784	&	243369	\\																																		
		0.010	&	310.1330	&	3.070030	&	-5101.33	&	6.666550	&	2.822720	&	383.583	&	3.184430	&	50685.0	&	443.768	&	719214	\\																																		
		0.100	&	3.017190	&	3.002780	&	2.076390	&	3.084240	&	3.135120	&	3.172380	&	3.153890	&	8.097200	&	3.189370	&	713.368	\\																																		
		0.250	&	3.003800	&	3.007740	&	3.346970	&	3.505030	&	3.876220	&	3.877080	&	4.027320	&	4.136330	&	4.076620	&	19.6437	\\																																		
		0.500	&	2.946620	&	2.950610	&	4.673810	&	4.684610	&	6.823970	&	6.827510	&	8.145180	&	8.149210	&	8.663970	&	9.222270	\\																																		
		1.000	&	2.561430	&	2.561860	&	5.711600	&	5.712290	&	14.6783	&	14.6803	&	34.1823	&	34.1878	&	63.3369	&	63.3474	\\																																		
		1.500	&	2.157020	&	2.157130	&	4.699750	&	4.699930	&	12.8685	&	12.8691	&	38.9961	&	38.9977	&	126.448	&	126.446	\\																																		
		2.500	&	1.680290	&	1.680320	&	2.935620	&	2.935660	&	6.283030	&	6.283090	&	14.7135	&	14.7136	&	37.7024	&	37.7028	\\	\midrule																																	
		&	\multicolumn{2}{c}{$\mu	=		1.0$}		&	\multicolumn{2}{c}{$\mu	=		2.25$}		&	\multicolumn{2}{c}{$\mu	=		3.5$}		&	\multicolumn{2}{c}{$\mu	=		4.75$}		&	\multicolumn{2}{c}{$\mu		=		6.0$}	\\																								
		0.001	&	-210310	&	1414.10	&	-558892	&	89841.7	&	0.000000	&	990527	&	0.000000	&	114993	&	-239975	&	404306	\\																																		
		0.010	&	-17.7853	&	3.141920	&	-556.894	&	11.9408	&	2.980570	&	1003.80	&	-8.489280	&	136335	&	-1199.53	&	193126	\\																																		
		0.100	&	3.018400	&	3.017830	&	2.986990	&	3.098540	&	3.140750	&	3.238600	&	3.155570	&	16.4591	&	3.145170	&	1928.11	\\																																		
		0.250	&	3.087680	&	3.092950	&	3.547050	&	3.599910	&	3.920310	&	3.922500	&	4.042080	&	4.334900	&	4.076220	&	46.2559	\\																																		
		0.500	&	3.189940	&	3.193100	&	5.120740	&	5.133300	&	7.171800	&	7.173440	&	8.297070	&	8.307780	&	8.714510	&	10.2229	\\																																		
		1.000	&	2.907330	&	2.907780	&	6.894940	&	6.895760	&	17.6078&	17.6102	&	39.497900	&	39.5041	&	69.3386	&	69.3548\\																																		
		1.500	&	2.431690	&	2.431800	&	5.688020	&	5.688230	&	15.9491	&	15.9497	&	49.146900	&	49.1490	&	160.446	&	160.543	\\																																		
		2.500	&	1.824640	&	1.824670	&	3.384910	&	3.384940	&	7.400970	&	7.401040	&	17.620600	&	17.6142	&	46.0955	&	46.0959	\\
		\bottomrule
	\end{tabular*}
\end{table}	
\end{center}

\section{Code needed in paper}

\subsection{R function for finding first to fourth moments and variance}

In this section, we present functions for finding $E(P)$ and $E(P), E(P^2), E(P^3), E(P^4), \text{Var}(P)$. Note these are in C++ and users will need the R package \texttt{Rcpp}\citep{Eddelbuettel2011} for compiling the functions, before these can be used in R. The input parameters are $\mu, \sigma$, but there is an additional C++ function included for the evaluation of the standard normal cdf $\Phi(\cdot)$.

\begin{footnotesize}
	\begin{verbatim}
		#include <Rcpp.h>
		#include <cmath>
		#include <cstddef>
		#include <limits>
		#include <vector>
		
		// normal_cdf is a C++ function for calculating the standard normal CDF.
		
		// [[Rcpp::export]]
		static inline double normal_cdf(const double x) {
			// Standard normal CDF using erfc for numerical stability
			return 0.5 * std::erfc(-x / std::sqrt(2.0));
		}
		
		// EP_logitnormalkpartcplus(double mu, double sigma) is a C++ function for 
		// calculating the expectation of a logit normal random variable with 
		// location parameter mu and scale parameter sigma.
		
		// [[Rcpp::export]]
		double EP_logitnormalkpartcplus(double mu, double sigma) {
			// In part one we calculate the component functions and fixed inputs, coef_nearzero, L and k.
			const double sigma2 = sigma * sigma;
			
			double sum_EP_terms = 0.0;
			double sum_EdiffP = 0.0;
			double sum_E2diffP_terms = 0.0;
			double sum_E3diffP = 0.0;
			double L = 1.507306;
			std::size_t k = 11;
			std::vector<double> coef_nearzero = {0,0.999991820, -0.499903021,  0.332646462, -0.246745120,  
				0.189036359, -0.139389342,  0.091239940, -0.048514467,  0.018953113, -0.004725329, 0.000556454};
			
			for (std::size_t i = 1; i <= k; ++i) {
				const double N = static_cast<double>(i);
				
				double p0a = std::exp(N * mu + 0.5 * N * N * sigma2);
				double p0b = std::exp(-N * mu + 0.5 * N * N * sigma2);
				
				// if element is infinity, the cdf to be multiplied
				// is zero. Adding this line stops NaN.
				if (!std::isfinite(p0a)) p0a = 0.0;
				if (!std::isfinite(p0b)) p0b = 0.0;
				
				const double p0c = coef_nearzero[i] * N;	
				const double p1 = ((i % 2 == 1) ? 1.0 : -1.0) - p0c; // (-1)^(N-1) - p0c	
				const double p2 = p0a * normal_cdf((-L - mu - N * sigma2) / sigma);
				const double p3 = p0b * normal_cdf((mu - L - N * sigma2) / sigma);
				const double p4 = p0c * (p0a * normal_cdf((-mu - N * sigma2) / sigma));
				const double p5 = p0c * (p0b * normal_cdf((mu - N * sigma2) / sigma));
				
				// In part two, we use the results of part one to find E(P) and the first three derivatives.
				sum_EP_terms += p1 * (p2 - p3) + p4 - p5;
				
			}
			
			const double EP = normal_cdf(mu / sigma) + sum_EP_terms;
			
			return EP;
		}
		
		// EP1to4_logitnormalkpartcplus(double mu, double sigma) is a C++ function 
		// for calculating the expectation of a logit normal random variable with 
		// location parameter mu and scale parameter sigma, along with the second, 
		// third, fourth moments and the variance. 
		
		// [[Rcpp::export]]
		std::vector<double> EP1to4_logitnormalkpartcplus(double mu, double sigma) {
			// In part one we calculate the component functions.
			const double sigma2 = sigma * sigma;
			
			double sum_EP_terms = 0.0;
			double sum_EdiffP = 0.0;
			double sum_E2diffP_terms = 0.0;
			double sum_E3diffP = 0.0;
			double L = 1.507306;
			std::size_t k = 11;
			std::vector<double> coef_nearzero = {0,0.999991820, -0.499903021,  0.332646462, -0.246745120, 
				0.189036359, -0.139389342,  0.091239940, -0.048514467,  0.018953113, -0.004725329, 0.000556454};
			
			for (std::size_t i = 1; i <= k; ++i) {
				const double N = static_cast<double>(i);
				
				double p0a = std::exp(N * mu + 0.5 * N * N * sigma2);
				double p0b = std::exp(-N * mu + 0.5 * N * N * sigma2);
				
				// if element is infinity, the cdf to be multiplied
				// is zero. Adding this line stops NaN.
				if (!std::isfinite(p0a)) p0a = 0.0;
				if (!std::isfinite(p0b)) p0b = 0.0;
				
				const double p0c = coef_nearzero[i] * N;	
				const double p1 = ((i % 2 == 1) ? 1.0 : -1.0) - p0c; // (-1)^(N-1) - p0c	
				const double p2 = p0a * normal_cdf((-L - mu - N * sigma2) / sigma);
				const double p3 = p0b * normal_cdf((mu - L - N * sigma2) / sigma);
				const double p4 = p0c * (p0a * normal_cdf((-mu - N * sigma2) / sigma));
				const double p5 = p0c * (p0b * normal_cdf((mu - N * sigma2) / sigma));
				
				// In part two, we use the results of part one to find E(P) and the first three derivatives.
				sum_EP_terms += p1 * (p2 - p3) + p4 - p5;
				sum_EdiffP += N * (p1 * (p2 + p3) + p4 + p5);
				sum_E2diffP_terms += (N * N) * (p1 * (p2 - p3) + p4 - p5);
				sum_E3diffP += (N * N * N) * (p1 * (p2 + p3) + p4 + p5);
			}
			
			const double EP = normal_cdf(mu / sigma) + sum_EP_terms;
			const double EdiffP = sum_EdiffP;
			const double E2diffP = sum_E2diffP_terms;
			const double E3diffP = sum_E3diffP;
			
			// In part three, we find higher order moments.
			const double EP2 = EP - EdiffP;
			const double EP3 = EP2 - 0.5 * (EdiffP - E2diffP);
			const double EP4 = EP3 - (1.0 / 3.0) * (EdiffP - E2diffP - 0.5 * (E2diffP - E3diffP));
			
			// In part four, we calculate central moments (Var, skewness and kurtosis)
			const double VarP = EP2 - EP * EP;
			
			std::vector<double> results;
			results.reserve(5);
			results.push_back(EP);
			results.push_back(EP2);
			results.push_back(EP3);
			results.push_back(EP4);
			results.push_back(VarP);
			
			return results;
		}
	\end{verbatim}
\end{footnotesize}	

\subsection{Other code used in this paper}

Additional code for implementing the EP examples to illustrate EP logistic regression is (1) requires the first-third moment of logit-normal random variables can be found at \url{https://github.com/johnbholmes/Logit-normal-moments-with-EP-applications}. This contains R code for implementing expectation propagation, and a Gibbs sampler for MCMC based comparison.

\end{document}